%% file: CAL-ECC_arxiv.tex
\documentclass{article} 
\usepackage{iclr2027_conference,times}

\input{math_commands.tex}

\PassOptionsToPackage{capitalize, noabbrev}{cleveref}

\usepackage[utf8]{inputenc}
\usepackage[T1]{fontenc}
\usepackage{hyperref}
\usepackage{url}
\usepackage{booktabs}
\usepackage{amsfonts,amssymb,amsmath,amsthm}
\usepackage{mathtools}
\usepackage{nicefrac}
\usepackage{microtype}
\usepackage{algorithm}
\usepackage{xcolor}
\usepackage{graphicx}
\usepackage{cleveref}
\usepackage{enumitem}
\usepackage{subcaption}
\usepackage{siunitx}
\usepackage{algpseudocode}
\usepackage{bm}
\usepackage{bbm}
\usepackage{multirow}
\usepackage{kotex}
\usepackage{rotating}
\usepackage{graphicx}

\newtheorem{theorem}{Theorem}[section]
\newtheorem{lemma}[theorem]{Lemma}

\newtheorem{proposition}[theorem]{Proposition}
\theoremstyle{definition}

\crefname{appendix}{Appendix}{Appendices}
\Crefname{appendix}{Appendix}{Appendices}
\crefname{proposition}{Proposition}{Propositions}
\Crefname{proposition}{Proposition}{Propositions}
\crefname{remark}{Remark}{Remarks}
\Crefname{remark}{Remark}{Remarks}
\crefname{corollary}{Corollary}{Corollaries}
\Crefname{corollary}{Corollary}{Corollaries}
\crefname{definition}{Definition}{Definitions}
\Crefname{definition}{Definition}{Definitions}
\crefname{assumption}{Assumption}{Assumptions}
\Crefname{assumption}{Assumption}{Assumptions}
\crefname{lemma}{Lemma}{Lemmas}
\Crefname{lemma}{Lemma}{Lemmas}

\newcommand{\Pinc}{\Phi_{\mathrm{inc}}}
\newcommand{\einc}{\hat{\mathbf{e}}_{\mathrm{inc}}}
\newcommand{\Finc}{F_{\mathrm{inc}}}

\usepackage{hyperref}
\usepackage{url}

\title{Certified Adaptive List Decoding for \\Error-Correcting Codes}

\author{Jaeho Hwang$^{1}$ \qquad Byung-Sun Won$^{2,*}$ \qquad Uiwon Hwang$^{3,}$\thanks{Corresponding authors (\{bswon, uiwon.hwang\}@ewha.ac.kr)}\\
	$^1$Human-Centered Artificial Intelligence Research Institute\\
	$^2$Institute of Mathematical Sciences\\
	$^3$Department of Computer Science and Engineering\\
	Ewha Womans University\\
}

\iclrfinalcopy 

\begin{document}

\maketitle

\begin{abstract}
Maximum-likelihood (ML) decoding of linear block codes is NP-hard, so neural
decoders are evaluated by how closely they approach it. We observe that proving
a candidate optimal can be far easier than finding it, and build a decoder that
does both. We cast decoding as a search over an absorbing-state discrete
diffusion model, whose sequential commitment of coordinates turns generation
into a tree, and run a probability beam search over it. A dual bound from a
linear-programming relaxation and a distance-based test, both derived from the
code, then certify per frame that the current candidate is the global ML
solution; because both are theorems, certified frames incur no
accuracy-computation trade-off. Across $15$ codes and five operating points we
certify up to $99.9\%$ of frames, the two tests dividing the space as their
derivations predict, one strong for sparse parity-check matrices and the other
for codes with a large minimum distance. Certification also behaves as a
per-frame reliability signal, with $96.7$-$100\%$ of all frame errors falling
on uncertified frames. Against an exact ML reference at $5$~dB the decoder is
within one standard deviation on three of the four codes where the reference is
uncensored, and against four neural baselines it attains the best bit-error rate
in $71$ of $75$ cells.
\end{abstract}

\section{Introduction}
\label{sec:intro}

Reliable digital communication rests on error-correcting codes, which let a
receiver recover a transmitted message from a noisy observation. The optimal
rule is maximum-likelihood (ML) decoding, but it is NP-hard for general linear
block codes~\citep{berlekamp1978}, so practical decoding is a search for good
approximations. Classical iterative decoders, most notably belief
propagation~\citep{gallager1962,richardson2001}, perform approximate
marginalization on the code's Tanner graph. Deep learning has since produced
\emph{model-based} decoders that unroll message passing into a trainable
network~\citep{nachmani2016,lugosch2017,nachmani2018,buchberger2021} and
\emph{model-free} decoders that learn the decoding map from reliability and
syndrome features~\citep{bennatan2018}. Model-free decoders currently report the strongest benchmark results, with masked-attention transformers~\citep{choukroun2022ecct,park2024crossmpt} and with generative variants that recast decoding as iterative denoising~\citep{choukroun2022ddecc,sbecc2026,lei2025}.

Each of these decoders tries to find the ML solution, but none of them can tell whether its output is that solution. Finding $\arg\min_{\mathbf{e}} \sum_j e_j |\ell_j|$ subject to the syndrome constraint is hard.
\emph{Verifying} that a particular $\mathbf{e}$ attains that minimum, for a particular
received word, can be much easier. Classical coding theory has long exploited this asymmetry, in LP decoding~\citep{feldman2005} and in ordered-statistics decoding~\citep{fossorier1995,kaneko1994}. But such certificates have always been tied to a decoder whose structure makes the bound computable: a convex relaxation solved to optimality, or a search in a fixed, prescribed order. A learned decoder is neither.

We propose CAL-ECC (Certified Adaptive List decoding for Error-Correcting Codes), which uses an absorbing-state discrete diffusion model over
error vectors~\citep{austin2021,shi2024} as its candidate generator. Because absorbing-state diffusion commits coordinates \emph{one at a time}, generation becomes a search tree whose width adapt to
the difficulty of the received word, and every node holds a partial
solution that can be tested for optimality. The diffusion model supplies the branching rule, and two bounds derived from the code, a dual bound from an LP relaxation and a distance-based test, decide when the search may stop. Both are theorems, so a frame passing either has provably attained the global ML solution.

Certification is decided frame by frame, so the decoder returns each estimate with a statement about it. That statement also predicts correctness, because nearly all frame errors occur on uncertified frames. On the codes with an exact ML reference we do not
approach the ceiling but reach it, on BCH$(63,51)$ at $6$~dB with an average of
$0.6$ function evaluations per frame.

\paragraph{Contributions.}
\begin{enumerate}[leftmargin=1.6em,itemsep=1pt,topsep=2pt]
	\item We introduce CAL-ECC, a neural decoder that combines adaptive list search with per-frame ML certification. Absorbing-state discrete diffusion constructs a search tree over error vectors, and probability beam search adapts its width to the difficulty of each received frame.
	\item We derive two complementary optimality certificates, an LP dual bound and a distance-based test. We prove that satisfying either condition establishes global ML optimality, enabling certified early termination rather than stopping at the first syndrome-consistent candidate.
	\item We show that certification reflects code structure, with the dual bound effective for sparse parity-check matrices and the distance-based test for codes with a large known minimum distance. Certification rates reach 99.9\%, and 96.7--100\% of frame errors fall on uncertified frames in the reliability analysis.
	\item Across 15 codes and five operating points, CAL-ECC achieves the lowest bit-error rate in 71 of 75 settings against four neural baselines. At 5\,dB, its frame-error rate is within one standard deviation of the exact ML reference on three of four evaluated codes with uncensored references and within two on the fourth.
\end{enumerate}

\section{Related Work}
\label{sec:related}

\paragraph{Neural decoders.}
Model-based decoders unroll belief propagation or min-sum into a trainable
network, preserving the Tanner-graph
computation~\citep{nachmani2016,cammerer2017,lugosch2017,nachmani2018}, and
remain tied to the connectivity of the underlying procedure. Model-free
decoders learn the map from channel observation to bits with general
architectures, made practical at moderate block length by the
reliability--syndrome preprocessing of~\citet{bennatan2018}; ECCT~\citep{choukroun2022ecct} built a transformer on that representation with self-attention masked by the parity-check matrix. CrossMPT~\citep{park2024crossmpt} treats magnitude and syndrome as distinct
modalities coupled by masked cross-attention, and we adopt it as our backbone.

\paragraph{Generative decoders.}
DDECC~\citep{choukroun2022ddecc} decodes by reverse denoising under code constraints, SB-ECC~\citep{sbecc2026} integrates a probability-flow ODE with the solver budget as a latency--accuracy knob, and ECCFM~\citep{lei2025} distils that trajectory into a single step. All three share a continuous state in which every coordinate is updated at once, so no
partial solution ever exists until the trajectory ends. Absorbing-state
discrete diffusion~\citep{austin2021,shi2024} commits coordinates
sequentially, so at every step there is a partially determined error vector,
and it is this object that both the beam search and the certificates act on.

\paragraph{List decoding, certificates and stopping rules.}
Successive cancellation list decoding for polar codes~\citep{tal2015} and
ordered-statistics decoding (OSD) for general linear
codes~\citep{fossorier1995} both search a structured candidate set and pick
the most likely valid member~\citep{chase1972}, the former with a fixed list
size and the latter in a fixed reprocessing order, which is what makes OSD's
early-termination conditions computable. LP decoding~\citep{feldman2005} certifies differently, from the integrality of a convex relaxation over the codeword polytope rather than from any search order, and GRAND~\citep{duffy2019} is ML by construction only because its guessing order is the likelihood order, a guarantee its tractable variants give up~\citep{duffy2021,duffy2022}. Each certificate is thus tied to the structure of its decoder, an exactly solved relaxation or a prescribed order; our beam width adapts to the received word, and to our knowledge no neural decoder reports per-frame optimality certificates. Soft-output GRAND~\citep{yuan2023sogrand} is
the closest prior work to the reliability signal of
Section~\ref{subsec:exp-cert}: a soft output is a calibrated forecast,
available everywhere and exact nowhere, whereas a certificate is a proof on
the frames where it is available and no statement elsewhere.
Appendix~\ref{app:related} expands on all of these.

\section{Preliminaries}
\label{sec:prelim}

\paragraph{Syndrome decoding.}
A binary linear $[n,k]$ code $\mathcal{C}\subseteq\mathbb{F}_2^n$ with
parity-check matrix (PCM) $H\in\mathbb{F}_2^{(n-k)\times n}$ satisfies
$\mathcal{C}=\{ \mathbf{c} : H \mathbf{c}^{\mathsf{T}}=\bm{0}\}$. We write $H_r$ for the $r$-th row,
$\mathrm{supp}(H_r)$ for its support, and $d_{\min}$ for the minimum distance.
Under BPSK over AWGN, with $\mathrm{bpsk}(b)=1-2b$, the receiver observes
$\mathbf{y} =\mathrm{bpsk}(\mathbf{c})+\mathbf{z}$, $\mathbf{z}\sim\mathcal{N}(\bm 0,\sigma^2I_n)$, and forms
\begin{equation}
	r_j=\mathbbm{1}[y_j<0],\qquad
	\bm{s}=H\mathbf{r}^{\mathsf{T}}\in\mathbb{F}_2^{\,n-k},\qquad
	\ell_j=2y_j/\sigma^2 ,
	\label{eq:hard-syndrome-llr}
\end{equation}
so $|\ell_j|$ measures bit reliability, and conditioning on $(|\boldsymbol{\ell}|,\bm{s})$ is codeword-invariant under BPSK symmetry~\citep{bennatan2018}. Writing $\mathbf{e} = \mathbf{r}\oplus \mathbf{c}$, the identity $H \mathbf{e}^{\mathsf{T}}=\bm{s}$ turns decoding into estimation of $\mathbf{e}$, with $\hat{\mathbf{c}} = \mathbf{r} \oplus \hat{\mathbf{e}}$~\citep{richardson2008}. 
\begin{lemma}
	\label{lem:ml}
	Over BPSK+AWGN, ML decoding is equivalent to
	\begin{equation}
		\hat{\mathbf{e}}_{\mathrm{ML}}
		=\arg\min_{\,\mathbf{e}\,:\;H\mathbf{e}^{\mathsf{T}}=\bm{s}}\;\Phi(\mathbf{e}),
		\qquad
		\Phi(\mathbf{e}):=\sum\nolimits_{j} e_j\,|\ell_j| ,
		\label{eq:ml}
	\end{equation}
	and $\hat{\mathbf{c}}_{\mathrm{ML}} = \mathbf{r} \oplus\hat{\mathbf{e}}_{\mathrm{ML}}$.
\end{lemma}

\noindent
The proof is in Appendix~\ref{app:proofs}. We call $\Phi$ the \emph{likelihood cost} and write $\Phi^\star:=\Phi(\hat{\mathbf{e}}_{\mathrm{ML}})$; it is additive over flipped coordinates, so a partially committed assignment has a valid running cost.

\section{Method}
\label{sec:method}
\begin{figure}[t!]
	\centering
	\includegraphics[width=\textwidth]{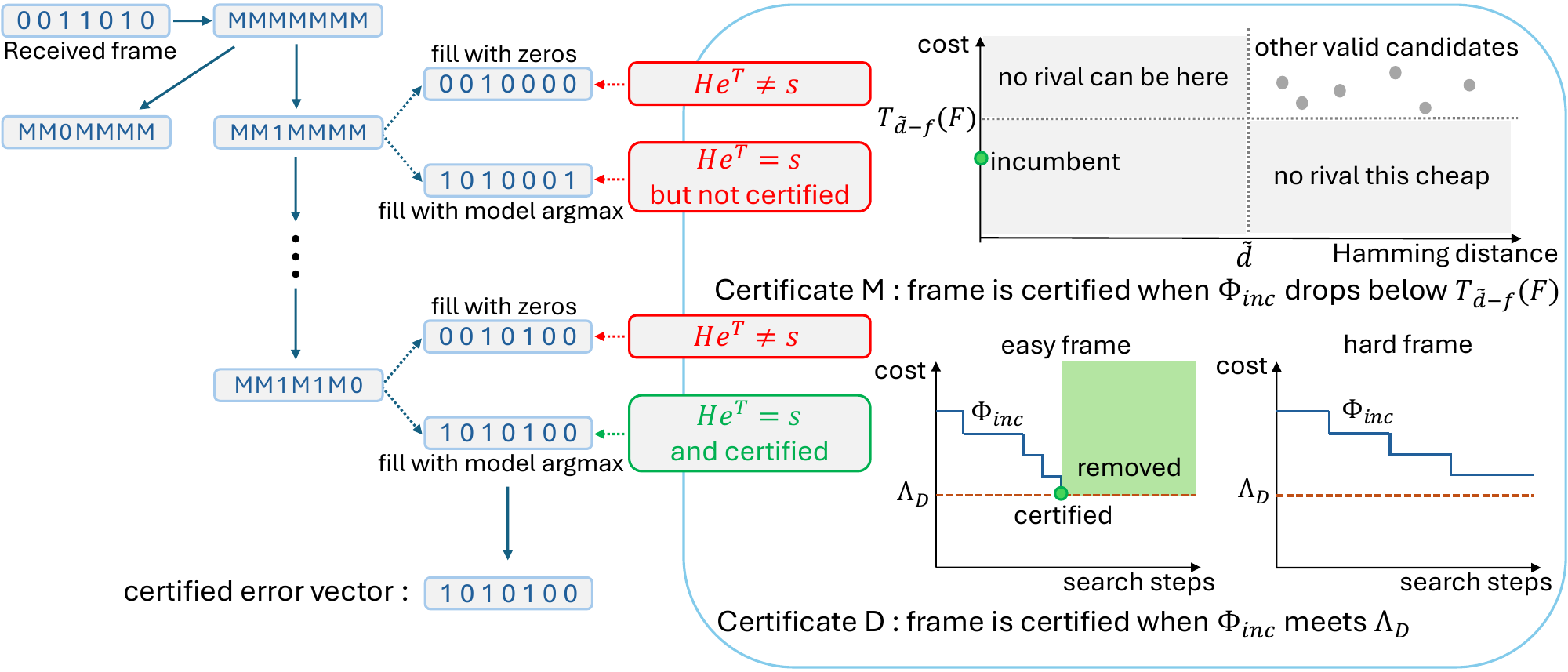}
	\caption{Overview of the proposed CAL-ECC.}
	\label{overview}
	\vspace{-5pt}
\end{figure}

Where prior work draws a single estimate from a generative decoder, we treat
it as a \emph{stochastic candidate generator} and add exact syndrome
filtering, which uses $H$ to verify candidates algebraically, and certified
early termination (Section~\ref{subsec:cert}), which halts a frame as soon as
its best valid candidate is provably the global ML solution.

\paragraph{Selection.}
Given a candidate set $\mathcal{L}_\theta$ we form
$\mathcal{L}^{\mathrm{valid}}=\{\mathbf{e}\in\mathcal{L}_\theta:H\mathbf{e}^{\mathsf{T}}=\bm{s}\}$
and return
\begin{equation}
	\hat{\mathbf{e}} =\arg\min_{\mathbf{e}\in\mathcal{L}^{\mathrm{valid}}}\Phi(\mathbf{e}),
	\qquad \hat{\mathbf{c}} = \mathbf{r} \oplus\hat{\mathbf{e}},
	\label{eq:select}
\end{equation}
falling back to the minimizer of $\Phi$ over all of $\mathcal{L}_\theta$ when
$\mathcal{L}^{\mathrm{valid}}=\emptyset$.

\begin{proposition}[List-restricted ML optimality]
	\label{prop:list-ml}
	If $\mathcal{L}^{\mathrm{valid}}\neq\emptyset$, then $\hat{\mathbf{e}}$ in
	\eqref{eq:select} is the exact ML estimate restricted to
	$\mathcal{L}_\theta$.
\end{proposition}

\noindent
This is immediate from Lemma~\ref{lem:ml} (Appendix~\ref{app:proofs}):
selection is always optimal within the list, so performance is governed
entirely by whether generation produces a set containing
$\hat{\mathbf{e}}_{\mathrm{ML}}$, and closing the remaining gap to
$\Phi^\star$ is what the certificates do.

\subsection{Candidate Generation}
\label{subsec:gen}

We instantiate the generator as absorbing-state
diffusion~\citep{austin2021,shi2024} over error vectors, with alphabet
$\mathcal{A}=\mathbb{F}_2\cup\{\mathsf{m}\}$ where $\mathsf{m}$ marks an
uncommitted coordinate. The forward process replaces each coordinate of the
true error $\mathbf{e}^\star=\mathbf{r}\oplus\mathbf{c}$ independently by
$\mathsf{m}$ with probability $t\in[0,1]$, giving $\mathbf{e}_t$ with mask set
$\mathcal{M}_t$, and the reverse model is trained by cross-entropy over masked
coordinates, a variational bound on the likelihood~\citep{shi2024}:
\begin{equation}
	\mathcal{L}(\theta)
	=-\,\mathbb{E}_{\mathbf{e}^\star,t,\mathbf{e}_t}\!\left[\tfrac{1}{|\mathcal{M}_t|}
	\sum\nolimits_{j\in\mathcal{M}_t}\log p_\theta\!\bigl(e^{\star}_{j}\mid \mathbf{e}_t,|\boldsymbol{\ell}|,\bm{s},t\bigr)\right].
	\label{eq:train}
\end{equation}
The syndrome enters only as conditioning and as the exact verifier in
\eqref{eq:select}, never as a differentiable penalty. At inference we start from
$\mathbf{e}_{(0)}=[\mathsf{m}]^n$ and commit coordinates over $T$ steps, choosing the
commitment order by model confidence. The model is conditioned on $t$, which must track the \emph{actual} fraction of uncommitted coordinates, $t=|\mathcal{M}|/n$, counted from the mask directly because the unit propagation of Section~\ref{subsec:beam} commits several coordinates at once.

\subsection{Probability Beam Search}
\label{subsec:beam}

Sequential commitment makes the generative process a tree. We score
a partial assignment $b$ by its accumulated log-probability
$\mathrm{score}(b)=\sum_{(j,a)\in b_{\mathrm{br}}}\log p_\theta(e_j=a\mid\text{context})$, where $b_{\mathrm{br}}\subseteq b$ are the coordinates committed by branching. At
each step every surviving assignment is branched on the two values of one
uncommitted coordinate, and we prune with a relative threshold
$\tau\in(0,1]$ against the best score in the frame,
\begin{equation}
	\mathrm{score}(b)\;\ge\;\max\nolimits_{b'}\mathrm{score}(b')+\log\tau ,
	\label{eq:beam-prune}
\end{equation}
retaining at most the top $L_{\max}$. The coordinate branched on is the uncommitted one with the largest $\max_a p_\theta(e_j=a\mid b)$ averaged over the surviving beams of the frame, and coordinates committed by unit propagation are forced by the code and add nothing to the score. Pruning adapts to difficulty: on easy frames the beam collapses, on hard frames it expands.

\paragraph{Unit propagation.}
If a parity check has exactly one uncommitted coordinate in its support, that
coordinate is determined by the parity the check requires, so committing it
costs no forward pass and removes no valid completion; we apply this to
exhaustion after every commitment, as a per-frame operation since all beams of
a frame share the mask pattern. The effect is largest for sparse $H$: on
MacKay$(96,48)$ a trajectory that would need $T=96$ forward passes completes
in at most $60$, and on the dense BCH$(63,45)$ the saving is from $63$ to
$59$. All NFE figures we report include this reduction.

\paragraph{Beam exhaustion.}
A lower bound $\Lambda_{\mathrm{B}}$ on the cost of any valid completion of the
\emph{surviving} beam is available from the committed flips and the checks
still unsatisfied (Appendix~\ref{app:beam-exhaust}); if the incumbent $\Pinc$
of Section~\ref{subsec:cert} costs no more, the search may stop
(Proposition~\ref{prop:beam-exhaust}). Because pruned branches lie outside its
scope, $\Lambda_{\mathrm{B}}$ is \emph{not} a lower bound on $\Phi^\star$ and
carries no claim of global optimality; we report the two stopping rates
separately throughout.

\subsection{Certified Early Termination}
\label{subsec:cert}

Viewing the beam search as branch-and-bound for \eqref{eq:ml}, the diffusion
model supplies the branching rule and $\tau$-pruning corresponds to node
elimination. The missing ingredient is a bound on $\Phi^\star$ itself. If the
incumbent cost is no larger, the incumbent is already globally optimal and the
frame may stop. At every step, immediately after the forward pass and
\emph{without any extra forward pass}, we form two completions of each
surviving assignment. $\mathbf{e}^{\mathrm{part}}$ fills uncommitted coordinates with $0$ and $\mathbf{e}^{\mathrm{greedy}}$ fills them with the $\arg\max$ of the current probabilities.

\paragraph{Certificate D: a dual bound.}
Let $\mathcal{R}=\{r:s_r=1\}$ be the checks unsatisfied by the hard decision. Every valid $\mathbf{e}$ flips an odd
number of coordinates in each such support, hence
$\sum_{j\in\mathrm{supp}(H_r)}e_j\ge1$, and the covering LP
$\min_{\mathbf{x}\in[0,1]^n}\sum_j|\ell_j|x_j$ subject to those constraints
lower-bounds $\Phi^\star$. Its dual is the packing LP
\begin{equation}
	\max_{\mathbf{y}\ge\bm{0}}\ \sum\nolimits_{r\in\mathcal{R}} y_r
	\quad\text{such that}\quad
	\sum\nolimits_{r\,:\,j\in\mathrm{supp}(H_r)} y_r \ \le\ |\ell_j|, \ \ \text{ for all }j ,
	\label{eq:lp-dual}
\end{equation}
so by weak duality~\citep{schrijver1998}, $\sum_r y_r\le\Phi^\star$ for
\emph{any} feasible $\mathbf{y}$. We construct one by greedy dual ascent over $\mathcal{R}$ (Appendix~\ref{app:algorithm}), capped at $K=12$ rounds; every intermediate $\mathbf{y}$ is feasible, so
$\Lambda_{\mathrm{D}}:=\sum_r y_r$ is valid \emph{at whatever round the ascent
	is truncated}, and Appendix~\ref{app:lpgap} measures what the exact optimum of
\eqref{eq:lp-dual} would add. Since $\Lambda_{\mathrm{D}}$ bounds $\Phi^\star$, a constant of the frame, it is computed \emph{once} at the fully masked state and reused at every step.

\paragraph{Certificate M: a distance-based test.}
The second certificate uses $d_{\min}$ and depends on the incumbent alone. Let
$F=\mathrm{supp}(\mathbf{e})$, $f=|F|$, and let $T_m(F)$ be the sum of the $m$ smallest
$|\ell_j|$ over $j\notin F$.

\begin{theorem}[Distance-based ML certificate]
	\label{thm:dmin}
	Let $\mathcal{C}$ be a binary $[n,k,d_{\min}]$ linear code, and let $\tilde d\le d_{\min}$. 
	If a valid $\mathbf{e}$ satisfies $f<\tilde d$ and $\Phi(\mathbf{e})\le T_{\tilde d-f}(F)$, then $\Phi(\mathbf{e})=\Phi^\star$.
\end{theorem}

\noindent
The proof is in Appendix~\ref{app:proof-dmin}. When $f\ge\tilde d$ the test is
vacuous, as $F$ may then contain a whole codeword, and no certification is attempted. $\tilde d$ is a fixed invariant of the code, not a
hyperparameter: we use the BCH bound for BCH codes and the minimum generator
row weight for polar codes (Table~\ref{tab:dmin},
Appendix~\ref{app:code-identity}), and where no lower bound is available the
certificate is disabled, degrading gracefully to Certificate~D alone.

\paragraph{Termination rule.}
A frame is certified, and removed from the search with output
$\einc$, when
\begin{equation}
	\Pinc\;\le\;\Lambda_{\mathrm{D}}
	\qquad\text{or}\qquad
	\Pinc\;\le\;T_{\tilde d-f}(\Finc) .
	\label{eq:stop}
\end{equation}
Frames may additionally stop by beam exhaustion
(Proposition~\ref{prop:beam-exhaust}). Those stops are counted separately and
carry no global claim.

\begin{theorem}[Correctness of certification]
	\label{thm:global-ml}
	A frame terminated by \eqref{eq:stop} outputs
	$\Phi(\einc)=\Phi^\star$, so its output is a global ML solution.
	Consequently its BER and FER are no larger than those obtained by running
	the search to completion.
\end{theorem}

\begin{proof}
	For the first branch, $\Lambda_{\mathrm{D}}\le\Phi^\star$ by weak duality
	of \eqref{eq:lp-dual}, so $\Pinc\le\Phi^\star$. Since
	$\einc$ is valid we also have $\Pinc\ge\Phi^\star$, and
	equality follows. The second branch is Theorem~\ref{thm:dmin}. The BER/FER
	claim holds because list-restricted optimality
	(Proposition~\ref{prop:list-ml}) is weaker than global optimality.
\end{proof}

\noindent
The floating-point implementation certifies to a tolerance $\varepsilon$ (Appendix~\ref{app:numerics}). The two certificates exploit different
structure and are complementary: Certificate~D is strong for sparse $H$, where
check supports rarely overlap and the feasible region of \eqref{eq:lp-dual} is
large, and Certificate~M for dense $H$ with a large $d_{\min}$, while at low
SNR $f$ frequently exceeds $\tilde d$ and voids Certificate~M by construction
(Section~\ref{subsec:exp-cert}).

\section{Experiments}
\label{sec:exp}

\paragraph{Setup.}
We evaluate on BCH~\citep{bose1960}, polar~\citep{arikan2009},
LDPC~\citep{gallager1962}, MacKay~\citep{mackay1999} and
CCSDS~\citep{ccsds2011} codes, spanning block lengths $31$--$128$ and rates
$0.49$--$0.81$, with parity-check matrices verified against the public
references (Appendix~\ref{app:code-identity}). All model-free entries use the same backbone scale ($N_{\mathrm{dec}}{=}6$, $d_{\mathrm{model}}{=}128$), and our generator uses the CrossMPT backbone so that the comparison isolates the decoding method rather than the architecture; each decoder follows its own published training schedule, ours included (Appendix~\ref{app:impl}). At inference we branch on one coordinate per step with beam width $L_{\max}{=}256$ and report three thresholds $\tau\in\{10^{-1},10^{-2},10^{-3}\}$; each cell is decoded until $500$ frame errors or a frame cap. Baselines are ECCT~\citep{choukroun2022ecct} and
CrossMPT~\citep{park2024crossmpt}, both single-shot, and the iterative
DDECC~\citep{choukroun2022ddecc} and SB-ECC~\citep{sbecc2026}, the latter two
with the standard syndrome stopping rule enabled so that their reported NFE is
realised cost rather than budget. Details are in Appendix~\ref{app:impl}.

\subsection{Accuracy}
\label{subsec:exp-main}

\begin{table}[t]
	\centering\setlength{\tabcolsep}{2.1pt}\renewcommand{\arraystretch}{1.0}
	\caption{Decoding accuracy, $-\ln(\mathrm{BER})$ at $E_b/N_0\in\{4,5,6\}$~dB (higher is better). Per cell the best entry is \underline{\textbf{bold and underlined}} and the runner-up \underline{underlined}. The $2$--$3$~dB columns are in Appendix~\ref{app:full-results}.}
	\label{tab:main}
	\resizebox{\textwidth}{!}{%
		\begin{tabular}{ll ccc ccc ccc ccc ccc ccc ccc}
			\toprule
			\multirow{2}{*}{Code} & \multirow{2}{*}{$(n,k)$} & \multicolumn{3}{c}{ECCT} & \multicolumn{3}{c}{CrossMPT} & \multicolumn{3}{c}{DDECC} & \multicolumn{3}{c}{SB-ECC} & \multicolumn{3}{c}{\textbf{CAL-ECC} $10^{-1}$} & \multicolumn{3}{c}{\textbf{CAL-ECC} $10^{-2}$} & \multicolumn{3}{c}{\textbf{CAL-ECC} $10^{-3}$} \\
			\cmidrule(lr){3-5}\cmidrule(lr){6-8}\cmidrule(lr){9-11}\cmidrule(lr){12-14}\cmidrule(lr){15-17}\cmidrule(lr){18-20}\cmidrule(lr){21-23}
			& & 4 & 5 & 6 & 4 & 5 & 6 & 4 & 5 & 6 & 4 & 5 & 6 & 4 & 5 & 6 & 4 & 5 & 6 & 4 & 5 & 6 \\
			\midrule
			\multirow{4}{*}{BCH} & (31,16) & 5.90 & 7.81 & 10.34 & 7.00 & 9.35 & 12.36 & 4.80 & 6.13 & 7.67 & \underline{7.29} & \underline{9.81} & \underline{12.90} & 6.61 & 8.93 & 11.83 & 7.19 & 9.56 & 12.47 & \underline{\textbf{7.48}} & \underline{\textbf{9.89}} & \underline{\textbf{13.04}} \\
			& (63,36) & 4.66 & 6.51 & 9.18 & 4.79 & 6.75 & 9.45 & 5.30 & 7.29 & 9.98 & 5.52 & 7.55 & 10.49 & 5.36 & 7.62 & 10.71 & \underline{6.44} & \underline{8.78} & \underline{12.00} & \underline{\textbf{7.43}} & \underline{\textbf{10.16}} & \underline{\textbf{13.53}} \\
			& (63,45) & 5.33 & 7.46 & 10.41 & 5.31 & 7.53 & 10.50 & 5.77 & 7.95 & 10.54 & 6.29 & 9.01 & 12.54 & 6.04 & 8.46 & 11.86 & \underline{7.12} & \underline{9.74} & \underline{13.01} & \underline{\textbf{7.82}} & \underline{\textbf{10.86}} & \underline{\textbf{14.49}} \\
			& (63,51) & 5.43 & 7.45 & 10.50 & 5.64 & 7.89 & 11.17 & 5.25 & 7.14 & 9.25 & 5.50 & 7.64 & 10.79 & 6.22 & 8.78 & 12.27 & \underline{6.65} & \underline{9.38} & \underline{12.89} & \underline{\textbf{6.80}} & \underline{\textbf{9.63}} & \underline{\textbf{13.28}} \\
			\midrule
			\multirow{5}{*}{Polar} & (64,32) & 6.63 & 9.04 & 11.89 & 7.63 & 10.26 & 13.29 & 7.05 & 9.50 & 12.23 & 7.52 & 10.15 & 13.26 & 7.61 & 10.30 & 13.67 & \underline{7.89} & \underline{10.58} & \underline{13.89} & \underline{\textbf{7.97}} & \underline{\textbf{10.67}} & \underline{\textbf{13.96}} \\
			& (64,48) & 6.19 & 8.28 & 10.97 & 6.41 & 8.56 & 11.22 & 6.00 & 8.00 & 10.40 & 6.59 & 8.62 & 11.23 & 6.59 & 8.63 & 11.24 & \underline{6.63} & \underline{8.65} & \underline{11.25} & \underline{\textbf{6.66}} & \underline{\textbf{8.67}} & \underline{\textbf{11.25}} \\
			& (128,64) & 4.84 & 7.33 & 10.72 & 5.64 & 8.71 & 13.24 & \underline{\textbf{8.31}} & \underline{\textbf{12.06}} & \underline{15.73} & \underline{7.77} & \underline{11.29} & 15.43 & 6.29 & 9.94 & 14.92 & 6.73 & 10.22 & 15.13 & 7.61 & 11.22 & \underline{\textbf{15.94}} \\
			& (128,86) & 6.31 & 9.10 & 12.48 & 7.62 & 10.94 & 15.25 & 7.84 & 11.18 & 15.25 & 7.87 & 11.07 & 14.98 & 8.04 & 11.77 & 16.09 & \underline{8.51} & \underline{12.14} & \underline{16.54} & \underline{\textbf{9.00}} & \underline{\textbf{12.65}} & \underline{\textbf{17.06}} \\
			& (128,96) & 6.29 & 9.09 & 12.31 & 7.08 & 10.17 & 13.14 & 7.34 & 10.22 & 13.14 & 7.28 & 10.22 & 13.18 & 7.53 & 10.54 & 13.30 & \underline{7.83} & \underline{10.59} & \underline{13.32} & \underline{\textbf{8.11}} & \underline{\textbf{10.67}} & \underline{\textbf{13.33}} \\
			\midrule
			\multirow{4}{*}{LDPC} & (49,24) & 6.10 & 8.58 & 12.16 & 6.70 & 9.52 & 13.30 & 6.00 & 8.48 & 11.78 & 7.15 & 10.34 & 14.18 & 7.38 & 10.35 & 14.32 & \underline{8.21} & \underline{10.95} & \underline{14.65} & \underline{\textbf{8.72}} & \underline{\textbf{11.49}} & \underline{\textbf{14.94}} \\
			& (121,60) & 5.11 & 8.23 & 12.98 & 5.79 & 9.29 & 14.94 & 5.52 & 9.04 & 14.35 & 6.23 & 10.13 & 16.08 & 6.24 & 10.18 & 16.31 & \underline{7.01} & \underline{10.64} & \underline{16.53} & \underline{\textbf{8.25}} & \underline{\textbf{12.24}} & \underline{\textbf{17.91}} \\
			& (121,70) & 6.31 & 10.11 & 15.55 & 7.04 & 11.36 & 17.28 & 6.95 & 11.31 & 17.16 & 7.75 & 12.41 & 18.98 & 7.72 & 12.55 & 19.19 & \underline{8.40} & \underline{13.10} & \underline{19.47} & \underline{\textbf{9.85}} & \underline{\textbf{14.43}} & \underline{\textbf{20.30}} \\
			& (121,80) & 7.19 & 11.40 & 17.12 & 8.00 & 12.66 & 18.24 & 8.16 & 12.69 & 17.74 & 8.77 & 13.89 & 21.39 & 8.77 & 13.81 & 21.09 & \underline{9.45} & \underline{14.30} & \underline{21.61} & \underline{\textbf{10.58}} & \underline{\textbf{15.42}} & \underline{\textbf{22.38}} \\
			\midrule
			\multirow{1}{*}{MacKay} & (96,48) & 7.13 & 10.46 & 14.51 & 7.82 & 11.53 & 15.60 & 8.28 & 12.12 & 15.95 & 8.60 & 12.40 & 16.21 & 8.85 & 12.72 & 16.26 & \underline{9.40} & \underline{13.01} & \underline{16.34} & \underline{\textbf{10.31}} & \underline{\textbf{13.40}} & \underline{\textbf{16.40}} \\
			\midrule
			\multirow{1}{*}{CCSDS} & (128,64) & 6.75 & 10.62 & 15.81 & 7.65 & 12.19 & 17.58 & 7.98 & 12.66 & 17.18 & 8.36 & 13.18 & 19.42 & 8.59 & 13.64 & 21.33 & \underline{8.96} & \underline{13.83} & \underline{21.57} & \underline{\textbf{9.95}} & \underline{\textbf{14.70}} & \underline{\textbf{21.79}} \\
			\bottomrule
	\end{tabular}}
\end{table}

CAL-ECC at $\tau=10^{-3}$ attains the best result in $43$ of the $45$ cells
shown in Table~\ref{tab:main} and in $28$ of the $30$ at $2$--$3$~dB
(Appendix~\ref{app:full-results}), or $71$ of $75$ overall. All four
exceptions are Polar$(128,64)$ at $E_b/N_0\le5$~dB, the cell with the lowest
certification rate in our suite, and Section~\ref{subsec:exp-cert} shows the
two facts have a common cause. The margin over the strongest prior decoder
tracks how much headroom that decoder leaves rather than block length as such.
It is smallest on the length-$64$ polar codes, below $0.7$ in
$-\ln(\mathrm{BER})$ and as little as $0.02$ on Polar$(64,48)$, where prior
decoders already operate near the ML ceiling (Section~\ref{subsec:exp-ml}), and
we take that cell as parity rather than improvement, reached at higher average
cost than SB-ECC (Section~\ref{subsec:exp-nfe}). It exceeds $1.3$ on
LDPC$(121,60/70)$ and reaches $3.04$ on BCH$(63,36)$, the largest in the
suite. In $E_b/N_0$ it widens on the dense-$H$ BCH codes, where a large
$d_{\min}$ keeps Certificate~M firing, and narrows on most sparse-$H$ codes,
where the strongest baseline closes in.

\subsection{Where the Gain Comes From}
\label{subsec:exp-ablation}

The decoder has three parts that could each account for the margins of
Table~\ref{tab:main}: the diffusion generator, the beam search over the tree
it induces, and the certificates. Table~\ref{tab:ablation} separates them on
four codes chosen to cover dense and sparse $H$ and to include both cells we
win and the one we lose; all three configurations use the same trained model.

\begin{table}[t]
	\centering\setlength{\tabcolsep}{4.5pt}\small
	\caption{Ablation at $E_b/N_0=4$~dB, $\tau=10^{-3}$: \textbf{(a)} the generator decoded greedily along a single trajectory ($\tau{=}1$, $L{=}1$); \textbf{(b)} the full beam search with certificates disabled; \textbf{(c)} the complete decoder. ``search $\Delta$BER'' is (b)$-$(a) in $-\ln(\mathrm{BER})$ and ``cert.\ NFE cut'' is (b)/(c) in NFE.}
	\label{tab:ablation}
	\resizebox{\linewidth}{!}{%
		\begin{tabular}{l rr rr rr rr r}
			\toprule
			\multirow{2}{*}{Code}
			& \multicolumn{2}{c}{(a) generator}
			& \multicolumn{2}{c}{(b) $+$ search}
			& \multicolumn{2}{c}{(c) $+$ certificates}
			& \multicolumn{2}{c}{effect} & \multirow{2}{*}{(c)/(a)} \\
			\cmidrule(lr){2-3}\cmidrule(lr){4-5}\cmidrule(lr){6-7}\cmidrule(lr){8-9}
			& $-\ln\mathrm{BER}$ & NFE & $-\ln\mathrm{BER}$ & NFE
			& $-\ln\mathrm{BER}$ & NFE
			& search & cert. & NFE \\
			& & & & & & & $\Delta$BER & NFE cut & ratio \\
			\midrule
			MacKay$(96,48)$ & $8.76$ & $53.8$ & $10.26$ & $54.7$  & $10.31$ & $23.1$  & $+1.50$ & $2.4\times$ & $\bm{0.43\times}$ \\
			LDPC$(121,80)$  & $8.72$ & $88.4$ & $10.65$ & $92.4$  & $10.58$ & $29.0$  & $+1.93$ & $3.2\times$ & $\bm{0.33\times}$ \\
			BCH$(63,45)$    & $5.82$ & $45.2$ & $7.81$  & $149.7$ & $7.82$  & $113.8$ & $+1.99$ & $1.3\times$ & $2.52\times$ \\
			Polar$(128,64)$ & $6.37$ & $82.9$ & $7.67$  & $171.8$ & $7.61$  & $156.4$ & $+1.30$ & $1.1\times$ & $1.89\times$ \\
			\bottomrule
	\end{tabular}}
\end{table}

\paragraph{The search produces the accuracy; the certificates change the cost, not the answer.}
Column (a), the same trained model decoded along a single greedy trajectory,
is \emph{not} competitive with the baselines: against the strongest prior
decoder on the same cell it is behind by $0.47$ on BCH$(63,45)$, $0.05$ on
LDPC$(121,80)$ and $1.94$ on Polar$(128,64)$, and ahead by only $0.16$ on
MacKay$(96,48)$. Adding the beam search moves it by $1.30$ to $1.99$, which
is the whole of our margin over prior work; the generator is a good branching
heuristic and nothing more. The search is almost free on the sparse-$H$ cells and triples or doubles the cost on the dense-$H$ cells, where the beam must spread before any completion satisfies every check. Column (c) then differs from (b) by $-0.07$ to $+0.05$
in $-\ln(\mathrm{BER})$, inside the sampling error, as
Theorem~\ref{thm:global-ml} predicts; what changes is the computation, by a
factor that tracks the certification rate, $3.2\times$ on LDPC$(121,80)$
where $72.4\%$ of frames are certified and $1.1\times$ on Polar$(128,64)$
where $18.7\%$ are. BCH$(63,45)$ is the apparent exception at $1.3\times$
despite $83.7\%$ certification: a certificate cannot fire before a valid
incumbent exists, and on a dense $H$ that takes most of the trajectory, so
there is simply less of the search left to remove.

\paragraph{The accuracy is not bought with computation.}
Only (a) and (c) are configurations one would actually deploy, and against (a)
the complete decoder is \emph{simultaneously more accurate and cheaper} on the
two sparse-$H$ cells, at $0.43\times$ the evaluations on MacKay$(96,48)$ for
$1.55$ and $0.33\times$ on LDPC$(121,80)$ for $1.86$, because the
certificates halt many frames earlier than a greedy trajectory ever finishes.
Moreover $4$~dB is where our computation is at its most expensive relative to
a fixed-budget decoder; the same four cells cost $2.7$, $1.2$, $0.7$ and
$12.1$ evaluations at $6$~dB (Section~\ref{subsec:exp-nfe}). Polar$(128,64)$, the one cell we lose, is weakest in every column, for the reason Section~\ref{subsec:exp-cert} gives.

\subsection{Reaching the Exact ML Ceiling}
\label{subsec:exp-ml}

\begin{figure}[t]
	\centering
	\begin{subfigure}[b]{0.195\textwidth}
		\includegraphics[width=\textwidth]{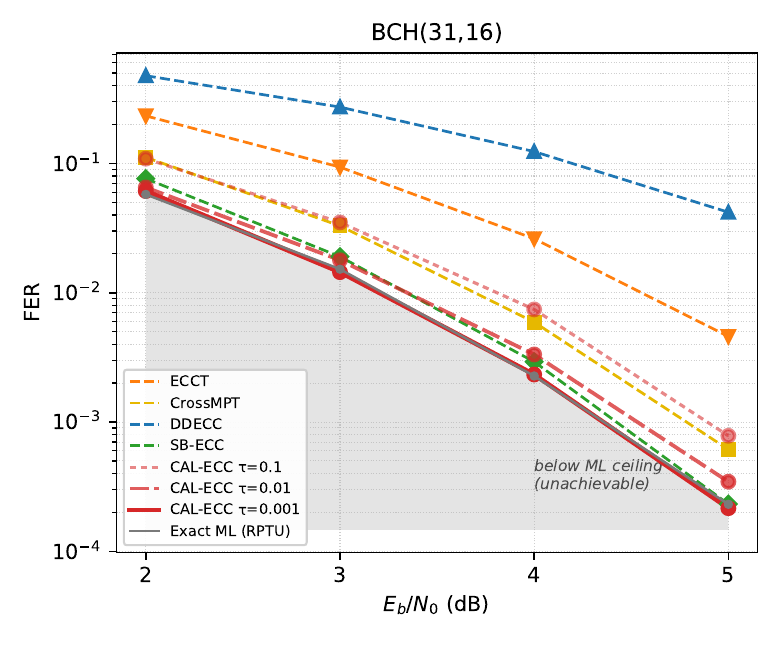}
		\caption{BCH$(31,16)$}
	\end{subfigure}\hfill
	\begin{subfigure}[b]{0.195\textwidth}
		\includegraphics[width=\textwidth]{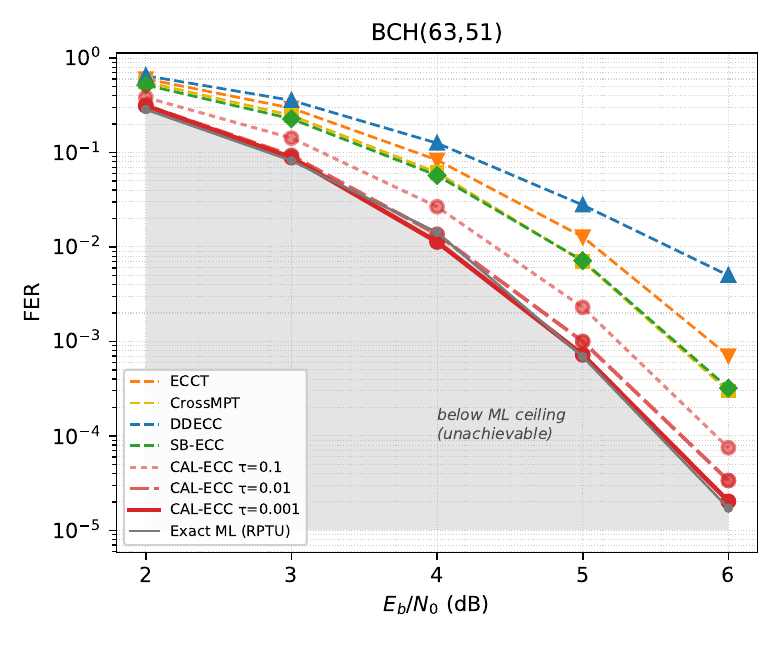}
		\caption{BCH$(63,51)$}
	\end{subfigure}\hfill
	\begin{subfigure}[b]{0.195\textwidth}
		\includegraphics[width=\textwidth]{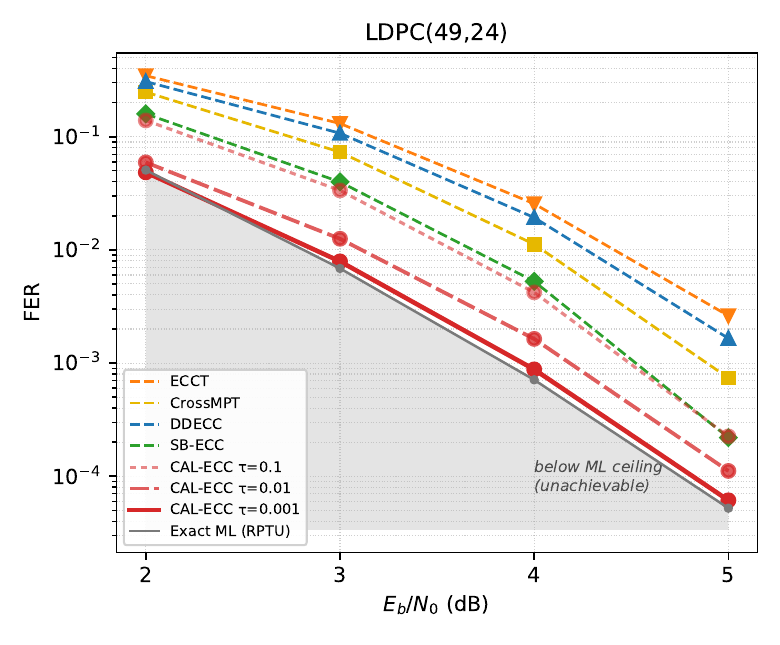}
		\caption{LDPC$(49,24)$}
	\end{subfigure}\hfill
	\begin{subfigure}[b]{0.195\textwidth}
		\includegraphics[width=\textwidth]{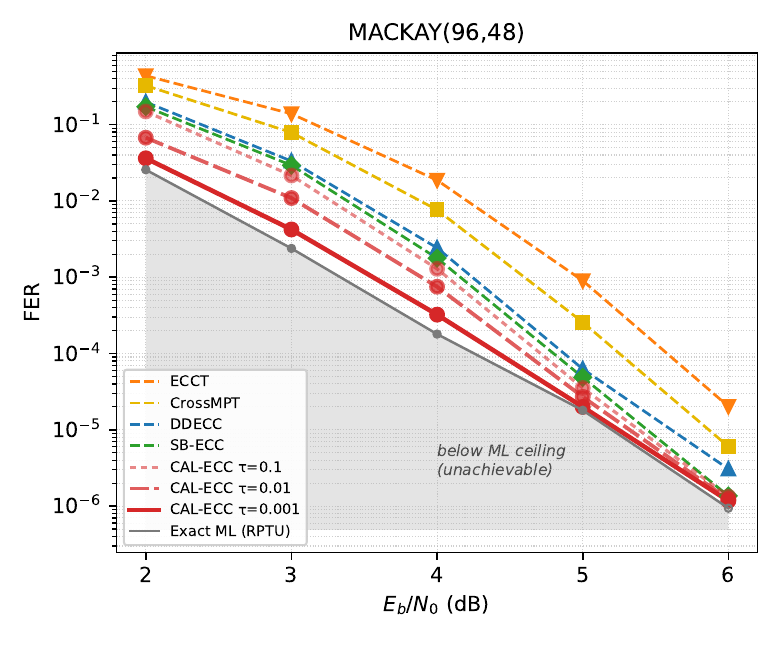}
		\caption{MacKay$(96,48)$}
	\end{subfigure}\hfill
	\begin{subfigure}[b]{0.195\textwidth}
		\includegraphics[width=\textwidth]{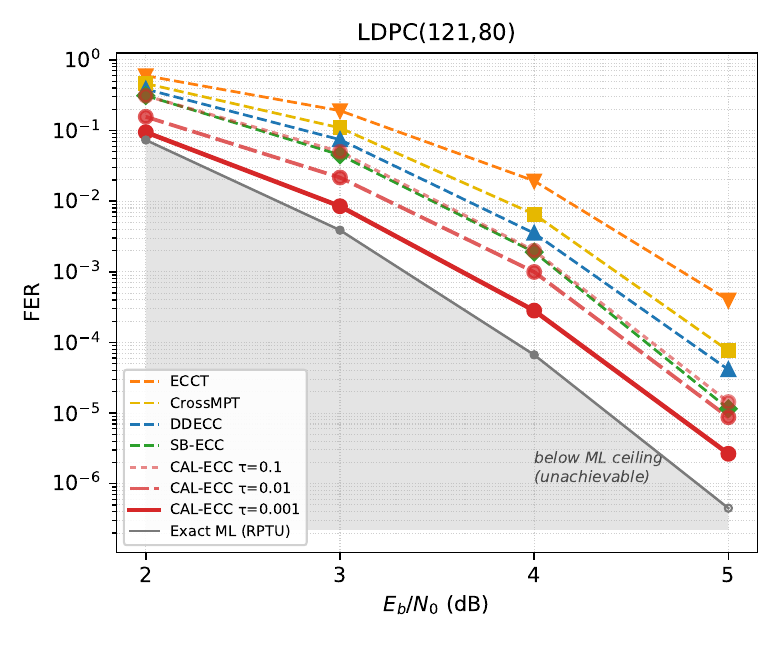}
		\caption{LDPC$(121,80)$}
	\end{subfigure}
	\caption{FER against the exact ML reference (grey) over the $E_b/N_0$ range where it exists; the shaded region below it is unachievable by \emph{any} decoder. CAL-ECC curves darken as $\tau$ decreases.}
	\label{fig:ml-ceiling}
	\vspace{-7pt}
\end{figure}
Exact ML references from integer-programming decoders are published for eleven of our fifteen codes, over the $E_b/N_0$ range where the integer program remains feasible~\citep{helmling2025}, labelled RPTU in the figures after the hosting institution. Figure~\ref{fig:ml-ceiling} shows five representative panels and Appendix~\ref{app:figures} the rest and, since ML is block-optimal, the corresponding FER at $5$~dB (Table~\ref{tab:ml-ceiling}), the highest point at which a reference exists for all five.

\paragraph{On four of the five codes the accuracy axis is closed.}
Taking the ratio to the reference, CAL-ECC shows $0.93$, $1.17$, $1.05$ and
$1.11$ on the four codes whose reference is uncensored. The reference is
measured to $100$ frame errors and ours to about $500$; treating both counts
as Poisson gives a combined relative standard deviation of
$\sqrt{1/100+1/500}=11\%$, so these deviations are $0.6\sigma$, $1.6\sigma$,
$0.5\sigma$ and $1.0\sigma$, within one standard deviation of exact ML on
three of the four codes and within two on the fourth.
Figure~\ref{fig:ml-ceiling} shows the same thing qualitatively: the three
curves obtained by varying $\tau$ separate at low $E_b/N_0$ and then collapse
onto the reference, which is what one expects once the candidate set contains
$\hat{\mathbf{e}}_{\mathrm{ML}}$ on essentially every frame. Two ratios below one are finite-sample fluctuations rather than a true advantage over a frame-optimal decoder: they move between operating points (BCH$(63,51)$ gives $0.79$ at $4$~dB but $1.05$ at $5$~dB), whereas a systematic bias would recur.

\subsection{Anatomy of Certification}
\label{subsec:exp-cert}

Three structural predictions follow from Section~\ref{subsec:cert}:
Certificate~D should be strong for sparse $H$, Certificate~M for codes with a
large known $d_{\min}$, and both should weaken at low $E_b/N_0$ where $f$
exceeds $\tilde d$. Table~\ref{tab:cert} tests all three. Throughout,
$\mathrm{cert}\%$ is the fraction of frames whose output is \emph{provably}
the global ML solution, and frames stopped by beam exhaustion are counted as
uncertified.
\begin{table}[t]
	\centering
	\setlength{\tabcolsep}{4pt}\small
	\caption{Certificate breakdown at $E_b/N_0=4$~dB, $\tau=10^{-3}$, sorted by the row weight of $H$. $[\mathrm{D}]$ and $[\mathrm{M}]$ are the fractions certified by each test; $\mathrm{cert}\%$ additionally includes frames with $\boldsymbol{s}=\boldsymbol{0}$. ``detected'' is the share of frame errors whose output fails the syndrome check, ``errors there'' the share falling in the uncertified fraction, and selectivity the bound $1/(1-\mathrm{cert})$.}
	\label{tab:cert}
	\begin{tabular}{lrrrrrrrr}
		\toprule
		Code & row wt.\ $H$ & $\mathrm{cert}\%$ & $[\mathrm{D}]\%$ & $[\mathrm{M}]\%$
		& uncert.\% & detected & errors there & selectivity \\
		\midrule
		MacKay$(96,48)$ & $6$    & $59.5$ & $59.2$ & $0.0$  & $40.5$ & $38.1\%$  & $98.8\%$  & $2.5\times$ \\
		LDPC$(49,24)$   & $7$    & $81.3$ & $15.1$ & $60.9$ & $18.7$ & $7.0\%$   & $98.9\%$  & $5.3\times$ \\
		BCH$(31,16)$    & $8$    & $89.2$ & $0.5$  & $70.6$ & $10.8$ & $1.7\%$   & $96.7\%$  & $9.3\times$ \\
		CCSDS$(128,64)$ & $8$    & $41.0$ & $40.9$ & $0.0$  & $59.0$ & $100.0\%$ & $100.0\%$ & $1.7\times$ \\
		LDPC$(121,70)$  & $11$   & $63.0$ & $62.5$ & $0.0$  & $37.0$ & $92.8\%$  & $100.0\%$ & $2.7\times$ \\
		LDPC$(121,80)$  & $11$   & $72.4$ & $70.9$ & $0.0$  & $27.6$ & $66.7\%$  & $100.0\%$ & $3.6\times$ \\
		Polar$(128,64)$ & $13.5$ & $18.7$ & $15.3$ & $3.3$  & $81.3$ & $94.5\%$  & $100.0\%$ & $1.2\times$ \\
		BCH$(63,36)$    & $18$   & $88.8$ & $0.0$  & $83.4$ & $11.2$ & $97.2\%$  & $100.0\%$ & $8.9\times$ \\
		BCH$(63,51)$    & $28$   & $79.6$ & $0.7$  & $53.6$ & $20.4$ & $0.0\%$   & $97.3\%$  & $4.9\times$ \\
		\bottomrule
	\end{tabular}
	\vspace{-7pt}
\end{table}

\paragraph{The two certificates are complementary, as predicted.}
For the dense-$H$ BCH codes Certificate~D contributes essentially nothing
($0.0$--$0.7\%$) while Certificate~M carries $53.6$--$83.4\%$, and for the
sparse-$H$ LDPC$(121,80)$, MacKay and CCSDS the roles are exactly reversed,
the zeros in $[\mathrm{M}]$ being the designed fallback where no lower bound
on $d_{\min}$ is available. LDPC$(49,24)$, the only code that is both sparse
and has an exactly known $d_{\min}$, has \emph{both} certificates firing
($15.1\%$ and $60.9\%$), so the complementarity follows from structure rather
than from code family. Polar$(128,64)$ is the predicted failure mode: its $H$
is neither sparse enough for D nor equipped with a large enough $\tilde d$ for
M, and up to $4$~dB it is the only cell where both are simultaneously weak
($15.3\%$ and $3.3\%$), recovering to $\mathrm{cert}=85.2\%$ at $6$~dB. It is
precisely at $E_b/N_0\le5$~dB, where certification is weakest, that this is
the only code on which a prior decoder wins; weak certification is necessary
for that outcome but not sufficient, since Polar$(128,96)$ at $4$~dB
certifies a comparable $21.6\%$ yet we win there by $+0.77$, and what
separates Polar$(128,64)$ is that its generator is also the furthest behind
(Table~\ref{tab:ablation}).

\paragraph{Low SNR voids Certificate~M by construction.}
On BCH$(63,45)$ the certification rate falls monotonically from $99.8\%$ at
$6$~dB to $24.7\%$ at $2$~dB, because $f$ exceeds $\tilde d$ on an increasing
fraction of frames and the test of Theorem~\ref{thm:dmin} cannot be applied
at all. Certification is thus a property of the operating point as much as of
the code, and we report it per cell (Appendix~\ref{app:full-results}).

\paragraph{Certification is a per-frame reliability signal.}
The last columns report a consequence we did not design for. Errors fall
almost entirely in the uncertified fraction: across all nine cells
$96.7$--$100\%$ of frame errors occur on uncertified frames, and on five of
them every observed error did, including all $508$ errors on LDPC$(121,80)$
against $1{,}302{,}404$ certified frames with none. Part of this any decoder
could claim: an output that fails the syndrome check has visibly failed, and
on CCSDS, BCH$(63,36)$, Polar$(128,64)$ and LDPC$(121,70)$ such detected
failures are $93$--$100\%$ of the errors (column ``detected''). The
certificate's own contribution lies in the remainder, the outputs that are
valid codewords and wrong, which no syndrome check can flag: on the five cells
with at least $100$ such errors the certificate captures $96.7$--$100\%$,
against $88.5$--$100\%$ for flagging the same fraction of frames by output
cost $\Phi$ (Appendix~\ref{app:reliability}). The mechanism is immediate in
hindsight, since a certified frame outputs the ML solution and ML rarely errs
at these operating points.

What the certificate adds over such a heuristic is not selectivity but a
statement about what remains. If uncertified frames were flagged for
retransmission, the last column gives $1/(1-\mathrm{cert})$, an upper bound
on the selectivity when every error is uncertified: on BCH$(31,16)$, flagging
$10.8\%$ of frames captures $96.7\%$ of the errors, a lift of $9.0\times$,
while Polar$(128,64)$ marks the limit, where the signal is perfect but
$81.3\%$ of frames would have to be flagged. Used instead as a filter, over
the nine cells we certify $4{,}108{,}339$ frames and observe $46$ errors
among them, an error rate of $1.1\times10^{-5}$, one to two orders of
magnitude below the decoder's own. A certificate proves that our output equals the ML output, not that it equals the transmitted codeword, so the residual $1.1\times10^{-5}$ is the rate at which ML itself fails on those frames, not a miscalibration of ours.

\subsection{Computational Cost}
\label{subsec:exp-nfe}

\paragraph{Two stopping rules.}
SB-ECC halts as soon as the hard decision satisfies every parity check, and we measure its realised cost under that rule rather than its nominal budget (Appendix~\ref{app:impl}). Our decoder can be measured under the same rule, reported as $\mathrm{NFE}_{\mathrm{sd}}$, but it does not stop there:
syndrome consistency says nothing about optimality, so the search continues
until the incumbent is \emph{proved} globally optimal, and the distance
between $\mathrm{NFE}_{\mathrm{sd}}$ and $\mathrm{NFE}$ (Table~\ref{tab:nfe}, Appendix~\ref{app:full-results}) is the price of certification. NFE counts single-input forward passes, a
forward pass over $|\mathcal{B}|$ beams counting as $|\mathcal{B}|$, so the
unit is the same in arithmetic for both decoders; what differs is latency, where our beams run in one batch and SB-ECC's evaluations are sequential, at the price of a peak activation memory that scales with $|\mathcal{B}|\le L_{\max}$ rather than with a single trajectory.

\paragraph{Certification is bought, and its price is set by the certification rate.}
Under the syndrome rule our cost crosses SB-ECC's at roughly $4$~dB: taking
$\tau=10^{-1}$ and the median over the fifteen codes, we are $7.6\times$
dearer at $2$~dB, comparable at $4$~dB ($1.01\times$), and $1.61\times$
cheaper at $6$~dB, where we are cheaper on all fifteen codes, and over all
$75$ cells the median ratio is $1.01\times$. That comparison uses $\tau=10^{-1}$; at the $\tau=10^{-3}$ setting of
Table~\ref{tab:main} the two ends of the range separate sharply, to a median of
$0.62\times$ at $6$~dB, below SB-ECC's cost on all fifteen codes, against
$163\times$ at $2$~dB. The accuracy of Table~\ref{tab:main} is therefore free
at high $E_b/N_0$ and expensive at low. The ordering is not incidental:
SB-ECC's per-frame cost is capped by its budget, whereas ours is the width of
a beam that grows when the received word is ambiguous, and it is exactly at
low $E_b/N_0$ that certification fails. Where certification is high the guarantee is nearly free, as on BCH$(63,45)$ at $6$~dB, where proving optimality raises
the cost from $0.5$ to $0.7$ evaluations, and where it is low it is
expensive. Because our computation adapts to the input, a single average is
not comparable to a fixed budget either: on LDPC$(121,80)$ at $4$~dB the mean
is $29.0$ but $p_{50}=1$, so the mean is set by a tail with $p_{99}=141$, and
we report percentiles for every cell (Appendix~\ref{app:full-results}).

\paragraph{Accuracy becomes free as certification saturates.}
Certified frames stop at the same point regardless of $\tau$, and because the
residual errors live almost entirely in the uncertified fraction, lowering
$\tau$ deepens the search exactly where it can still help. On LDPC$(121,80)$,
moving from $\tau=10^{-1}$ to $10^{-3}$ costs $20.4\times$ at $2$~dB
($\mathrm{cert}=8.9\%$) for a gain of $+0.84$, but is free at $5$~dB
($\mathrm{cert}=94.3\%$), where the average NFE is unchanged to three
significant figures while $-\ln(\mathrm{BER})$ improves by $1.61$, a
$5\times$ reduction in BER at identical cost. Grouping all $75$ cells by
certification rate, the median cost of that move is $1.06\times$ above $95\%$
certification and $8.74\times$ below $50\%$ (Appendix~\ref{app:tau}). CAL-ECC is thus a family indexed by $\tau$, an operating characteristic no fixed-budget decoder offers.

\section{Discussion}
\label{sec:discussion}

\textbf{Conclusion.}
We built a decoder that knows when it is optimal. Absorbing-state discrete
diffusion commits coordinates one at a time, which turns generation into a
search tree; a beam search over that tree supplies candidates, and an LP dual
bound and a distance-based test, both derived from the code, certify per frame that the incumbent is the global ML solution. Empirically the two certificates are complementary in the way their
derivations predict, certification behaves as a per-frame reliability signal,
and on the codes with an exact ML reference the family of curves indexed by
$\tau$ collapses onto it, on BCH$(63,51)$ at an average of $0.6$ function evaluations per frame.

\textbf{Limitations.}
Three limitations remain, and they share a cause. \emph{(i)} Certification degrades at low $E_b/N_0$, where the incumbent weight
exceeds $\tilde d$ and voids Certificate~M, and below about $4$~dB our computation exceeds that of a budget-limited decoder (Section \ref{subsec:exp-nfe}), reaching $591\times$ on LDPC$(121,60)$ at $2$~dB. At $6$~dB the same
comparison is $0.62\times$ at $\tau=10^{-3}$. \emph{(ii)}
Polar$(128,64)$ is the one cell where a prior decoder wins, and it is also the
cell where both certificates are weakest. \emph{(iii)} The gap to the ML
ceiling grows with block length, reaching $4.23\times$ on LDPC$(121,80)$.
Where certification is high the method is accurate, cheap and able to certify
its own output, and where it is low, none of the three holds.

\textbf{Outlook.}
The most direct extension is a tighter relaxation: the Feldman polytope also
constrains satisfied checks, which is exactly the structure missing on
Polar$(128,64)$ and on the BCH codes (Appendix~\ref{app:lpgap}), while on
sparse $H$ the exact LP optimum would certify a further $2.7$--$7.3\%$ of
frames at $4$~dB. Second, even loose lower bounds on $d_{\min}$ would extend
Certificate~M to LDPC$(121,\cdot)$, MacKay and CCSDS. Third, a receiver that
knows which $10$--$40\%$ of its frames are uncertified can route exactly those
to retransmission or to a more expensive decoder; how a certificate, a proof
on a subset of frames, should be combined with the calibrated estimate of
soft-output GRAND~\citep{yuan2023sogrand} is an open question.

\newpage

\subsection*{Reproducibility statement}
The decoder is stated in full as Algorithm~\ref{alg:beam}, and every claim of
Section~\ref{sec:method} is proved in Appendix~\ref{app:proofs}, with the
numerical tolerances under which the certificates are applied in
Appendix~\ref{app:numerics}. Architecture, training and inference
hyperparameters, the measurement protocol and the baselines' stopping rules are
given in Appendix~\ref{app:impl}, and the hardware in the same appendix.
Appendix~\ref{app:code-identity} verifies that our parity-check matrices define
the same codes as the published references by $\mathbb{F}_2$ row-space
equality, and lists the distance bounds used by Certificate~M with their
sources (Table~\ref{tab:dmin}), since an incorrect bound would make the
certificate unsound. Frames are generated from a seeded generator, so every
cell is reproducible from its seed. Code will be
released upon publication.


\subsubsection*{Acknowledgments}
This work was supported by Global - Learning \& Academic research institution for Master's$\cdot$PhD students, and Postdocs (G-LAMP) Program of the National Research Foundation of Korea (NRF) grant funded by the Ministry of Education (No. RS-2025-25442252), the National Research Foundation of Korea (NRF) grant funded by the Korean government (MSIT) (Grant No. RS-2022-NR069041, RS-2025-00561169), Institute of Information \& communications Technology Planning \& Evaluation(IITP) under the Leading Generative AI Human Resources Development(IITP-2027-RS-2026-25546026) grant funded by the Korea government(MSIT), and the Ewha Womans University Research Grant of 2026.

\bibliography{iclr2027_conference}
\bibliographystyle{iclr2027_conference}

\clearpage
\appendix

\section{Deferred Proofs}
\label{app:proofs}

\subsection{Proof of Lemma~\ref{lem:ml}}
\label{app:proof-ml}

\begin{proof}
	ML decoding minimizes $\|\mathbf{y}-\mathrm{bpsk}(\mathbf{c})\|_2^2$ over $\mathbf{c}\in\mathcal{C}$.
	Coordinatewise, $\mathrm{bpsk}(c_j)\in\{\pm1\}$ and
	$(y_j+1)^2-(y_j-1)^2=4y_j$. The hard decision $r_j=\mathbbm{1}[y_j<0]$
	selects whichever modulation value is closer to $y_j$, so flipping it
	(i.e.\ setting $e_j=1$, where $\mathbf{e} = \mathbf{r} \oplus \mathbf{c}$) increases the squared distance
	by exactly $4|y_j|$. Summing over coordinates,
	\[
	\|\mathbf{y}-\mathrm{bpsk}(\mathbf{c})\|_2^2
	=\|\mathbf{y}-\mathrm{bpsk}(\mathbf{r})\|_2^2+\sum\nolimits_{j} 4|y_j|\,e_j .
	\]
	Substituting $|\ell_j|=2|y_j|/\sigma^2$, i.e.\ $4|y_j|=2\sigma^2|\ell_j|$,
	gives $\|\mathbf{y}-\mathrm{bpsk}(\mathbf{c})\|_2^2=\mathrm{const}+2\sigma^2\,\Phi(\mathbf{e})$.
	Neither the constant nor the positive factor $2\sigma^2$ affects the
	minimizer. Finally $\mathbf{c}\in\mathcal{C}\iff H\mathbf{e}^{\mathsf{T}}=\bm{s}$ by the
	syndrome identity $H\mathbf{e}^{\mathsf{T}}=H\mathbf{r}^{\mathsf{T}}\oplus H\mathbf{c}^{\mathsf{T}}$,
	so minimizing the squared distance over $\mathcal{C}$ coincides with
	minimizing $\Phi$ over $\{\mathbf{e}:H\mathbf{e}^{\mathsf{T}}=\bm{s}\}$.
\end{proof}

\subsection{Proof of Proposition~\ref{prop:list-ml}}
\label{app:proof-list}

\begin{proof}
	By Lemma~\ref{lem:ml}, for any $\mathbf{e}$ with $H\mathbf{e}^{\mathsf{T}}=\bm{s}$ the
	negative log-likelihood of $\mathbf{e}$ is an affine, strictly increasing function of
	$\Phi(\mathbf{e})$ with coefficients independent of $\mathbf{e}$. The set
	$\mathcal{L}^{\mathrm{valid}}$ collects exactly the elements of
	$\mathcal{L}_\theta$ satisfying that constraint, and \eqref{eq:select}
	returns the minimizer of $\Phi$ over it. Hence $\hat{\mathbf{e}}$ maximizes the
	likelihood over $\mathcal{L}_\theta\cap\{\mathbf{e}:H\mathbf{e}^{\mathsf{T}}=\bm{s}\}$, which
	is the ML estimate restricted to $\mathcal{L}_\theta$.
\end{proof}

\subsection{Proof of Theorem~\ref{thm:dmin}}
\label{app:proof-dmin}
\begin{proof}
	Let $\mathbf{e}'\neq \mathbf{e}$ be valid. Then $\mathbf{c} = \mathbf{e} \oplus \mathbf{e}'$ satisfies $H\mathbf{c}^{\mathsf{T}}=\bm 0$, so $\mathbf{c}$ is a nonzero codeword and $|\mathrm{supp}(\mathbf{c})|\ge d_{\min}\ge\tilde d$. 
	Since
	$\mathrm{supp}(\mathbf{e}')=(F\setminus\mathrm{supp}(\mathbf{c}))\cup(\mathrm{supp}(\mathbf{c})\setminus F)$,
	\[
	\Phi(\mathbf{e}')-\Phi(\mathbf{e})
	=\!\!\sum_{j\in\mathrm{supp}(\mathbf{c})\setminus F}\!\!|\ell_j|
	\;-\!\!\sum_{j\in\mathrm{supp}(\mathbf{c})\cap F}\!\!|\ell_j| .
	\]
	As $|\mathrm{supp}(\mathbf{c})\cap F|\le f$ we get
	$|\mathrm{supp}(\mathbf{c})\setminus F|\ge\tilde d-f$, so the first sum is at least
	$T_{\tilde d-f}(F)$ and the second at most $\Phi(\mathbf{e})$. Hence
	$\Phi(\mathbf{e}')-\Phi(\mathbf{e})\ge T_{\tilde d-f}(F)-\Phi(\mathbf{e})\ge0$.
\end{proof}

\subsection{Beam Exhaustion: Construction and Proof of Proposition~\ref{prop:beam-exhaust}}
\label{app:proof-beam}
\label{app:beam-exhaust}
Section~\ref{subsec:beam} introduced $\Lambda_{\mathrm{B}}$; formally, let $b$ be a surviving partial assignment, $C(b)\subseteq[n]$ its committed coordinates, and
$F(b)=\{j\in C(b):b_j=1\}$ the coordinates it has already flipped. Let $\mathcal{R}(b)=\{r:\bigoplus_{j\in C(b)\cap\mathrm{supp}(H_r)}b_j\neq s_r\}$ be the checks whose committed parity already disagrees with the syndrome; a check that agrees is completed at no cost by zeros and is excluded. Define
\begin{equation}
	\Lambda_{\mathrm{B}}(b)
	=\underbrace{\sum_{j\in F(b)}|\ell_j|}_{\text{committed cost}}
	\;+\;
	\underbrace{\max_{r\in\mathcal{R}(b)}\ \min_{j\in\mathrm{supp}(H_r)\setminus C(b)}|\ell_j|}_{\text{minimum additional cost}} ,
	\label{eq:lambda-b}
\end{equation}
with the second term $0$ when $\mathcal{R}(b)=\emptyset$ and $+\infty$ when some $r\in\mathcal{R}(b)$ has no free coordinate left, the minimum over an empty set, and $\Lambda_{\mathrm{B}}=\min_{b}\Lambda_{\mathrm{B}}(b)$ over the surviving beam.

\begin{proposition}[Beam exhaustion]
	\label{prop:beam-exhaust}
	Let $\Pinc$ be the cost of the best valid candidate found so far.
	If $\Pinc\le\Lambda_{\mathrm{B}}$, no continuation of any surviving
	assignment yields a valid candidate of cost below $\Pinc$, and
	returning that candidate is therefore no worse than running the beam to
	completion.
\end{proposition}
\begin{proof}
	Fix a surviving $b$ and let $\mathbf{e}$ be any completion of $b$ with
	$H\mathbf{e}^{\mathsf{T}}=\bm{s}$. Since $\Phi$ is additive and nonnegative,
	$\Phi(\mathbf{e})\ge\sum_{j\in F(b)}|\ell_j|$. Suppose
	$\mathcal{R}(b)\neq\emptyset$ and pick any $r\in\mathcal{R}(b)$. Because the committed parity of $b$ on check $r$ disagrees with $s_r$ while $\mathbf{e}$ satisfies it, $\mathbf{e}$ must flip an odd number, hence at least one, of the coordinates of $\mathrm{supp}(H_r)$ outside $C(b)$, and the cheapest such flip
	costs $\min_{j\in\mathrm{supp}(H_r)\setminus C(b)}|\ell_j|$. This holds for
	every $r\in\mathcal{R}(b)$, hence for the maximizing one, and that cost is
	disjoint from the committed sum. Therefore
	$\Phi(\mathbf{e})\ge\Lambda_{\mathrm{B}}(b)$, and if $\mathcal{R}(b)=\emptyset$ the inequality is immediate. Taking the minimum over surviving $b$, every valid
	completion of every surviving assignment has cost at least $\Lambda_{\mathrm{B}}$, so $\Pinc\le\Lambda_{\mathrm{B}}$ implies
	that no such completion improves on the incumbent.
\end{proof}

\paragraph{Scope of the bound.}
The quantifier ranges over completions of \emph{surviving} assignments only. A
branch removed earlier by \eqref{eq:beam-prune} is not a completion of any
surviving $b$, and its descendants are therefore unconstrained by
$\Lambda_{\mathrm{B}}$. Since $\hat{\mathbf{e}}_{\mathrm{ML}}$ may lie in such a discarded
subtree, $\Lambda_{\mathrm{B}}$ is \textbf{not} a lower bound on $\Phi^\star$
and beam exhaustion carries no claim of global optimality. This is why we
report beam-exhaustion stops separately from certified stops. Conflating them
would overstate the certification rate, by up to $70$ percentage points at
$2$~dB, where Certificate~M is voided by $f\ge\tilde d$ and beam exhaustion
accounts for the majority of stops. $\Lambda_{\mathrm{B}}$ is moreover defined only while the surviving beam is nonempty, and the test requires an incumbent: a check whose support is fully committed but unsatisfied leaves $\Lambda_{\mathrm{B}}=\infty$, so Algorithm~\ref{alg:beam} guards the comparison with $\Pinc<\infty$. A frame whose beam has emptied, or which ends without a valid candidate, is handled by \eqref{eq:select} directly.

\subsection{Numerical Safety}
\label{app:numerics}

The dual ascent accumulates in double precision, whereas the reliabilities
themselves are single precision, so the dominant error is the representation
of $|\ell_j|$ rather than the accumulation. The two tests are applied with a
margin $\varepsilon=10^{-4}$, but in opposite directions, because their
structure differs. Certificate~M is a genuine inequality, so we tighten it to
$\Phi(\mathbf{e})\le T_{\tilde d-f}(F)-\varepsilon$; this can only forgo a valid
certification, never issue a false one. Certificate~D is an equality test in
disguise: since $\Lambda_{\mathrm{D}}\le\Phi^\star\le\Pinc$ analytically,
$\Pinc\le\Lambda_{\mathrm{D}}$ can hold only when all three coincide, and a
strict tolerance would never fire under floating point. We therefore implement
it as $\Pinc\le\Lambda_{\mathrm{D}}+\varepsilon$, which is permissive by
$\varepsilon$: a frame certified by D satisfies
$\Pinc\le\Phi^\star+\varepsilon$, so its output is $\varepsilon$-optimal in
likelihood cost rather than exactly optimal. Two distinct valid candidates
differ in $\Phi$ by the signed reliability weight of a nonzero codeword, which
can in principle fall below $\varepsilon$; on such a frame the two candidates
are indistinguishable in likelihood to one part in $10^{4}$, and the exhaustive
check on BCH$(31,16)$ in Appendix~\ref{app:code-identity} finds no frame on
which the tolerance changes the certified output at any tested $E_b/N_0$.

\section{Decoding Algorithm}
\label{app:algorithm}

Algorithm~\ref{alg:beam} states the decoder in full. Two points are worth
highlighting. First, the dual bound $\Lambda_{\mathrm{D}}$ is computed
\emph{once}, before any forward pass, and reused at every step. It lower-bounds
$\Phi^\star$, a constant of the frame rather than of the search, so the
per-step overhead of Certificate~D is a single scalar comparison. Second, the
two stopping tests are kept separate because their guarantees differ.
Lines~\ref{ln:cert}--\ref{ln:certret} implement the certificates, and a frame
leaving there is provably globally ML.
Lines~\ref{ln:exh}--\ref{ln:exhret} implement beam exhaustion, where branches
removed by $\tau$-pruning lie outside the scope of $\Lambda_{\mathrm{B}}$ and
no global claim is made. We track the two exits separately throughout the paper
and report only the first as certified.

\begin{algorithm}[t]
	\caption{Adaptive list decoding with certified early termination.
		Numerical tolerances are discussed in Appendix~\ref{app:numerics}.}
	\label{alg:beam}
	\begin{algorithmic}[1]
		\Require reliabilities $|\boldsymbol{\ell}|$, syndrome $\boldsymbol{s}$, PCM $H$,
		model $p_\theta$, threshold $\tau$, width $L_{\max}$, dual rounds $K$,
		distance bound $\tilde d$ (optional).
		Throughout, $\Finc=\mathrm{supp}\,\einc$ and $f=|\Finc|$.
		\If{$\boldsymbol{s}=\boldsymbol{0}$}
		\Return $\hat{\mathbf{c}}=\mathbf{r}$, \textsc{certified}
		\Comment{$\Phi=0$ trivially optimal, zero NFE}
		\label{ln:skip0}
		\EndIf
		\State $\Lambda_{\mathrm{D}}\leftarrow$
		\Call{DualAscent}{$|\boldsymbol{\ell}|,\boldsymbol{s},H,K$}
		\Comment{\eqref{eq:lp-dual}; once, no forward pass}
		\State $\mathcal{B}\leftarrow\{[\mathsf{m}]^n\}$;\quad
		$\Pinc\leftarrow\infty$;\quad $\einc\leftarrow\bot$
		\While{$\mathcal{B}\neq\emptyset$ and uncommitted coordinates remain}
		\State $P\leftarrow p_\theta(\cdot\mid\mathcal{B},|\boldsymbol{\ell}|,\boldsymbol{s},t)$,
		\ \ $t=|\mathcal{M}|/n$
		\Comment{NFE $\mathrel{+}=|\mathcal{B}|$}
		\ForAll{$b\in\mathcal{B}$}
		\Comment{incumbent; no extra forward pass}
		\State $\mathbf{e}^{\mathrm{part}}\leftarrow b$ with masks set to $0$;\ \
		$\mathbf{e}^{\mathrm{greedy}}\leftarrow b$ with masks set to $\arg\max P$
		\State update $(\einc,\Pinc)$ with whichever of the two
		is valid and has smaller $\Phi$
		\EndFor
		\If{$\Pinc\le\Lambda_{\mathrm{D}}$ \textbf{or}
			$\bigl(\tilde d\text{ given},\ f<\tilde d,\text{ and }
			\Pinc\le T_{\tilde d-f}(\Finc)\bigr)$}
		\label{ln:cert}
		\State \Return $\hat{\mathbf{c}} = \mathbf{r} \oplus\einc$,
		\textsc{certified}
		\Comment{global ML (Thm.~\ref{thm:global-ml})}
		\label{ln:certret}
		\EndIf
		\State $\Lambda_{\mathrm{B}}\leftarrow\min_{b\in\mathcal{B}}
		\Lambda_{\mathrm{B}}(b)$
		\Comment{\eqref{eq:lambda-b}; surviving beam only}
		\If{$\Pinc<\infty$ \textbf{and} $\Pinc\le\Lambda_{\mathrm{B}}$}
		\label{ln:exh}
		\State \Return $\hat{\mathbf{c}} = \mathbf{r} \oplus\einc$,
		\textsc{uncertified}
		\Comment{beam exhausted (Prop.~\ref{prop:beam-exhaust})}
		\label{ln:exhret}
		\EndIf
		\State pick the uncommitted coordinate of highest mean confidence and
		branch every $b\in\mathcal{B}$ on its two values
		\label{ln:branch}
		\State prune by \eqref{eq:beam-prune} and keep at most $L_{\max}$
		\label{ln:prune}
		\State apply unit propagation to exhaustion
		\Comment{forced commitments; lossless, no score change}
		\label{ln:up}
		\EndWhile
		\State $\mathcal{L}^{\mathrm{valid}}\leftarrow
		\{\mathbf{e} \in\mathcal{B}:H\mathbf{e}^{\mathsf{T}}=\boldsymbol{s}\}
		\cup(\{\einc\}\setminus\{\bot\})$
		\Comment{beam candidates and the incumbent}
		\If{$\mathcal{L}^{\mathrm{valid}}\neq\emptyset$}
		\State $\hat{\mathbf{e}} \leftarrow
		\arg\min_{\mathbf{e} \in\mathcal{L}^{\mathrm{valid}}}\Phi(\mathbf{e})$
		\ElsIf{$\mathcal{B}\neq\emptyset$}
		\State $\hat{\mathbf{e}} \leftarrow\arg\min_{\mathbf{e}\in\mathcal{B}}\Phi(\mathbf{e})$
		\Comment{fallback: no valid candidate found}
		\Else
		\State $\hat{\mathbf{e}} \leftarrow\boldsymbol{0}$
		\Comment{empty beam, no incumbent: output the hard decision}
		\EndIf
		\State \Return $\hat{\mathbf{c}} =\mathbf{r}\oplus\hat{\mathbf{e}}$,
		\textsc{uncertified}
	\end{algorithmic}
\end{algorithm}
\paragraph{Dual ascent.}
Starting from slacks $u_j\leftarrow|\ell_j|$, each round picks $r^\ast=\arg\max_{r\in\mathcal{R}}\min_{j\in\mathrm{supp}(H_r)}u_j$ over the unsatisfied checks only, removes it from $\mathcal{R}$, raises $y_{r^\ast}$ by
$v=\min_{j\in\mathrm{supp}(H_{r^\ast})}u_j$, and decrements $u_j$ by $v$ on
that support. Since $\mathbf{u}\ge\bm 0$ is maintained, feasibility holds
throughout. The ascent returns a feasible dual rather than the optimum of
\eqref{eq:lp-dual}, a choice made for batched tensor execution
(Appendix~\ref{app:lpgap}).

\paragraph{Cost and batching.}
The dual ascent is $O(K\,|\mathcal{R}|\,w)$ for maximum row weight $w$ and runs
once. Certificate~M is one sort, $O(n\log n)$, re-evaluated only when the
incumbent changes, and $\Lambda_{\mathrm{B}}$ is a masked minimum over check
supports. All three are negligible against a transformer forward pass over
$|\mathcal{B}|$ beams, accounting in profiling for under $3\%$ of wall-clock
time, with the forward pass at $94$--$96\%$. Frames are decoded in parallel.
Because a frame's mask pattern is shared by all of its beams, the coordinate
chosen for branching is a per-frame quantity and a frame's beams always
terminate at the same step. Certified frames are removed from the batch as they
finish, which shrinks the working tensor as decoding proceeds.

\section{Implementation Details}
\label{app:impl}

\paragraph{Architecture.}
The generator uses the CrossMPT backbone~\citep{park2024crossmpt} with
$N_{\mathrm{dec}}=6$ decoder layers and $d_{\mathrm{model}}=128$, matching the
baselines so that comparisons isolate the decoding method. The output head
produces logits over $\mathcal{A}=\{0,1,\mathsf{m}\}$, and the $\mathsf{m}$ logit
is masked to $-\infty$ at inference, since no coordinate is ever predicted to
remain absorbed. The noise level enters through $\sigma(t)$ with $t=|\mathcal{M}|/n$, the actual fraction of uncommitted coordinates. Because unit propagation commits several coordinates per step (Section~\ref{subsec:beam}), $t$ is computed from the mask count rather than from the step index.

\paragraph{Training.}
Adam~\citep{kingma2015} with initial learning rate $10^{-4}$ and cosine
annealing to $5\times10^{-6}$, for $4000$ epochs of $500$ iterations at batch
size $128$, or $2\times10^{6}$ steps and $2.56\times10^{8}$ samples. Training
samples are drawn on the fly: we transmit the all-zero codeword, without loss
of generality by the linearity of the code and the symmetry of the channel,
over AWGN at an $E_b/N_0$ drawn uniformly from $\{2,\dots,7\}$~dB, and draw the
mask fraction $t$ uniformly on $[0,1]$. The loss is the cross-entropy over
masked coordinates \eqref{eq:train}; the syndrome enters only as conditioning
and as the exact verifier, never as a differentiable penalty. All models are
trained with a single seed ($42$). Training masks are drawn uniformly at random, whereas at inference the mask follows the confidence order and unit propagation, a mismatch shared with
masked diffusion generally; we do not correct for it.

\paragraph{Training budget.}
The baselines follow their published schedules, which differ from ours and from
each other: $10^{6}$ steps for ECCT and CrossMPT, $1.5\times10^{6}$ for SB-ECC
and $2\times10^{6}$ for DDECC, against our $2\times10^{6}$, all at batch size
$128$. Our generator therefore sees twice the optimisation of the two
single-shot baselines. The ablation of Section~\ref{subsec:exp-ablation} shows
that this is not where our margin comes from: decoded along a single greedy
trajectory, the same trained model is behind the strongest prior decoder on
three of the four cells examined, by up to $1.94$ in $-\ln(\mathrm{BER})$, and
the entire margin over prior work appears only once the beam search is added.
A generator trained for longer is not a better decoder; it is a better
branching heuristic.

\paragraph{Certificate parameters.}
The dual ascent is capped at $K=12$ rounds. A sweep over $K\in\{4,6,8,12,\infty\}$ showed that $K=12$ recovers the untruncated greedy bound on every code at $4$ and $6$~dB with no loss in certification rate. The
greedy bound is feasible but not optimal for \eqref{eq:lp-dual}; its gap to
the exact optimum, and the certification that gap costs, are reported in
Appendix~\ref{app:lpgap}. Truncation is safe at any round, as every
intermediate $\mathbf{y}$ is dual feasible. Both tests use margin $\varepsilon=10^{-4}$, applied as Appendix~\ref{app:numerics} describes. The distance bounds $\tilde d$ are the fixed
code invariants of Table~\ref{tab:dmin}, and where no bound is available
Certificate~M is disabled.

\paragraph{Measurement protocol.}
Each cell is decoded until $500$ frame errors or the frame cap, which is $10^{8}$ for
BCH and polar codes and $10^{9}$ for the sparse-$H$ families. Frames are
generated in fixed-size batches with a seeded generator, so the three $\tau$ values within a cell see an identical frame sequence and their comparison is paired.

\paragraph{Baseline stopping rules.}
DDECC and SB-ECC both admit the standard syndrome stopping rule, which halts once the
current hard decision satisfies every parity check, and we enable it for both,
so that the NFE we report for them is realised cost rather than budget. For
SB-ECC this matters considerably. Its nominal budget is $10$ function
evaluations, while the realised mean ranges from $0.36$ to $7.56$ across our
cells (Table~\ref{tab:app-sbnfe}). Using the budget would overstate our
advantage by up to an order of magnitude.

\paragraph{Hardware.}
Measurements were made on a single NVIDIA RTX PRO 6000 Blackwell GPU with 96 GB
of memory. Batch size affects memory and throughput but not results. When a
batch exceeds available memory the run is retried at a smaller size, which
changes the frame sequence and therefore perturbs each cell within its sampling
error, but leaves the paired comparison across $\tau$ intact. Wall-clock
timings are reported for reference only, as several cells were measured on
shared devices, and no conclusion in this paper depends on them.

\section{Code Identity Verification}
\label{app:code-identity}

Our parity-check matrices come from the code database distributed with
ECCT~\citep{choukroun2022ecct} and CrossMPT~\citep{park2024crossmpt}, which all
model-free baselines in our comparison also use, so every decoder here is
evaluated on byte-identical matrices. Comparing against published ML references
additionally requires that these matrices define the same codes as the
references. Two claims depend on it. The ML references of
Section~\ref{subsec:exp-ml} are meaningless if the codes differ, and more
seriously, Certificate~M is unsound if $\tilde d$ exceeds the true $d_{\min}$
of the code actually being decoded (Theorem~\ref{thm:dmin}).

\paragraph{Row-space equality.}
For each code we compute the $\mathbb{F}_2$ row space of our $H$ and of the
reference by Gaussian elimination and check
$\mathrm{rank}(H_{\mathrm{ours}})=\mathrm{rank}(H_{\mathrm{ref}})=
\mathrm{rank}([H_{\mathrm{ours}};H_{\mathrm{ref}}])$, which holds if and only
if the two matrices define the same code. Where a reference is distributed as a
generator matrix we recover $H$ from its null space first. All LDPC, MacKay and
CCSDS matrices match their references.

\paragraph{Polar codes.}
No reference parity-check matrix is distributed for the polar codes, so we
verify them internally. Recovering the information set
$\mathcal{A}=\{i:F^{\otimes m}_i\in\ker H\}$ from $H$, we check two properties.
\emph{(i) Nesting}: polar codes of the same length from a common construction
are nested, $\mathcal{A}_{k_1}\subset\mathcal{A}_{k_2}$ for $k_1<k_2$, and this
holds for $(128,64)\subset(128,86)\subset(128,96)$ and for
$(64,32)\subset(64,48)$, so the codes at each length come from one construction
rather than from mixed sources. \emph{(ii) Construction}: all five information
sets coincide exactly with the Bhattacharyya construction in bit-reversed index
order. That ordering is nearly invariant to the design SNR, so this identifies
the construction family but not the design point.

\begin{table}[t]
	\centering\small
	\caption{Distance bounds $\tilde d\le d_{\min}$ used by Certificate~M.
		Codes not listed have no proven bound available and the certificate is
		disabled for them.}
	\label{tab:dmin}
	\begin{tabular}{lrl}
		\toprule
		Code & $\tilde d$ & Source \\
		\midrule
		BCH$(31,16)$   & $7$  & exhaustive over $2^{16}$ (exact $d_{\min}$) \\
		BCH$(63,36)$   & $11$ & BCH bound \\
		BCH$(63,45)$   & $7$  & BCH bound \\
		BCH$(63,51)$   & $5$  & BCH bound \\
		Polar$(64,32)$ & $8$  & $\min_{i\in\mathcal{A}}2^{\mathrm{wt}(i)}$ \\
		Polar$(64,48)$ & $4$  & $\min_{i\in\mathcal{A}}2^{\mathrm{wt}(i)}$ \\
		Polar$(128,64)$& $8$  & $\min_{i\in\mathcal{A}}2^{\mathrm{wt}(i)}$ \\
		Polar$(128,86)$& $8$  & $\min_{i\in\mathcal{A}}2^{\mathrm{wt}(i)}$ \\
		Polar$(128,96)$& $4$  & $\min_{i\in\mathcal{A}}2^{\mathrm{wt}(i)}$ \\
		LDPC$(49,24)$  & $8$  & exhaustive over $2^{24}$ (exact $d_{\min}$) \\
		\bottomrule
	\end{tabular}
\end{table}

\paragraph{Distance bounds and end-to-end validation.}
The values of $\tilde d$ are code invariants, not tuned parameters.
For BCH codes we use the BCH bound~\citep{bose1960,macwilliams1977}, and for
polar codes $\min_{i\in\mathcal{A}}2^{\mathrm{wt}(i)}$ evaluated on the
recovered information set~\citep{arikan2009,hussami2009}, which equals $d_{\min}$ when the information set is closed under the partial order induced by polarisation, as it is for the Bhattacharyya construction verified above~\citep{hussami2009}. Table~\ref{tab:dmin} lists them
with their sources, and the polar values were recomputed independently from our
matrices and agree with the values used by the decoder in all five cases. We
additionally cross-check every $\tilde d$ against an upper bound obtained by
random codeword sampling, and find the two coincide on most of the ten codes
for which a bound is defined, confirming that it is not merely valid but often
tight. As a direct check of the
whole certification path we compare against exhaustive ML on BCH$(31,16)$,
where the $2^{16}$ codebook is enumerable, verifying that
$\Lambda_{\mathrm{D}}\le\Phi^\star$ for every frame and that every frame
certified by \eqref{eq:stop} attains $\Phi^\star$. Both hold with zero
violations at all tested $E_b/N_0$. A weaker but far larger check is available
on the remaining codes. Certified frames output the ML solution, so their error
rate should track that of ML rather than that of our decoder, and over the
millions of frames of Section~\ref{subsec:exp-cert} we find $96.7$--$100\%$ of
all frame errors falling on uncertified frames, consistent with this prediction.

\section{Extended Related Work}
\label{app:related}

\paragraph{Model-free decoders.} Later work on ECCT refined the masking~\citep{park2023}, learned the code jointly with the decoder~\citep{choukroun2024b}, pursued zero-shot generalization across codes~\citep{choukroun2024a}, and cut cost through quantization~\citep{levy2024} or state-space backbones~\citep{cohen2025}.

\paragraph{One-step generative decoders.}
ECCFM~\citep{lei2025} distils the denoising trajectory into a single forward
pass, the opposite end of the latency axis from ours: it fixes the cost at one
evaluation and makes no statement about the output, whereas our cost adapts per
frame and falls below one evaluation exactly on the frames it can certify. The
two are complementary rather than competing, and a one-step generator could
serve as our branching heuristic; its implementation is not public at the time
of writing, so we do not compare empirically.

\paragraph{Guessing-based universal decoding.}
GRAND~\citep{duffy2019} inverts the usual formulation. Instead of exploiting
code structure it guesses the noise the channel added, in decreasing order of
likelihood, and tests each hypothesis for codebook membership. If the guessing
order is exactly the likelihood order, the first hypothesis that passes the
membership test \emph{is} the ML codeword, so GRAND is an ML decoder by
construction. That guarantee is stronger than ours, since it holds on every frame.
Its cost, however, is the number of guesses, which is exponential in the worst
case. The practical variants relax the order to make it tractable, and in doing
so give the guarantee up. ORBGRAND~\citep{duffy2021} orders by reliability
rank rather than by likelihood, and SRGRAND~\citep{duffy2022} quantizes the
reliabilities. Our setting is the mirror image. The learned branching rule of
Section~\ref{subsec:beam} is emphatically \emph{not} a likelihood ordering, so
membership alone certifies nothing. The certificates of
Section~\ref{subsec:cert} recover a per-frame ML guarantee on a search whose
order is not controlled, at the price that they apply to a fraction of frames
rather than to all of them.

Soft-output GRAND~\citep{yuan2023sogrand} is the closest prior work to the
reliability signal of Section~\ref{subsec:exp-cert}. It accompanies each
decoding with an estimate of the probability that the decoding is correct, at
almost no additional cost, and proposes the same two uses we do, suppressing undetected errors and deciding whether to request a retransmission. The two
signals are different objects. A soft output is a calibrated \emph{forecast},
evaluated by how well its probabilities match observed frequencies. A
certificate is a \emph{proof}, and a certified frame is ML-optimal as a matter
of theorem. Ours is therefore the stronger statement on the frames where it is
available and no statement at all elsewhere, whereas a soft output is available
everywhere and exact nowhere.

\paragraph{Optimality certificates.}
LP decoding~\citep{feldman2005} relaxes ML decoding over the codeword polytope
and possesses the ML certificate property, so an integral optimum of the relaxation is provably the ML codeword. This is the canonical statement that
solving and certifying differ in difficulty. Later work sharpened the
relaxation with adaptive cuts and characterized its
pseudocodewords~\citep{taghavi2008,vontobel2007}. In the ordered-statistics
literature, sufficient conditions comparing a candidate's likelihood cost
against the minimum distance terminate reprocessing
early~\citep{fossorier1995,kaneko1994}. In both lines the certificate is tied to the decoder's own structure, the exactly solved relaxation in one and the prescribed reprocessing order in the other, which is what makes the bound computable. To our knowledge no
neural decoder reports per-frame optimality certificates. We transplant these
ideas onto a learned search and show that the resulting certification rate is
itself informative. It is predicted by code structure and predictive of correctness.

\paragraph{Adaptive computation.}
Early exit and input-adaptive computation are established in deep
learning~\citep{teerapittayanon2016,graves2016} and appear in decoding as
syndrome-based stopping. DDECC, SB-ECC and classical BP all halt once the
current hard decision satisfies every parity check. That test is a necessary
condition for a valid codeword but says nothing about optimality, and a
syndrome-consistent estimate may be far from the ML solution. Our stopping rule
is a sufficient condition for global optimality, and we report both accountings
side by side so that the cost of the stronger guarantee is visible
(Section~\ref{subsec:exp-nfe}).
\section{Full Results}
\label{app:full-results}
Table~\ref{tab:nfe} is the subset of cells discussed in Section~\ref{subsec:exp-nfe}.
\begin{table}[t]
	\centering\setlength{\tabcolsep}{4pt}\small
	\caption{Computation, measured under two stopping rules.
		$\mathrm{NFE}_{\mathrm{sd}}$ counts evaluations until the first
		syndrome-consistent candidate, which is the rule SB-ECC uses, whose realised cost we measure rather than assume. $\mathrm{NFE}$ is our actual cost, which continues until optimality is \emph{proved}. Sorted by $\mathrm{cert}\%$. Under SB-ECC's own rule we are cheaper where certification is high, except on dense $H$, where reaching a syndrome-consistent candidate is itself slow. Table~\ref{tab:app-nfe} gives every cell.}
	\label{tab:nfe}
	\begin{tabular}{lr rrrr rr r}
		\toprule
		\multirow{2}{*}{Code} & \multirow{2}{*}{dB}
		& \multicolumn{4}{c}{$\mathrm{NFE}_{\mathrm{sd}}$}
		& \multicolumn{2}{c}{ours, $\tau{=}10^{-3}$} & \multirow{2}{*}{$\mathrm{cert}\%$} \\
		\cmidrule(lr){3-6}\cmidrule(lr){7-8}
		& & SB-ECC & $10^{-1}$ & $10^{-2}$ & $10^{-3}$ & NFE & $p_{99}$ & \\
		\midrule
		BCH$(63,45)$ & 6 & $0.55$ & $0.4$ & $0.5$ & $0.5$ & $0.7$ & $1$ & $99.8$ \\
		LDPC$(121,80)$ & 6 & $1.12$ & $0.7$ & $0.7$ & $0.7$ & $1.2$ & $1$ & $99.5$ \\
		LDPC$(121,70)$ & 6 & $1.38$ & $0.9$ & $0.9$ & $0.9$ & $1.6$ & $1$ & $99.2$ \\
		Polar$(64,48)$ & 6 & $0.47$ & $0.4$ & $0.4$ & $0.4$ & $1.9$ & $58$ & $97.3$ \\
		MacKay$(96,48)$ & 6 & $1.66$ & $0.9$ & $0.9$ & $0.9$ & $2.7$ & $55$ & $96.6$ \\
		LDPC$(121,80)$ & 5 & $1.83$ & $0.9$ & $0.9$ & $0.9$ & $5.9$ & $89$ & $94.3$ \\
		LDPC$(121,70)$ & 5 & $2.09$ & $1.0$ & $1.0$ & $1.0$ & $8.5$ & $88$ & $91.3$ \\
		BCH$(63,36)$ & 4 & $2.72$ & $7.7$ & $28.6$ & $95.7$ & $191.1$ & $3062$ & $88.8$ \\
		CCSDS$(128,64)$ & 5 & $2.65$ & $1.0$ & $1.0$ & $1.0$ & $19.9$ & $80$ & $75.4$ \\
		Polar$(128,86)$ & 5 & $1.67$ & $1.1$ & $1.1$ & $1.1$ & $31.4$ & $122$ & $73.9$ \\
		LDPC$(121,80)$ & 4 & $2.69$ & $1.6$ & $1.8$ & $3.4$ & $29.0$ & $141$ & $72.4$ \\
		Polar$(128,64)$ & 3 & $4.82$ & $37.5$ & $149.8$ & $853.5$ & $1174.2$ & $7094$ & $3.5$ \\
		LDPC$(121,60)$ & 2 & $7.56$ & $70.8$ & $824.2$ & $3395.5$ & $4470.1$ & $8998$ & $2.0$ \\
		\bottomrule
	\end{tabular}
\end{table}

Tables~\ref{tab:app-ber} and~\ref{tab:app-fer} give every measured cell in BER and FER respectively, covering $15$ codes at $5$ operating points and $3$ thresholds for CAL-ECC against four baselines.
Table~\ref{tab:app-cert} decomposes the certification rate of every cell into
its two certificates and Table~\ref{tab:app-nfe} reports the corresponding
computational cost, both at $\tau=10^{-3}$. The certification rate is nearly
independent of $\tau$, since the dual bound $\Lambda_{\mathrm{D}}$ is a
constant of the frame, varying by under one percentage point across the three
thresholds on every cell. Table~\ref{tab:app-sbnfe} gives the realised NFE of
SB-ECC under its own stopping rule, which is what
Section~\ref{subsec:exp-nfe} compares against.

Three patterns in Table~\ref{tab:app-cert} hold across all cells rather than
the nine of Table~\ref{tab:cert}. \emph{(i)} The $[\mathrm{D}]$ and
$[\mathrm{M}]$ figures are almost disjoint. Two codes are the exception, and
both are informative. LDPC$(49,24)$, the one code that is simultaneously sparse
and has an exactly known $d_{\min}$, has both certificates contributing
double-digit shares at $2$--$4$~dB ($15.1\%$ and $60.9\%$ at $4$~dB), and
Polar$(128,64)$ enters the same situation at $5$--$6$~dB, once its error weight
falls below $\tilde d$ often enough for Certificate~M to apply. Everywhere else
one of the two carries essentially all of the certification.
\emph{(ii)} Certification rises monotonically with
$E_b/N_0$ on every code without exception, from below $1\%$ at $2$~dB on the
length-$128$ polar codes to above $99\%$ at $6$~dB on BCH. \emph{(iii)} The
complement, meaning frames stopped by beam exhaustion, mirrors this. It dominates at
low $E_b/N_0$, where $f\ge\tilde d$ voids Certificate~M and the dual bound is
loose, and falls below $1\%$ at high $E_b/N_0$. Reporting beam exhaustion as
certification, as a single ``stopped early'' figure would, would overstate the
certified fraction by up to $70$ percentage points at $2$~dB.

Table~\ref{tab:app-nfe} supports the computational claims of
Section~\ref{subsec:exp-nfe}. Measured under the syndrome rule,
$\mathrm{NFE}_{\mathrm{sd}}$ is at or near $1$ on every code at $6$~dB and on all but the dense-$H$ BCH codes and Polar$(128,64)$ at $5$~dB. The distance from $\mathrm{NFE}_{\mathrm{sd}}$ to
$\mathrm{NFE}$ is the cost of proving optimality and is highly non-uniform, being
negligible on BCH at $6$~dB and a factor of $20$ or more on the length-$128$ polar
codes at $4$~dB. The $p_{99}$ figures show why an average alone is not
comparable to a fixed budget. Here $p_{50}=1$ on every cell at $E_b/N_0\ge5$~dB, so
the mean is set by a tail rather than by the typical frame.

\begin{table}[t]
	\centering\setlength{\tabcolsep}{5pt}\small
	\caption{Realised NFE of SB-ECC with its own syndrome stopping rule enabled. The nominal budget of $10$ is an upper bound that is rarely reached, and we use these values wherever SB-ECC's cost appears.}
	\label{tab:app-sbnfe}
	\begin{tabular}{l ccccc}
		\toprule
		Code & $2$~dB & $3$~dB & $4$~dB & $5$~dB & $6$~dB \\
		\midrule
		BCH$(31,16)$ & $4.41$ & $3.27$ & $2.29$ & $1.48$ & $0.86$ \\
		BCH$(63,36)$ & $6.14$ & $4.28$ & $2.72$ & $1.68$ & $0.99$ \\
		BCH$(63,45)$ & $5.17$ & $3.19$ & $1.83$ & $1.05$ & $0.55$ \\
		BCH$(63,51)$ & $5.63$ & $3.42$ & $1.70$ & $0.81$ & $0.36$ \\
		\addlinespace[2pt]
		Polar$(64,32)$ & $5.88$ & $4.50$ & $3.31$ & $2.32$ & $1.50$ \\
		Polar$(64,48)$ & $3.68$ & $2.46$ & $1.57$ & $0.91$ & $0.47$ \\
		Polar$(128,64)$ & $6.61$ & $4.82$ & $3.57$ & $2.65$ & $1.88$ \\
		Polar$(128,86)$ & $5.70$ & $3.68$ & $2.47$ & $1.67$ & $1.03$ \\
		Polar$(128,96)$ & $5.91$ & $3.56$ & $2.13$ & $1.34$ & $0.77$ \\
		\addlinespace[2pt]
		LDPC$(49,24)$ & $5.25$ & $3.87$ & $2.78$ & $1.92$ & $1.21$ \\
		LDPC$(121,60)$ & $7.56$ & $5.35$ & $3.62$ & $2.65$ & $1.86$ \\
		LDPC$(121,70)$ & $6.46$ & $4.24$ & $2.93$ & $2.09$ & $1.38$ \\
		LDPC$(121,80)$ & $5.99$ & $3.88$ & $2.69$ & $1.83$ & $1.12$ \\
		\addlinespace[2pt]
		MacKay$(96,48)$ & $5.92$ & $4.41$ & $3.31$ & $2.43$ & $1.66$ \\
		\addlinespace[2pt]
		CCSDS$(128,64)$ & $6.54$ & $4.71$ & $3.53$ & $2.65$ & $1.87$ \\
		\bottomrule
	\end{tabular}
\end{table}

\begin{table}[t]
	\centering\setlength{\tabcolsep}{2pt}\renewcommand{\arraystretch}{1.05}
	\caption{Certification across all cells at $\tau=10^{-3}$. $\mathrm{cert}$ is the fraction of frames proved globally ML, while $[\mathrm{D}]$ and $[\mathrm{M}]$ are the parts attributed to each certificate, and the remainder of $\mathrm{cert}$ is the $\mathrm{skip}_0$ fraction with $\boldsymbol{s}=\boldsymbol{0}$. Frames stopped by beam exhaustion are not counted. All figures are percentages.}
	\label{tab:app-cert}
	\resizebox{\textwidth}{!}{%
		\begin{tabular}{l ccc ccc ccc ccc ccc}
			\toprule
			\multirow{2}{*}{Code} & \multicolumn{3}{c}{$2$~dB} & \multicolumn{3}{c}{$3$~dB} & \multicolumn{3}{c}{$4$~dB} & \multicolumn{3}{c}{$5$~dB} & \multicolumn{3}{c}{$6$~dB} \\
			\cmidrule(lr){2-4}\cmidrule(lr){5-7}\cmidrule(lr){8-10}\cmidrule(lr){11-13}\cmidrule(lr){14-16}
			& cert & $[\mathrm{D}]$ & $[\mathrm{M}]$ & cert & $[\mathrm{D}]$ & $[\mathrm{M}]$ & cert & $[\mathrm{D}]$ & $[\mathrm{M}]$ & cert & $[\mathrm{D}]$ & $[\mathrm{M}]$ & cert & $[\mathrm{D}]$ & $[\mathrm{M}]$ \\
			\midrule
			BCH$(31,16)$ & $50.4$ & $1.0$ & $45.6$ & $72.5$ & $0.9$ & $62.8$ & $89.2$ & $0.5$ & $70.6$ & $97.5$ & $0.1$ & $64.6$ & $99.7$ & $0.0$ & $48.4$ \\
			BCH$(63,36)$ & $31.4$ & $0.0$ & $31.1$ & $62.4$ & $0.0$ & $60.9$ & $88.8$ & $0.0$ & $83.4$ & $98.6$ & $0.0$ & $82.5$ & $99.9$ & $0.0$ & $64.8$ \\
			BCH$(63,45)$ & $24.7$ & $0.1$ & $23.2$ & $54.5$ & $0.1$ & $49.4$ & $83.7$ & $0.1$ & $68.0$ & $97.3$ & $0.0$ & $62.8$ & $99.8$ & $0.0$ & $41.6$ \\
			BCH$(63,51)$ & $22.3$ & $0.7$ & $18.7$ & $49.9$ & $1.0$ & $39.0$ & $79.6$ & $0.7$ & $53.6$ & $95.7$ & $0.3$ & $48.2$ & $99.6$ & $0.0$ & $29.2$ \\
			\addlinespace[2pt]
			Polar$(64,32)$ & $5.1$ & $0.1$ & $4.9$ & $18.4$ & $0.1$ & $17.8$ & $45.7$ & $0.1$ & $43.1$ & $77.4$ & $0.1$ & $68.8$ & $95.4$ & $0.0$ & $72.9$ \\
			Polar$(64,48)$ & $11.6$ & $2.2$ & $7.7$ & $31.3$ & $4.0$ & $20.8$ & $60.5$ & $4.4$ & $37.5$ & $85.9$ & $2.6$ & $44.6$ & $97.3$ & $0.7$ & $33.9$ \\
			Polar$(128,64)$ & $0.3$ & $0.3$ & $0.0$ & $3.5$ & $3.2$ & $0.2$ & $18.7$ & $15.3$ & $3.3$ & $52.5$ & $31.2$ & $20.6$ & $85.2$ & $23.1$ & $57.0$ \\
			Polar$(128,86)$ & $0.6$ & $0.0$ & $0.6$ & $6.4$ & $0.1$ & $6.2$ & $31.6$ & $0.3$ & $30.0$ & $73.9$ & $0.2$ & $65.7$ & $96.6$ & $0.0$ & $70.2$ \\
			Polar$(128,96)$ & $0.4$ & $0.2$ & $0.2$ & $4.2$ & $1.9$ & $1.9$ & $21.6$ & $6.3$ & $11.9$ & $57.0$ & $8.9$ & $33.1$ & $87.7$ & $4.6$ & $43.9$ \\
			\addlinespace[2pt]
			LDPC$(49,24)$ & $31.4$ & $15.8$ & $15.3$ & $56.9$ & $19.6$ & $35.7$ & $81.3$ & $15.1$ & $60.9$ & $95.1$ & $6.3$ & $74.8$ & $99.4$ & $1.2$ & $68.0$ \\
			LDPC$(121,60)$ & $2.0$ & $2.0$ & $0.0$ & $14.3$ & $14.3$ & $0.0$ & $46.2$ & $46.1$ & $0.0$ & $82.8$ & $81.9$ & $0.0$ & $97.8$ & $92.2$ & $0.0$ \\
			LDPC$(121,70)$ & $4.9$ & $4.9$ & $0.0$ & $25.3$ & $25.2$ & $0.0$ & $63.0$ & $62.5$ & $0.0$ & $91.3$ & $88.0$ & $0.0$ & $99.2$ & $84.8$ & $0.0$ \\
			LDPC$(121,80)$ & $8.9$ & $8.9$ & $0.0$ & $35.0$ & $34.9$ & $0.0$ & $72.4$ & $70.9$ & $0.0$ & $94.3$ & $86.1$ & $0.0$ & $99.5$ & $72.9$ & $0.0$ \\
			\addlinespace[2pt]
			MacKay$(96,48)$ & $9.5$ & $9.5$ & $0.0$ & $29.3$ & $29.2$ & $0.0$ & $59.5$ & $59.2$ & $0.0$ & $85.0$ & $82.5$ & $0.0$ & $96.6$ & $85.9$ & $0.0$ \\
			\addlinespace[2pt]
			CCSDS$(128,64)$ & $1.9$ & $1.9$ & $0.0$ & $12.5$ & $12.5$ & $0.0$ & $41.0$ & $40.9$ & $0.0$ & $75.4$ & $74.7$ & $0.0$ & $94.2$ & $89.1$ & $0.0$ \\
			\bottomrule
	\end{tabular}}
	\par\smallskip\parbox{\textwidth}{\footnotesize \textit{Note.} Zeros in $[\mathrm{M}]$ for LDPC$(121,\cdot)$, MacKay and CCSDS are the designed fallback, since no proven $d_{\min}$ bound exists and the test is disabled (Appendix~\ref{app:code-identity}). Certification rises monotonically with $E_b/N_0$ on every code. All figures are percentages rounded to one decimal place. $[\mathrm{D}]$ and $[\mathrm{M}]$ are near-disjoint except on LDPC$(49,24)$, the only code that is both sparse and has an exactly known $d_{\min}$, and on Polar$(128,64)$ at $5$--$6$~dB, where the error weight finally falls below $\tilde d$ often enough for Certificate~M to apply.}
\end{table}

\begin{table}[t]
	\centering\setlength{\tabcolsep}{2.5pt}\renewcommand{\arraystretch}{1.05}
	\caption{Computation across all cells at $\tau=10^{-3}$. $\mathrm{NFE}_{\mathrm{sd}}$ counts evaluations until the first syndrome-consistent candidate, which is the stopping rule SB-ECC uses. $\mathrm{NFE}$ continues until optimality is proved, and $p_{99}$ is the $99$th percentile of the per-frame NFE.}
	\label{tab:app-nfe}
	\resizebox{\textwidth}{!}{%
		\begin{tabular}{l ccc ccc ccc ccc ccc}
			\toprule
			\multirow{2}{*}{Code} & \multicolumn{3}{c}{$2$~dB} & \multicolumn{3}{c}{$3$~dB} & \multicolumn{3}{c}{$4$~dB} & \multicolumn{3}{c}{$5$~dB} & \multicolumn{3}{c}{$6$~dB} \\
			\cmidrule(lr){2-4}\cmidrule(lr){5-7}\cmidrule(lr){8-10}\cmidrule(lr){11-13}\cmidrule(lr){14-16}
			& $\mathrm{sd}$ & NFE & $p_{99}$ & $\mathrm{sd}$ & NFE & $p_{99}$ & $\mathrm{sd}$ & NFE & $p_{99}$ & $\mathrm{sd}$ & NFE & $p_{99}$ & $\mathrm{sd}$ & NFE & $p_{99}$ \\
			\midrule
			BCH$(31,16)$ & $14.8$ & $213.8$ & $1733$ & $6.6$ & $61.4$ & $1109$ & $2.2$ & $11.8$ & $346$ & $0.8$ & $2.0$ & $22$ & $0.5$ & $0.6$ & $1$ \\
			BCH$(63,36)$ & $1599.8$ & $2521.7$ & $5728$ & $519.1$ & $963.8$ & $4432$ & $95.7$ & $191.1$ & $3062$ & $9.4$ & $17.4$ & $545$ & $1.0$ & $1.1$ & $1$ \\
			BCH$(63,45)$ & $572.3$ & $1960.6$ & $4651$ & $203.6$ & $657.0$ & $3298$ & $37.3$ & $113.8$ & $2086$ & $3.6$ & $9.7$ & $168$ & $0.5$ & $0.7$ & $1$ \\
			BCH$(63,51)$ & $70.8$ & $1439.1$ & $3798$ & $28.4$ & $421.2$ & $2503$ & $7.1$ & $65.3$ & $1300$ & $1.2$ & $6.2$ & $72$ & $0.3$ & $0.6$ & $1$ \\
			\addlinespace[2pt]
			Polar$(64,32)$ & $56.6$ & $375.8$ & $5578$ & $11.1$ & $89.1$ & $1032$ & $2.1$ & $32.6$ & $109$ & $1.0$ & $13.1$ & $59$ & $0.8$ & $3.3$ & $58$ \\
			Polar$(64,48)$ & $26.6$ & $722.7$ & $4613$ & $7.0$ & $140.3$ & $2136$ & $1.7$ & $28.9$ & $176$ & $0.7$ & $8.6$ & $61$ & $0.4$ & $1.9$ & $58$ \\
			Polar$(128,64)$ & $3985.7$ & $4650.8$ & $10581$ & $853.5$ & $1174.2$ & $7094$ & $68.1$ & $156.4$ & $2547$ & $2.7$ & $39.8$ & $89$ & $1.0$ & $12.1$ & $87$ \\
			Polar$(128,86)$ & $1347.8$ & $2835.4$ & $13151$ & $155.8$ & $447.5$ & $7212$ & $8.7$ & $93.1$ & $407$ & $1.1$ & $31.4$ & $122$ & $0.7$ & $4.7$ & $120$ \\
			Polar$(128,96)$ & $975.4$ & $3560.6$ & $11461$ & $138.7$ & $604.4$ & $6019$ & $9.2$ & $112.0$ & $761$ & $1.1$ & $51.0$ & $124$ & $0.6$ & $15.0$ & $124$ \\
			\addlinespace[2pt]
			LDPC$(49,24)$ & $77.2$ & $329.1$ & $2121$ & $17.4$ & $67.9$ & $1066$ & $2.6$ & $10.9$ & $196$ & $0.9$ & $2.4$ & $30$ & $0.7$ & $0.9$ & $1$ \\
			LDPC$(121,60)$ & $3395.5$ & $4470.1$ & $8998$ & $699.5$ & $1194.8$ & $6039$ & $45.8$ & $130.3$ & $2528$ & $1.7$ & $16.2$ & $85$ & $0.9$ & $2.7$ & $83$ \\
			LDPC$(121,70)$ & $1877.0$ & $2712.3$ & $7561$ & $240.5$ & $470.8$ & $4301$ & $9.7$ & $47.6$ & $675$ & $1.0$ & $8.5$ & $88$ & $0.9$ & $1.6$ & $1$ \\
			LDPC$(121,80)$ & $977.3$ & $1709.0$ & $6277$ & $93.6$ & $225.8$ & $2737$ & $3.4$ & $29.0$ & $141$ & $0.9$ & $5.9$ & $89$ & $0.7$ & $1.2$ & $1$ \\
			\addlinespace[2pt]
			MacKay$(96,48)$ & $245.3$ & $428.0$ & $2946$ & $24.7$ & $72.9$ & $754$ & $1.9$ & $23.1$ & $64$ & $1.0$ & $8.9$ & $56$ & $0.9$ & $2.7$ & $55$ \\
			\addlinespace[2pt]
			CCSDS$(128,64)$ & $1245.4$ & $1595.9$ & $6938$ & $123.3$ & $224.3$ & $2939$ & $4.4$ & $49.3$ & $81$ & $1.0$ & $19.9$ & $80$ & $0.9$ & $5.4$ & $80$ \\
			\bottomrule
	\end{tabular}}
	\par\smallskip\parbox{\textwidth}{\footnotesize \textit{Note.} $\mathrm{NFE}_{\mathrm{sd}}\le1$ at $6$~dB on every code and at $5$~dB on all but four, and the gap to $\mathrm{NFE}$ is the price of certification. $p_{50}\le1$ on every cell at $E_b/N_0\ge5$~dB, so the mean is set by the tail
		rather than by the typical frame. At $4$~dB the length-$128$ polar codes,
		Polar$(64,32)$, CCSDS and LDPC$(121,60)$ still have $p_{50}$ between $49$ and
		$115$, and those are the cells where the beam has not yet collapsed.}
\end{table}

\begin{sidewaystable}[p]
	\centering\setlength{\tabcolsep}{3pt}\renewcommand{\arraystretch}{1.05}
	\caption{Complete BER results, $-\ln(\mathrm{BER})$ at $E_b/N_0\in\{2,\dots,6\}$~dB, where higher is better.}
	\label{tab:app-ber}
	\textbf{\small (a) Baselines}\par\smallskip
	\resizebox{!}{0.13\textwidth}{%
		\begin{tabular}{l ccccc ccccc ccccc ccccc}
			\toprule
			\multirow{2}{*}{Code} & \multicolumn{5}{c}{ECCT} & \multicolumn{5}{c}{CrossMPT} & \multicolumn{5}{c}{DDECC} & \multicolumn{5}{c}{SB-ECC} \\
			\cmidrule(lr){2-6}\cmidrule(lr){7-11}\cmidrule(lr){12-16}\cmidrule(lr){17-21}
			& 2 & 3 & 4 & 5 & 6 & 2 & 3 & 4 & 5 & 6 & 2 & 3 & 4 & 5 & 6 & 2 & 3 & 4 & 5 & 6 \\
			\midrule
			BCH$(31,16)$ & $3.42$ & $4.49$ & $5.90$ & $7.81$ & $10.34$ & $3.93$ & $5.21$ & $7.00$ & $9.35$ & $12.36$ & $2.96$ & $3.75$ & $4.80$ & $6.13$ & $7.67$ & $4.03$ & $5.41$ & $7.29$ & $9.81$ & $12.90$ \\
			BCH$(63,36)$ & $2.70$ & $3.46$ & $4.66$ & $6.51$ & $9.18$ & $2.74$ & $3.52$ & $4.79$ & $6.75$ & $9.45$ & $2.98$ & $3.92$ & $5.30$ & $7.29$ & $9.98$ & $3.03$ & $4.04$ & $5.52$ & $7.55$ & $10.49$ \\
			BCH$(63,45)$ & $3.02$ & $3.89$ & $5.33$ & $7.46$ & $10.41$ & $3.02$ & $3.89$ & $5.31$ & $7.53$ & $10.50$ & $3.16$ & $4.18$ & $5.77$ & $7.95$ & $10.54$ & $3.29$ & $4.48$ & $6.29$ & $9.01$ & $12.54$ \\
			BCH$(63,51)$ & $3.15$ & $4.02$ & $5.43$ & $7.45$ & $10.50$ & $3.19$ & $4.14$ & $5.64$ & $7.89$ & $11.17$ & $3.10$ & $3.94$ & $5.25$ & $7.14$ & $9.25$ & $3.14$ & $4.07$ & $5.50$ & $7.64$ & $10.79$ \\
			\addlinespace[2pt]
			Polar$(64,32)$ & $3.68$ & $4.94$ & $6.63$ & $9.04$ & $11.89$ & $4.06$ & $5.59$ & $7.63$ & $10.26$ & $13.29$ & $3.80$ & $5.18$ & $7.05$ & $9.50$ & $12.23$ & $4.07$ & $5.54$ & $7.52$ & $10.15$ & $13.26$ \\
			Polar$(64,48)$ & $3.44$ & $4.57$ & $6.19$ & $8.28$ & $10.97$ & $3.53$ & $4.74$ & $6.41$ & $8.56$ & $11.22$ & $3.34$ & $4.44$ & $6.00$ & $8.00$ & $10.40$ & $3.57$ & $4.85$ & $6.59$ & $8.62$ & $11.23$ \\
			Polar$(128,64)$ & $2.55$ & $3.34$ & $4.84$ & $7.33$ & $10.72$ & $2.65$ & $3.69$ & $5.64$ & $8.71$ & $13.24$ & $3.60$ & $5.51$ & $8.31$ & $12.06$ & $15.73$ & $3.52$ & $5.24$ & $7.77$ & $11.29$ & $15.43$ \\
			Polar$(128,86)$ & $3.28$ & $4.43$ & $6.31$ & $9.10$ & $12.48$ & $3.53$ & $5.09$ & $7.62$ & $10.94$ & $15.25$ & $3.70$ & $5.41$ & $7.84$ & $11.18$ & $15.25$ & $3.67$ & $5.39$ & $7.87$ & $11.07$ & $14.98$ \\
			Polar$(128,96)$ & $3.22$ & $4.35$ & $6.29$ & $9.09$ & $12.31$ & $3.38$ & $4.79$ & $7.08$ & $10.17$ & $13.14$ & $3.42$ & $4.96$ & $7.34$ & $10.22$ & $13.14$ & $3.40$ & $4.91$ & $7.28$ & $10.22$ & $13.18$ \\
			\addlinespace[2pt]
			LDPC$(49,24)$ & $3.14$ & $4.31$ & $6.10$ & $8.58$ & $12.16$ & $3.35$ & $4.71$ & $6.70$ & $9.52$ & $13.30$ & $3.10$ & $4.23$ & $6.00$ & $8.48$ & $11.78$ & $3.66$ & $5.11$ & $7.15$ & $10.34$ & $14.18$ \\
			LDPC$(121,60)$ & $2.49$ & $3.35$ & $5.11$ & $8.23$ & $12.98$ & $2.58$ & $3.64$ & $5.79$ & $9.29$ & $14.94$ & $2.47$ & $3.52$ & $5.52$ & $9.04$ & $14.35$ & $2.62$ & $3.85$ & $6.23$ & $10.13$ & $16.08$ \\
			LDPC$(121,70)$ & $2.87$ & $4.10$ & $6.31$ & $10.11$ & $15.55$ & $3.00$ & $4.44$ & $7.04$ & $11.36$ & $17.28$ & $2.90$ & $4.35$ & $6.95$ & $11.31$ & $17.16$ & $3.10$ & $4.75$ & $7.75$ & $12.41$ & $18.98$ \\
			LDPC$(121,80)$ & $3.21$ & $4.67$ & $7.19$ & $11.40$ & $17.12$ & $3.35$ & $5.02$ & $8.00$ & $12.66$ & $18.24$ & $3.31$ & $5.05$ & $8.16$ & $12.69$ & $17.74$ & $3.52$ & $5.53$ & $8.77$ & $13.89$ & $21.39$ \\
			\addlinespace[2pt]
			MacKay$(96,48)$ & $3.43$ & $4.89$ & $7.13$ & $10.46$ & $14.51$ & $3.59$ & $5.27$ & $7.82$ & $11.53$ & $15.60$ & $3.67$ & $5.55$ & $8.28$ & $12.12$ & $15.95$ & $3.83$ & $5.68$ & $8.60$ & $12.40$ & $16.21$ \\
			\addlinespace[2pt]
			CCSDS$(128,64)$ & $2.96$ & $4.34$ & $6.75$ & $10.62$ & $15.81$ & $3.13$ & $4.78$ & $7.65$ & $12.19$ & $17.58$ & $3.14$ & $4.97$ & $7.98$ & $12.66$ & $17.18$ & $3.29$ & $5.23$ & $8.36$ & $13.18$ & $19.42$ \\
			\bottomrule
	\end{tabular}}
	\par\medskip
	\textbf{\small (b) CAL-ECC}\par\smallskip
	\resizebox{!}{0.13\textwidth}{%
		\begin{tabular}{l ccccc ccccc ccccc}
			\toprule
			\multirow{2}{*}{Code} & \multicolumn{5}{c}{$\tau{=}10^{-1}$} & \multicolumn{5}{c}{$\tau{=}10^{-2}$} & \multicolumn{5}{c}{$\tau{=}10^{-3}$} \\
			\cmidrule(lr){2-6}\cmidrule(lr){7-11}\cmidrule(lr){12-16}
			& 2 & 3 & 4 & 5 & 6 & 2 & 3 & 4 & 5 & 6 & 2 & 3 & 4 & 5 & 6 \\
			\midrule
			BCH$(31,16)$ & $3.78$ & $5.00$ & $6.61$ & $8.93$ & $11.83$ & $4.11$ & $5.45$ & $7.19$ & $9.56$ & $12.47$ & $4.17$ & $5.63$ & $7.48$ & $9.89$ & $13.04$ \\
			BCH$(63,36)$ & $3.00$ & $3.92$ & $5.36$ & $7.62$ & $10.71$ & $3.53$ & $4.70$ & $6.44$ & $8.78$ & $12.00$ & $3.80$ & $5.32$ & $7.43$ & $10.16$ & $13.53$ \\
			BCH$(63,45)$ & $3.32$ & $4.37$ & $6.04$ & $8.46$ & $11.86$ & $3.68$ & $5.09$ & $7.12$ & $9.74$ & $13.01$ & $3.77$ & $5.40$ & $7.82$ & $10.86$ & $14.49$ \\
			BCH$(63,51)$ & $3.29$ & $4.42$ & $6.22$ & $8.78$ & $12.27$ & $3.37$ & $4.65$ & $6.65$ & $9.38$ & $12.89$ & $3.37$ & $4.69$ & $6.80$ & $9.63$ & $13.28$ \\
			\addlinespace[2pt]
			Polar$(64,32)$ & $4.00$ & $5.57$ & $7.61$ & $10.30$ & $13.67$ & $4.25$ & $5.85$ & $7.89$ & $10.58$ & $13.89$ & $4.34$ & $5.95$ & $7.97$ & $10.67$ & $13.96$ \\
			Polar$(64,48)$ & $3.57$ & $4.81$ & $6.59$ & $8.63$ & $11.24$ & $3.66$ & $4.91$ & $6.63$ & $8.65$ & $11.25$ & $3.66$ & $4.94$ & $6.66$ & $8.67$ & $11.25$ \\
			Polar$(128,64)$ & $2.76$ & $4.03$ & $6.29$ & $9.94$ & $14.92$ & $3.06$ & $4.45$ & $6.73$ & $10.22$ & $15.13$ & $3.36$ & $5.10$ & $7.61$ & $11.22$ & $15.94$ \\
			Polar$(128,86)$ & $3.54$ & $5.34$ & $8.04$ & $11.77$ & $16.09$ & $3.86$ & $5.80$ & $8.51$ & $12.14$ & $16.54$ & $4.06$ & $6.17$ & $9.00$ & $12.65$ & $17.06$ \\
			Polar$(128,96)$ & $3.43$ & $5.06$ & $7.53$ & $10.54$ & $13.30$ & $3.63$ & $5.42$ & $7.83$ & $10.59$ & $13.32$ & $3.70$ & $5.63$ & $8.11$ & $10.67$ & $13.33$ \\
			\addlinespace[2pt]
			LDPC$(49,24)$ & $3.75$ & $5.24$ & $7.38$ & $10.35$ & $14.32$ & $4.35$ & $6.02$ & $8.21$ & $10.95$ & $14.65$ & $4.53$ & $6.40$ & $8.72$ & $11.49$ & $14.94$ \\
			LDPC$(121,60)$ & $2.65$ & $3.86$ & $6.24$ & $10.18$ & $16.31$ & $3.07$ & $4.53$ & $7.01$ & $10.64$ & $16.53$ & $3.36$ & $5.26$ & $8.25$ & $12.24$ & $17.91$ \\
			LDPC$(121,70)$ & $3.14$ & $4.76$ & $7.72$ & $12.55$ & $19.19$ & $3.70$ & $5.54$ & $8.40$ & $13.10$ & $19.47$ & $4.14$ & $6.47$ & $9.85$ & $14.43$ & $20.30$ \\
			LDPC$(121,80)$ & $3.53$ & $5.46$ & $8.77$ & $13.81$ & $21.09$ & $4.05$ & $6.23$ & $9.45$ & $14.30$ & $21.61$ & $4.37$ & $6.96$ & $10.58$ & $15.42$ & $22.38$ \\
			\addlinespace[2pt]
			MacKay$(96,48)$ & $3.88$ & $5.92$ & $8.85$ & $12.72$ & $16.26$ & $4.56$ & $6.57$ & $9.40$ & $13.01$ & $16.34$ & $5.09$ & $7.43$ & $10.31$ & $13.40$ & $16.40$ \\
			\addlinespace[2pt]
			CCSDS$(128,64)$ & $3.32$ & $5.23$ & $8.59$ & $13.64$ & $21.33$ & $3.78$ & $5.72$ & $8.96$ & $13.83$ & $21.57$ & $4.35$ & $6.72$ & $9.95$ & $14.70$ & $21.79$ \\
			\bottomrule
	\end{tabular}}
	\par\smallskip\parbox{\textheight}{\footnotesize \textit{Note.} Within every block the five columns are $E_b/N_0=2,3,4,5,6$~dB, and codes appear in the same order in (a) and (b).}
\end{sidewaystable}

\begin{sidewaystable}[p]
	\centering\setlength{\tabcolsep}{3pt}\renewcommand{\arraystretch}{1.05}
	\caption{Complete FER results at $E_b/N_0\in\{2,\dots,6\}$~dB, where lower is better.}
	\label{tab:app-fer}
	\textbf{\small (a) Baselines}\par\smallskip
	\resizebox{!}{0.13\textwidth}{%
		\begin{tabular}{l ccccc ccccc ccccc ccccc}
			\toprule
			\multirow{2}{*}{Code} & \multicolumn{5}{c}{ECCT} & \multicolumn{5}{c}{CrossMPT} & \multicolumn{5}{c}{DDECC} & \multicolumn{5}{c}{SB-ECC} \\
			\cmidrule(lr){2-6}\cmidrule(lr){7-11}\cmidrule(lr){12-16}\cmidrule(lr){17-21}
			& 2 & 3 & 4 & 5 & 6 & 2 & 3 & 4 & 5 & 6 & 2 & 3 & 4 & 5 & 6 & 2 & 3 & 4 & 5 & 6 \\
			\midrule
			BCH$(31,16)$ & $2.3$e$-1$ & $9.4$e$-2$ & $2.6$e$-2$ & $4.6$e$-3$ & $4.2$e$-4$ & $1.1$e$-1$ & $3.3$e$-2$ & $5.9$e$-3$ & $6.1$e$-4$ & $3.3$e$-5$ & $4.8$e$-1$ & $2.7$e$-1$ & $1.2$e$-1$ & $4.2$e$-2$ & $1.2$e$-2$ & $7.6$e$-2$ & $1.9$e$-2$ & $2.9$e$-3$ & $2.3$e$-4$ & $1.1$e$-5$ \\
			BCH$(63,36)$ & $6.7$e$-1$ & $3.9$e$-1$ & $1.4$e$-1$ & $2.8$e$-2$ & $2.2$e$-3$ & $6.3$e$-1$ & $3.4$e$-1$ & $1.1$e$-1$ & $1.8$e$-2$ & $1.3$e$-3$ & $4.9$e$-1$ & $2.3$e$-1$ & $7.0$e$-2$ & $1.3$e$-2$ & $1.3$e$-3$ & $4.1$e$-1$ & $1.7$e$-1$ & $4.6$e$-2$ & $6.6$e$-3$ & $3.8$e$-4$ \\
			BCH$(63,45)$ & $6.0$e$-1$ & $3.0$e$-1$ & $8.5$e$-2$ & $1.2$e$-2$ & $7.0$e$-4$ & $5.9$e$-1$ & $2.9$e$-1$ & $7.7$e$-2$ & $9.4$e$-3$ & $5.8$e$-4$ & $5.1$e$-1$ & $2.3$e$-1$ & $5.8$e$-2$ & $9.7$e$-3$ & $1.2$e$-3$ & $3.7$e$-1$ & $1.2$e$-1$ & $2.1$e$-2$ & $1.5$e$-3$ & $4.9$e$-5$ \\
			BCH$(63,51)$ & $6.0$e$-1$ & $3.0$e$-1$ & $8.3$e$-2$ & $1.3$e$-2$ & $7.1$e$-4$ & $5.6$e$-1$ & $2.5$e$-1$ & $6.1$e$-2$ & $7.0$e$-3$ & $3.0$e$-4$ & $6.5$e$-1$ & $3.6$e$-1$ & $1.3$e$-1$ & $2.8$e$-2$ & $5.0$e$-3$ & $5.2$e$-1$ & $2.3$e$-1$ & $5.7$e$-2$ & $7.2$e$-3$ & $3.2$e$-4$ \\
			\addlinespace[2pt]
			Polar$(64,32)$ & $3.4$e$-1$ & $1.2$e$-1$ & $2.6$e$-2$ & $2.7$e$-3$ & $1.8$e$-4$ & $1.8$e$-1$ & $4.5$e$-2$ & $6.9$e$-3$ & $5.5$e$-4$ & $2.7$e$-5$ & $2.2$e$-1$ & $6.4$e$-2$ & $1.2$e$-2$ & $1.5$e$-3$ & $1.9$e$-4$ & $1.4$e$-1$ & $3.2$e$-2$ & $4.5$e$-3$ & $3.3$e$-4$ & $1.5$e$-5$ \\
			Polar$(64,48)$ & $4.7$e$-1$ & $2.0$e$-1$ & $4.7$e$-2$ & $6.6$e$-3$ & $4.6$e$-4$ & $4.0$e$-1$ & $1.5$e$-1$ & $3.2$e$-2$ & $4.0$e$-3$ & $2.7$e$-4$ & $4.8$e$-1$ & $2.0$e$-1$ & $5.3$e$-2$ & $9.4$e$-3$ & $1.1$e$-3$ & $3.0$e$-1$ & $9.8$e$-2$ & $2.0$e$-2$ & $2.8$e$-3$ & $2.1$e$-4$ \\
			Polar$(128,64)$ & $9.3$e$-1$ & $6.6$e$-1$ & $2.5$e$-1$ & $3.3$e$-2$ & $1.6$e$-3$ & $7.9$e$-1$ & $4.0$e$-1$ & $8.2$e$-2$ & $4.6$e$-3$ & $7.2$e$-5$ & $2.1$e$-1$ & $3.6$e$-2$ & $2.3$e$-3$ & $6.8$e$-5$ & $4.0$e$-6$ & $2.6$e$-1$ & $5.4$e$-2$ & $4.6$e$-3$ & $1.5$e$-4$ & $6.0$e$-6$ \\
			Polar$(128,86)$ & $7.7$e$-1$ & $3.9$e$-1$ & $8.7$e$-2$ & $7.2$e$-3$ & $3.0$e$-4$ & $5.5$e$-1$ & $1.7$e$-1$ & $1.8$e$-2$ & $7.7$e$-4$ & $1.5$e$-5$ & $3.3$e$-1$ & $7.3$e$-2$ & $7.1$e$-3$ & $2.6$e$-4$ & $6.6$e$-6$ & $3.3$e$-1$ & $7.1$e$-2$ & $6.9$e$-3$ & $3.1$e$-4$ & $6.7$e$-6$ \\
			Polar$(128,96)$ & $7.8$e$-1$ & $3.8$e$-1$ & $8.0$e$-2$ & $6.4$e$-3$ & $2.8$e$-4$ & $6.3$e$-1$ & $2.2$e$-1$ & $3.0$e$-2$ & $1.7$e$-3$ & $9.3$e$-5$ & $4.6$e$-1$ & $1.2$e$-1$ & $1.4$e$-2$ & $1.1$e$-3$ & $6.8$e$-5$ & $4.6$e$-1$ & $1.3$e$-1$ & $1.5$e$-2$ & $9.9$e$-4$ & $6.0$e$-5$ \\
			\addlinespace[2pt]
			LDPC$(49,24)$ & $3.5$e$-1$ & $1.3$e$-1$ & $2.6$e$-2$ & $2.6$e$-3$ & $8.9$e$-5$ & $2.5$e$-1$ & $7.3$e$-2$ & $1.1$e$-2$ & $7.4$e$-4$ & $1.9$e$-5$ & $3.1$e$-1$ & $1.1$e$-1$ & $1.9$e$-2$ & $1.7$e$-3$ & $6.2$e$-5$ & $1.6$e$-1$ & $4.0$e$-2$ & $5.3$e$-3$ & $2.2$e$-4$ & $4.6$e$-6$ \\
			LDPC$(121,60)$ & $8.5$e$-1$ & $4.9$e$-1$ & $1.2$e$-1$ & $7.0$e$-3$ & $8.7$e$-5$ & $7.2$e$-1$ & $3.1$e$-1$ & $4.2$e$-2$ & $1.5$e$-3$ & $7.4$e$-6$ & $6.7$e$-1$ & $2.6$e$-1$ & $3.8$e$-2$ & $1.2$e$-3$ & $7.1$e$-6$ & $6.1$e$-1$ & $2.0$e$-1$ & $1.9$e$-2$ & $4.0$e$-4$ & $1.1$e$-6$ \\
			LDPC$(121,70)$ & $7.2$e$-1$ & $2.9$e$-1$ & $4.1$e$-2$ & $1.2$e$-3$ & $7.2$e$-6$ & $5.6$e$-1$ & $1.6$e$-1$ & $1.4$e$-2$ & $2.3$e$-4$ & $1.6$e$-6$ & $4.9$e$-1$ & $1.3$e$-1$ & $1.0$e$-2$ & $1.4$e$-4$ & $1.5$e$-6$ & $4.1$e$-1$ & $8.6$e$-2$ & $4.5$e$-3$ & $4.4$e$-5$ & $6.5$e$-8$ \\
			LDPC$(121,80)$ & $6.0$e$-1$ & $1.9$e$-1$ & $1.9$e$-2$ & $4.0$e$-4$ & $2.3$e$-6$ & $4.6$e$-1$ & $1.1$e$-1$ & $6.5$e$-3$ & $7.6$e$-5$ & $9.8$e$-7$ & $3.9$e$-1$ & $7.4$e$-2$ & $3.5$e$-3$ & $4.1$e$-5$ & $9.8$e$-7$ & $3.1$e$-1$ & $4.5$e$-2$ & $1.9$e$-3$ & $1.1$e$-5$ & $7.0$e$-9$ \\
			\addlinespace[2pt]
			MacKay$(96,48)$ & $4.4$e$-1$ & $1.4$e$-1$ & $1.9$e$-2$ & $9.0$e$-4$ & $2.0$e$-5$ & $3.2$e$-1$ & $7.9$e$-2$ & $7.6$e$-3$ & $2.6$e$-4$ & $6.0$e$-6$ & $2.0$e$-1$ & $3.3$e$-2$ & $2.4$e$-3$ & $6.3$e$-5$ & $3.1$e$-6$ & $1.7$e$-1$ & $2.9$e$-2$ & $1.8$e$-3$ & $4.9$e$-5$ & $1.4$e$-6$ \\
			\addlinespace[2pt]
			CCSDS$(128,64)$ & $6.7$e$-1$ & $2.5$e$-1$ & $3.5$e$-2$ & $1.1$e$-3$ & $9.3$e$-6$ & $5.0$e$-1$ & $1.3$e$-1$ & $9.7$e$-3$ & $1.4$e$-4$ & $2.2$e$-6$ & $3.2$e$-1$ & $5.5$e$-2$ & $2.8$e$-3$ & $3.1$e$-5$ & $2.5$e$-6$ & $2.9$e$-1$ & $4.5$e$-2$ & $2.1$e$-3$ & $1.8$e$-5$ & $6.3$e$-8$ \\
			\bottomrule
	\end{tabular}}
	\par\medskip
	\textbf{\small (b) CAL-ECC}\par\smallskip
	\resizebox{!}{0.13\textwidth}{%
		\begin{tabular}{l ccccc ccccc ccccc}
			\toprule
			\multirow{2}{*}{Code} & \multicolumn{5}{c}{$\tau{=}10^{-1}$} & \multicolumn{5}{c}{$\tau{=}10^{-2}$} & \multicolumn{5}{c}{$\tau{=}10^{-3}$} \\
			\cmidrule(lr){2-6}\cmidrule(lr){7-11}\cmidrule(lr){12-16}
			& 2 & 3 & 4 & 5 & 6 & 2 & 3 & 4 & 5 & 6 & 2 & 3 & 4 & 5 & 6 \\
			\midrule
			BCH$(31,16)$ & $1.1$e$-1$ & $3.5$e$-2$ & $7.4$e$-3$ & $7.8$e$-4$ & $4.6$e$-5$ & $6.5$e$-2$ & $1.8$e$-2$ & $3.3$e$-3$ & $3.5$e$-4$ & $2.1$e$-5$ & $6.1$e$-2$ & $1.4$e$-2$ & $2.3$e$-3$ & $2.2$e$-4$ & $9.6$e$-6$ \\
			BCH$(63,36)$ & $4.4$e$-1$ & $2.0$e$-1$ & $5.4$e$-2$ & $6.3$e$-3$ & $3.1$e$-4$ & $2.3$e$-1$ & $8.4$e$-2$ & $1.7$e$-2$ & $1.9$e$-3$ & $8.6$e$-5$ & $1.6$e$-1$ & $4.0$e$-2$ & $5.6$e$-3$ & $4.3$e$-4$ & $1.7$e$-5$ \\
			BCH$(63,45)$ & $3.9$e$-1$ & $1.6$e$-1$ & $3.5$e$-2$ & $3.3$e$-3$ & $1.2$e$-4$ & $2.0$e$-1$ & $5.9$e$-2$ & $9.8$e$-3$ & $8.6$e$-4$ & $3.7$e$-5$ & $1.6$e$-1$ & $3.5$e$-2$ & $3.6$e$-3$ & $2.2$e$-4$ & $7.3$e$-6$ \\
			BCH$(63,51)$ & $3.8$e$-1$ & $1.4$e$-1$ & $2.7$e$-2$ & $2.3$e$-3$ & $7.5$e$-5$ & $3.1$e$-1$ & $9.3$e$-2$ & $1.4$e$-2$ & $1.0$e$-3$ & $3.4$e$-5$ & $3.1$e$-1$ & $8.9$e$-2$ & $1.1$e$-2$ & $7.2$e$-4$ & $2.0$e$-5$ \\
			\addlinespace[2pt]
			Polar$(64,32)$ & $1.2$e$-1$ & $2.7$e$-2$ & $3.8$e$-3$ & $2.7$e$-4$ & $9.7$e$-6$ & $9.3$e$-2$ & $2.0$e$-2$ & $2.9$e$-3$ & $2.0$e$-4$ & $7.5$e$-6$ & $8.7$e$-2$ & $1.9$e$-2$ & $2.6$e$-3$ & $1.8$e$-4$ & $6.9$e$-6$ \\
			Polar$(64,48)$ & $2.9$e$-1$ & $9.9$e$-2$ & $1.9$e$-2$ & $2.8$e$-3$ & $2.1$e$-4$ & $2.7$e$-1$ & $9.3$e$-2$ & $1.9$e$-2$ & $2.7$e$-3$ & $2.1$e$-4$ & $2.7$e$-1$ & $9.1$e$-2$ & $1.9$e$-2$ & $2.7$e$-3$ & $2.1$e$-4$ \\
			Polar$(128,64)$ & $5.4$e$-1$ & $1.7$e$-1$ & $2.0$e$-2$ & $5.9$e$-4$ & $4.7$e$-6$ & $3.8$e$-1$ & $1.1$e$-1$ & $1.3$e$-2$ & $4.4$e$-4$ & $3.8$e$-6$ & $2.8$e$-1$ & $5.4$e$-2$ & $5.1$e$-3$ & $1.6$e$-4$ & $1.8$e$-6$ \\
			Polar$(128,86)$ & $3.2$e$-1$ & $6.3$e$-2$ & $4.8$e$-3$ & $1.3$e$-4$ & $1.9$e$-6$ & $2.2$e$-1$ & $3.8$e$-2$ & $2.9$e$-3$ & $8.7$e$-5$ & $1.1$e$-6$ & $1.8$e$-1$ & $2.7$e$-2$ & $1.8$e$-3$ & $5.1$e$-5$ & $6.3$e$-7$ \\
			Polar$(128,96)$ & $4.1$e$-1$ & $9.9$e$-2$ & $1.1$e$-2$ & $7.5$e$-4$ & $5.3$e$-5$ & $3.1$e$-1$ & $7.0$e$-2$ & $8.3$e$-3$ & $7.2$e$-4$ & $5.2$e$-5$ & $2.9$e$-1$ & $5.8$e$-2$ & $6.7$e$-3$ & $6.9$e$-4$ & $5.1$e$-5$ \\
			\addlinespace[2pt]
			LDPC$(49,24)$ & $1.4$e$-1$ & $3.4$e$-2$ & $4.2$e$-3$ & $2.2$e$-4$ & $4.2$e$-6$ & $5.9$e$-2$ & $1.3$e$-2$ & $1.6$e$-3$ & $1.1$e$-4$ & $2.9$e$-6$ & $4.9$e$-2$ & $7.9$e$-3$ & $8.8$e$-4$ & $6.1$e$-5$ & $2.0$e$-6$ \\
			LDPC$(121,60)$ & $5.9$e$-1$ & $2.0$e$-1$ & $2.0$e$-2$ & $4.1$e$-4$ & $9.5$e$-7$ & $3.7$e$-1$ & $9.6$e$-2$ & $9.0$e$-3$ & $2.6$e$-4$ & $7.7$e$-7$ & $2.6$e$-1$ & $4.3$e$-2$ & $2.5$e$-3$ & $5.1$e$-5$ & $1.9$e$-7$ \\
			LDPC$(121,70)$ & $4.1$e$-1$ & $8.9$e$-2$ & $4.9$e$-3$ & $4.3$e$-5$ & $6.1$e$-8$ & $2.2$e$-1$ & $4.0$e$-2$ & $2.5$e$-3$ & $2.5$e$-5$ & $4.6$e$-8$ & $1.3$e$-1$ & $1.4$e$-2$ & $5.7$e$-4$ & $6.5$e$-6$ & $2.0$e$-8$ \\
			LDPC$(121,80)$ & $3.1$e$-1$ & $5.0$e$-2$ & $2.0$e$-3$ & $1.4$e$-5$ & $1.2$e$-8$ & $1.6$e$-1$ & $2.2$e$-2$ & $9.9$e$-4$ & $8.7$e$-6$ & $7.0$e$-9$ & $9.5$e$-2$ & $8.5$e$-3$ & $2.8$e$-4$ & $2.6$e$-6$ & $3.0$e$-9$ \\
			\addlinespace[2pt]
			MacKay$(96,48)$ & $1.5$e$-1$ & $2.2$e$-2$ & $1.3$e$-3$ & $3.5$e$-5$ & $1.3$e$-6$ & $6.7$e$-2$ & $1.1$e$-2$ & $7.5$e$-4$ & $2.7$e$-5$ & $1.2$e$-6$ & $3.6$e$-2$ & $4.2$e$-3$ & $3.2$e$-4$ & $2.0$e$-5$ & $1.2$e$-6$ \\
			\addlinespace[2pt]
			CCSDS$(128,64)$ & $2.8$e$-1$ & $4.5$e$-2$ & $1.7$e$-3$ & $1.2$e$-5$ & $6.0$e$-9$ & $1.7$e$-1$ & $2.8$e$-2$ & $1.2$e$-3$ & $1.0$e$-5$ & $4.0$e$-9$ & $9.2$e$-2$ & $9.8$e$-3$ & $4.4$e$-4$ & $4.2$e$-6$ & $3.0$e$-9$ \\
			\bottomrule
	\end{tabular}}
	\par\smallskip\parbox{\textheight}{\footnotesize \textit{Note.} Within every block the five columns are $E_b/N_0=2,3,4,5,6$~dB, and codes appear in the same order in (a) and (b).}
\end{sidewaystable}


\section{The $\tau$ Axis}
\label{app:tau}

Section~\ref{subsec:exp-nfe} illustrated the interaction between $\tau$ and
cost on a single code. Table~\ref{tab:app-tau} summarises it over all cells,
grouped by certification rate.

The relationship is strong, with the median cost of moving from $\tau=10^{-1}$ to
$10^{-3}$ at $1.06\times$ above $95\%$ certification and $8.74\times$ below
$50\%$. The mechanism is the one described in
Section~\ref{subsec:exp-cert}. Certified frames stop at the same point
regardless of $\tau$, so lowering it can only deepen the search on the
uncertified fraction, and because the residual errors live almost entirely in
that fraction the search is deepened exactly where it can still help.

The ranges in the last column show the limit of this reasoning. Certification
bounds the \emph{fraction} of frames on which $\tau$ acts, but not the
\emph{depth} to which those frames search. Three cells, BCH$(63,36)$,
$(63,45)$ and $(63,51)$ at $5$~dB, have certification above $95\%$ yet cost
$8.6\times$, $4.4\times$ and $2.1\times$. On BCH$(63,36)$ the $1.4\%$
uncertified frames have $p_{99}$ growing from $52$ to $545$ as $\tau$ falls. A
high certification rate therefore guarantees that additional accuracy is
confined to a small set of frames, not that it is free. It is nearly free when those
frames are also shallow, which is the situation at $6$~dB. At the other end the
same accounting explains why $\tau$ is an expensive axis at low $E_b/N_0$. On
Polar$(128,64)$ at $2$~dB certification is $0.3\%$, so essentially every frame
is searched to full depth and the cost multiplier reaches $46.5\times$ for a
gain of $+0.60$ in $-\ln(\mathrm{BER})$. We do not recommend operating there.

Finally, a smaller $\tau$ is wasted where the candidate set already contains
the optimum. On Polar$(64,48)$, MacKay$(96,48)$ and Polar$(128,96)$ at $6$~dB
the gain from $10^{-1}$ to $10^{-3}$ is $+0.01$, $+0.14$ and $+0.03$ at no
additional cost. The search has nothing left to find. These cells are
identifiable in advance from the certification rate together with the spacing
of the FER curves in Figure~\ref{fig:ml-ceiling}, which is the diagnostic value
of reporting the $\tau$ family rather than a single operating point.

\begin{table}[t]
	\centering\small
	\caption{Cost of lowering $\tau$ from $10^{-1}$ to $10^{-3}$, grouped by certification rate over all $75$ cells.}
	\label{tab:app-tau}
	\begin{tabular}{lrrr}
		\toprule
		$\mathrm{cert}\%$ at $\tau=10^{-3}$ & cells & median cost & range \\
		\midrule
		$>95\%$ & $17$ & $1.06\times$ & $1.00$--$8.61\times$ \\
		$80$--$95\%$ & $12$ & $1.02\times$ & $1.00$--$19.87\times$ \\
		$50$--$80\%$ & $15$ & $1.28\times$ & $1.00$--$27.93\times$ \\
		$<50\%$ & $31$ & $8.74\times$ & $1.08$--$48.85\times$ \\
		\bottomrule
	\end{tabular}
\end{table}

\section{Reliability Signal against Heuristic Baselines}
\label{app:reliability}

Section~\ref{subsec:exp-cert} separates frame errors into detected failures,
whose output is not a codeword, and undetected errors, whose output is a
valid codeword. Detected failures arise when the search finds no valid
candidate and falls back to the hard decision, so their output has cost
$\Phi=0$ and fails the syndrome check; any decoder that checks its syndrome
flags them, and they are not evidence for the certificate. On CCSDS,
BCH$(63,36)$, Polar$(128,64)$ and LDPC$(121,70)$ at $4$~dB they are
$93$--$100\%$ of all errors (Table~\ref{tab:cert}). LDPC$(49,24)$ is the
opposite case, with $7\%$ detected and a median $\Phi$ of $12.39$ on
erroneous outputs against $3.28$ on correct ones.

Table~\ref{tab:app-reliability} restricts to undetected errors and compares
three ways of flagging the same fraction of valid outputs: the uncertified
frames, the top fraction by output cost $\Phi$, and the top fraction by NFE.
On the five cells with at least $100$ undetected errors the certificate
captures $96.7$--$100\%$ and the $\Phi$ threshold $88.5$--$100\%$, the
certificate ahead by $3$--$8$ points on four of the five and tied on
LDPC$(121,80)$. The NFE threshold tracks the certificate almost exactly,
because a certified frame stops immediately and its NFE is a shadow of the
certificate; but its budget is the certification rate itself, which without
the certificate would have to be tuned per cell. The remaining four cells
have $0$--$36$ undetected errors and support no comparison.

\begin{table}[t]
	\centering\setlength{\tabcolsep}{5pt}\small
	\caption{Capture of \emph{undetected} errors at $E_b/N_0=4$~dB, $\tau=10^{-3}$, when the same fraction of valid outputs is flagged by three rules: uncertified frames (the certificate), highest output cost $\Phi$, and highest NFE. Cells are ordered by the number of undetected errors; those below $100$ are listed for completeness only.}
	\label{tab:app-reliability}
	\begin{tabular}{lrrrrr}
		\toprule
		Code & undetected errors & flag\% & certificate & $\Phi$ threshold & NFE threshold \\
		\midrule
		BCH$(31,16)$    & $573$ & $10.8$ & $96.7\%$  & $88.5\%$  & $96.9\%$  \\
		BCH$(63,51)$    & $565$ & $20.4$ & $97.3\%$  & $91.7\%$  & $97.3\%$  \\
		LDPC$(49,24)$   & $490$ & $18.7$ & $98.8\%$  & $95.7\%$  & $98.8\%$  \\
		MacKay$(96,48)$ & $310$ & $40.5$ & $98.1\%$  & $91.6\%$  & $98.4\%$  \\
		LDPC$(121,80)$  & $169$ & $27.6$ & $100.0\%$ & $100.0\%$ & $100.0\%$ \\
		\addlinespace[2pt]
		LDPC$(121,70)$  & $36$  & $37.0$ & $100.0\%$ & $91.7\%$  & $100.0\%$ \\
		Polar$(128,64)$ & $28$  & $81.2$ & $100.0\%$ & $100.0\%$ & $100.0\%$ \\
		BCH$(63,36)$    & $14$  & $10.7$ & $100.0\%$ & $100.0\%$ & $50.0\%$  \\
		CCSDS$(128,64)$ & $0$   & ---    & ---       & ---       & ---       \\
		\bottomrule
	\end{tabular}
\end{table}

\section{Greedy Dual Ascent against the Exact LP Optimum}
\label{app:lpgap}

The greedy ascent of Section~\ref{subsec:cert} returns a feasible point of
\eqref{eq:lp-dual}, so $\Lambda_{\mathrm{D}}\le\mathrm{LP}^\star\le\Phi^\star$
where $\mathrm{LP}^\star$ is the optimum of the relaxation. To measure the
gap we regenerate $20{,}000$ frames per cell from the seeded frame sequence,
compute $\Lambda_{\mathrm{D}}$ and $\mathrm{LP}^\star$ (solved exactly with
HiGHS) for each, and join them to the per-frame output of the decoder at
$\tau=10^{-3}$. Table~\ref{tab:app-lpgap} reports the fraction of frames with
$\mathrm{LP}^\star>\Lambda_{\mathrm{D}}$, the mean and $90$th percentile of
$\mathrm{LP}^\star/\Lambda_{\mathrm{D}}$ over those frames, and the frames the
exact bound would additionally have certified, namely uncertified frames with
a valid output satisfying $\Lambda_{\mathrm{D}}<\Pinc\le\mathrm{LP}^\star$.
This last figure is a lower bound on the gain, since it counts only frames
whose final output was already ML-optimal and omits earlier certification of
an intermediate incumbent.

The greedy bound is close to optimal in value, with a mean gap of $1$--$9\%$.
On sparse $H$ that gap falls where incumbents concentrate, and the exact
optimum would certify a further $2.7$--$7.3\%$ of all frames, or $7$--$20\%$
of the uncertified ones, raising LDPC$(121,70)$ from $63.0\%$ to about
$70\%$. On the dense BCH codes the gap is the largest ($p_{90}$ up to $1.27$)
and the additional certification is zero, so the limit there is the
relaxation, which uses only the covering constraints of unsatisfied checks,
and not the ascent. Polar$(128,64)$ lies between at $4.4\%$. We keep the
greedy ascent because it runs as a batched tensor operation and the exact LP
costs milliseconds per frame on a CPU; the next lever on dense $H$ is the
Feldman polytope rather than the exact optimum.

\begin{table}[t]
	\centering\setlength{\tabcolsep}{5pt}\small
	\caption{Greedy dual bound $\Lambda_{\mathrm{D}}$ against the exact LP optimum $\mathrm{LP}^\star$ at $E_b/N_0=4$~dB, $\tau=10^{-3}$, $20{,}000$ frames per cell. ``gap'' is the fraction of frames with $\mathrm{LP}^\star>\Lambda_{\mathrm{D}}$; mean and $p_{90}$ are of $\mathrm{LP}^\star/\Lambda_{\mathrm{D}}$ over those frames. ``added cert.'' is the fraction of all frames (and of uncertified frames) that the exact optimum would additionally have certified, a lower bound on its gain.}
	\label{tab:app-lpgap}
	\begin{tabular}{lrrrrr}
		\toprule
		Code & gap & mean & $p_{90}$ & added cert.\ (all) & (of uncert.) \\
		\midrule
		MacKay$(96,48)$ & $51.4\%$ & $1.012$ & $1.039$ & $2.94\%$ & $7.3\%$  \\
		LDPC$(49,24)$   & $41.3\%$ & $1.021$ & $1.078$ & $2.70\%$ & $14.5\%$ \\
		BCH$(31,16)$    & $27.7\%$ & $1.046$ & $1.173$ & $0.19\%$ & $1.8\%$  \\
		CCSDS$(128,64)$ & $62.1\%$ & $1.030$ & $1.101$ & $6.10\%$ & $10.4\%$ \\
		LDPC$(121,70)$  & $57.4\%$ & $1.032$ & $1.110$ & $7.25\%$ & $19.7\%$ \\
		LDPC$(121,80)$  & $47.8\%$ & $1.020$ & $1.074$ & $3.94\%$ & $14.3\%$ \\
		Polar$(128,64)$ & $70.5\%$ & $1.055$ & $1.154$ & $4.39\%$ & $5.4\%$  \\
		BCH$(63,36)$    & $49.9\%$ & $1.088$ & $1.272$ & $0.00\%$ & $0.0\%$  \\
		BCH$(63,51)$    & $25.8\%$ & $1.065$ & $1.229$ & $0.29\%$ & $1.4\%$  \\
		\bottomrule
	\end{tabular}
\end{table}

\section{Additional Figures}
\label{app:figures}

\begin{table}[t]
	\centering\setlength{\tabcolsep}{4pt}\small
	\caption{FER at $E_b/N_0=5$~dB against the exact ML reference. CAL-ECC at $\tau=10^{-3}$, with codes sorted by block length. No decoder can go below the reference beyond finite-sample fluctuation (Section~\ref{subsec:exp-ml}). $^\ddagger$The reference for LDPC$(121,80)$ is censored at this operating point and is quoted for scale only. The uncensored comparison for that code is at $4$~dB, where it is $2.82{\times}10^{-4}$ against $6.67{\times}10^{-5}$.}
	\label{tab:ml-ceiling}
	\begin{tabular}{lccccc|c}
		\toprule
		Code & ECCT & CrossMPT & DDECC & SB-ECC & \textbf{CAL-ECC} & exact ML \\
		\midrule
		BCH$(31,16)$    & $4.59{\times}10^{-3}$ & $6.13{\times}10^{-4}$ & $4.20{\times}10^{-2}$ & $2.32{\times}10^{-4}$ & $\bm{2.16{\times}10^{-4}}$ & $2.31{\times}10^{-4}$ \\
		LDPC$(49,24)$   & $2.62{\times}10^{-3}$ & $7.38{\times}10^{-4}$ & $1.65{\times}10^{-3}$ & $2.19{\times}10^{-4}$ & $\bm{6.11{\times}10^{-5}}$ & $5.21{\times}10^{-5}$ \\
		BCH$(63,51)$    & $1.28{\times}10^{-2}$ & $7.02{\times}10^{-3}$ & $2.78{\times}10^{-2}$ & $7.16{\times}10^{-3}$ & $\bm{7.24{\times}10^{-4}}$ & $6.89{\times}10^{-4}$ \\
		MacKay$(96,48)$ & $9.02{\times}10^{-4}$ & $2.57{\times}10^{-4}$ & $6.28{\times}10^{-5}$ & $4.87{\times}10^{-5}$ & $\bm{2.00{\times}10^{-5}}$ & $1.80{\times}10^{-5}$ \\
		LDPC$(121,80)$  & $3.97{\times}10^{-4}$ & $7.64{\times}10^{-5}$ & $4.12{\times}10^{-5}$ & $1.15{\times}10^{-5}$ & $\bm{2.65{\times}10^{-6}}$ & $4.50{\times}10^{-7}$$^\ddagger$ \\
		\bottomrule
	\end{tabular}
\end{table}

\paragraph{What that leaves to optimize.}
No decoder, whether neural, algebraic or exhaustive, can enter the shaded
region of Figure~\ref{fig:app-fer}, so on these codes the only axis on
which a decoder can still improve is computation. On BCH$(63,51)$ we attain
the ceiling at an average of $0.6$ function evaluations per frame at $6$~dB,
below the $\mathrm{NFE}=1$ of a single forward pass. Part of that figure is
structural rather than earned: at $6$~dB the syndrome of BCH$(63,51)$ is
already zero on $70.4\%$ of frames, which cost no forward pass under any
decoder's accounting, and the remaining frames average $2.0$ evaluations,
still an order of magnitude below the $T\times L_{\max}$ upper bound of the
generator. The margin over prior work on these codes is not uniform. On
BCH$(31,16)$ at $5$~dB SB-ECC also attains the ceiling, at $1.00\times$
against our $0.93\times$, and the shortest code in our suite leaves no room
for either decoder to distinguish itself. The gap opens on the longer codes:
on BCH$(63,51)$ the strongest baseline sits at $10.2\times$ the ceiling while
we sit at $1.05\times$, and on LDPC$(121,80)$ at $4$~dB, where the reference
is uncensored, the figures are $28.5\times$ and $4.23\times$, a residual gap
roughly one seventh of the previous state of the art.

\paragraph{The remaining gap grows with block length.}
LDPC$(121,80)$ is the one code in Table~\ref{tab:ml-ceiling} that has not
converged, and the ordering follows block length rather than code family,
which is consistent with the mechanism: the candidate set is generated by a
fixed number of sequential commitments, and the probability of missing
$\hat{\mathbf{e}}_{\mathrm{ML}}$ compounds with $n$. Every grey curve stops at
$5$ or $6$~dB, where the integer program becomes infeasible at the frame
counts needed, and four of our fifteen codes have no published reference at
all. Above those cut-offs no external check on optimality is available, which
is exactly where the certificates apply.

Figures~\ref{fig:app-ber} and~\ref{fig:app-fer} extend the comparison of
Section~\ref{subsec:exp-main} to the full code suite in BER and FER
respectively. CAL-ECC appears as a family of three curves rather than as a
single point, which no fixed-budget decoder offers. The family is widest at low
$E_b/N_0$ and narrows as certification saturates, closing entirely on the cells
identified in Appendix~\ref{app:tau}.

Figure~\ref{fig:app-fer} carries the exact ML reference wherever one is
published, which is the case for eleven of the fifteen codes, and the shaded region
below it is unachievable by any decoder. Five of those panels are reproduced in
Figure~\ref{fig:ml-ceiling} of the main text. Seeing all eleven together makes
the pattern visible in one view. On BCH$(31,16)$ and BCH$(63,51)$ the CAL-ECC
family has collapsed onto the ceiling and no accuracy remains to be won,
whereas on the longer codes the family is still separated. Every reference
terminates at $5$ or $6$~dB, above which no external check on optimality
exists and the certificates supply the per-frame optimality statement in its
place.

\newcommand{\gridfig}[2]{%
	\begin{subfigure}[b]{0.3\textwidth}
		#1
		\caption{#2}
\end{subfigure}}

\begin{figure}[p]
	\centering
	\gridfig{\includegraphics[width=\linewidth]{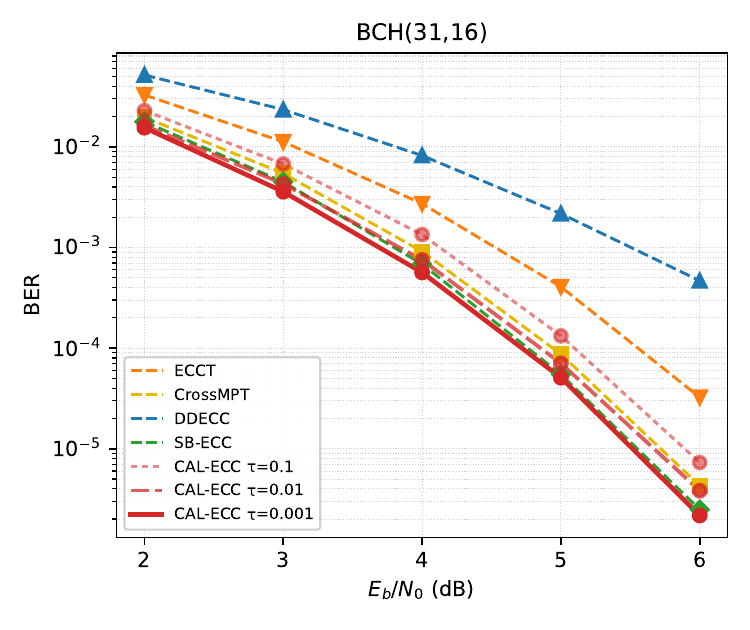}}{BCH$(31,16)$}\hfill
	\gridfig{\includegraphics[width=\linewidth]{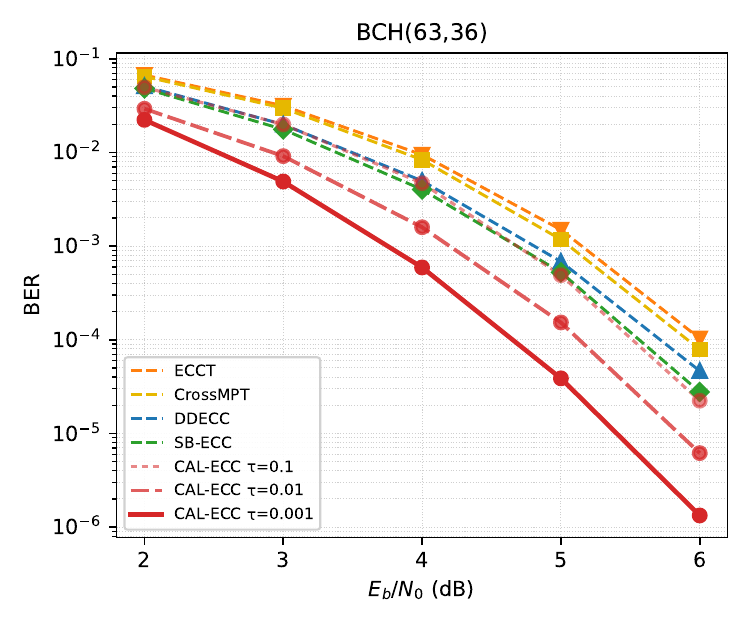}}{BCH$(63,36)$}\hfill
	\gridfig{\includegraphics[width=\linewidth]{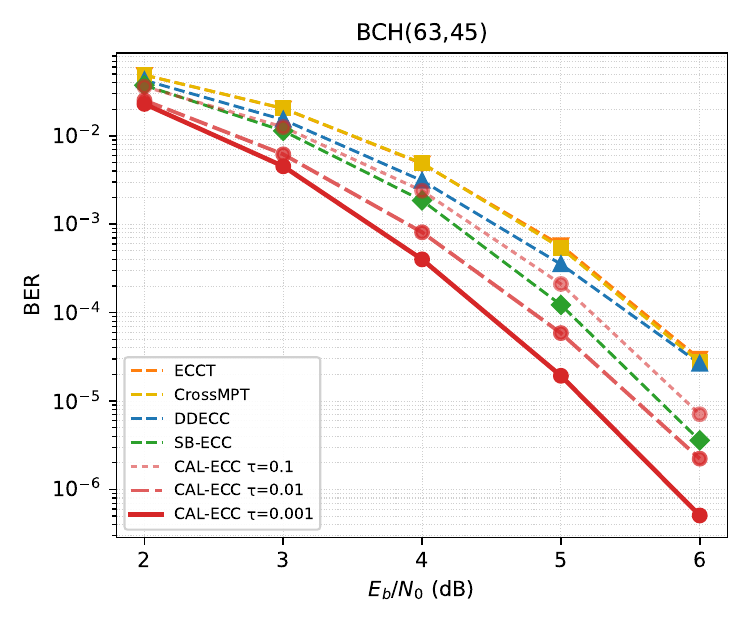}}{BCH$(63,45)$}\\[3pt]
	\gridfig{\includegraphics[width=\linewidth]{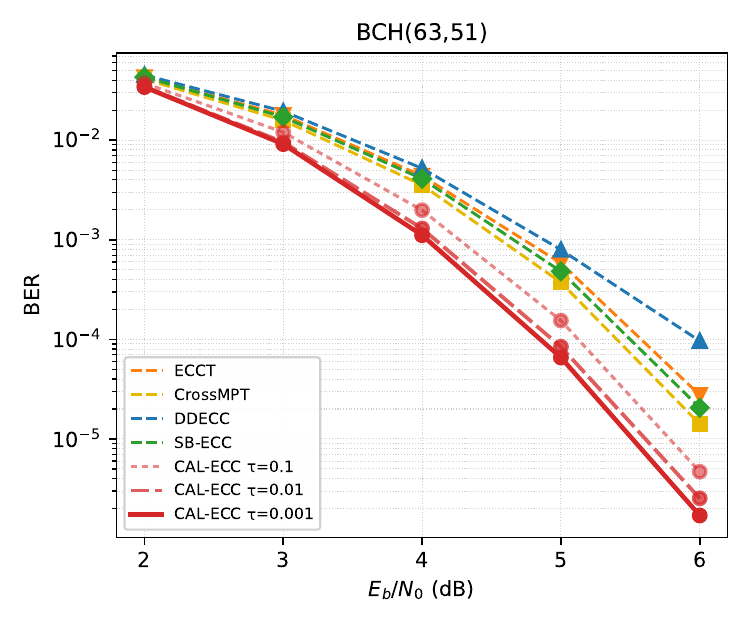}}{BCH$(63,51)$}\hfill
	\gridfig{\includegraphics[width=\linewidth]{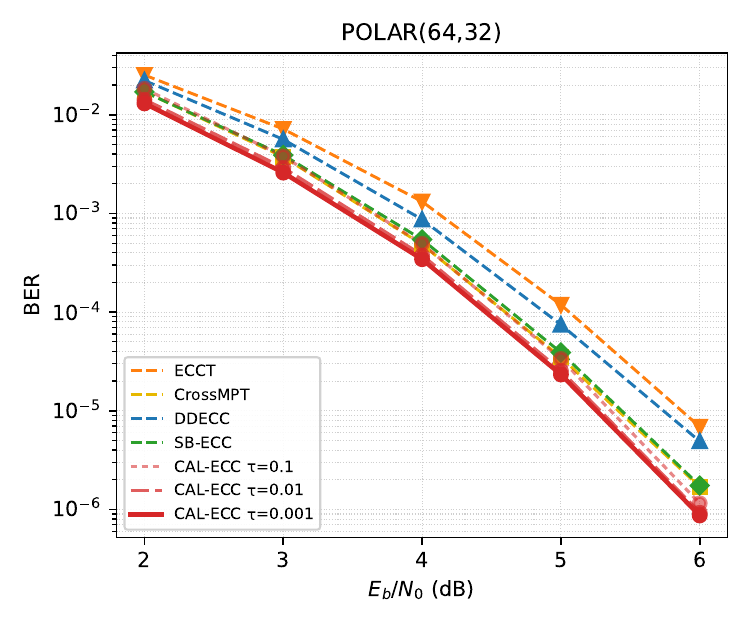}}{Polar$(64,32)$}\hfill
	\gridfig{\includegraphics[width=\linewidth]{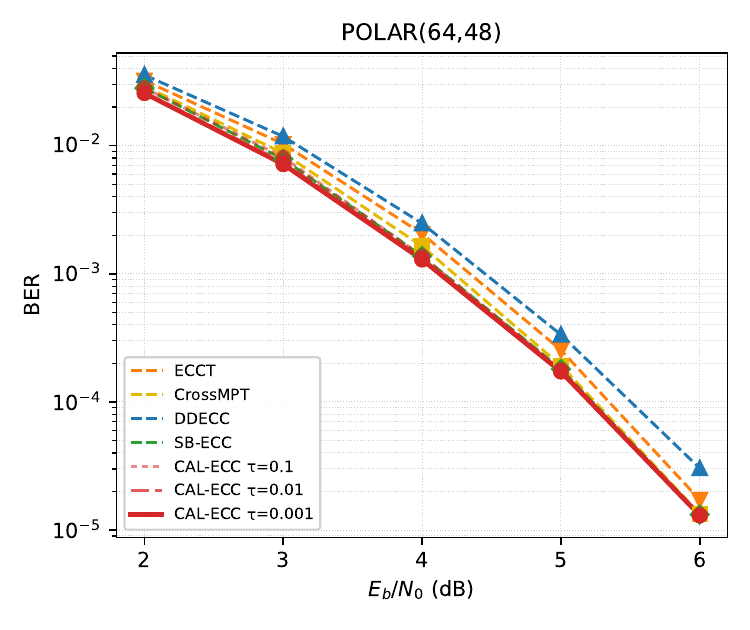}}{Polar$(64,48)$}\\[3pt]
	\gridfig{\includegraphics[width=\linewidth]{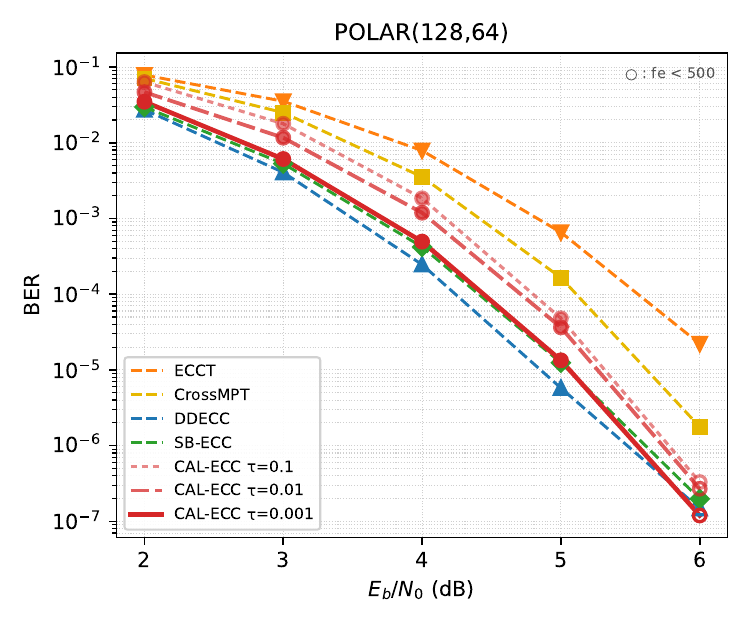}}{Polar$(128,64)$}\hfill
	\gridfig{\includegraphics[width=\linewidth]{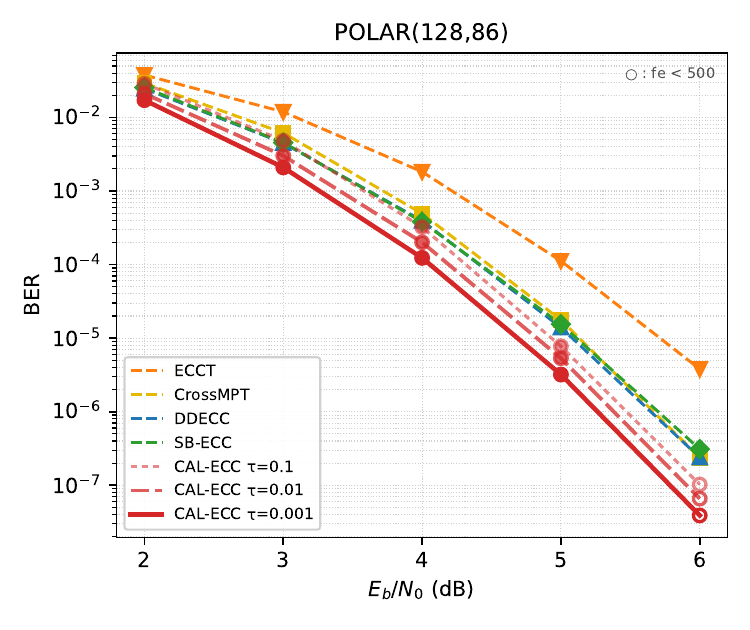}}{Polar$(128,86)$}\hfill
	\gridfig{\includegraphics[width=\linewidth]{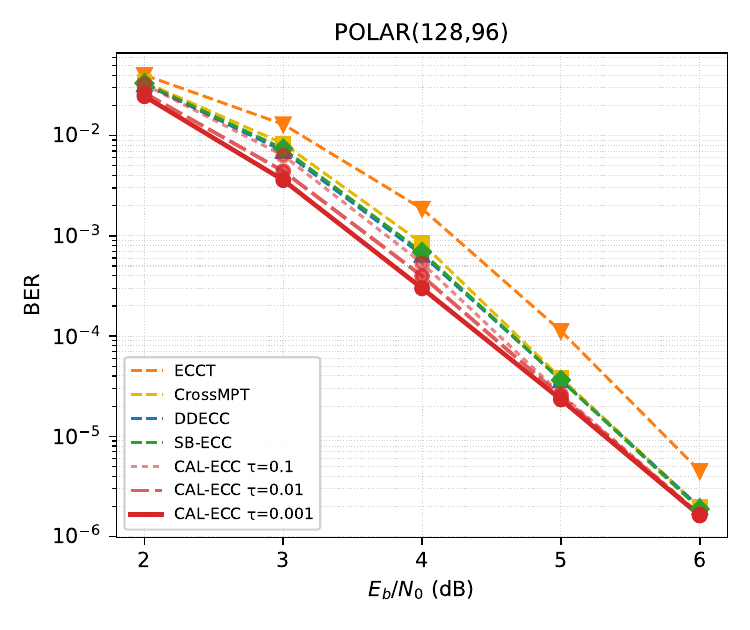}}{Polar$(128,96)$}\\[3pt]
	\gridfig{\includegraphics[width=\linewidth]{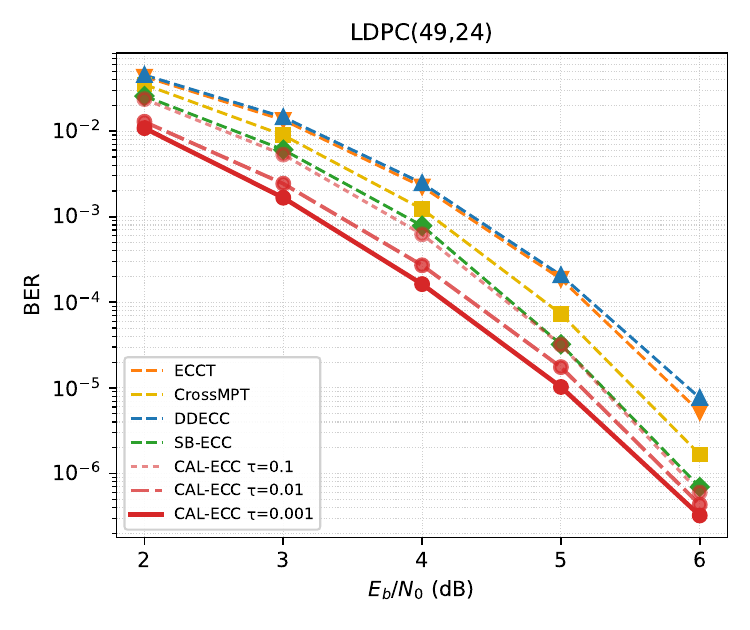}}{LDPC$(49,24)$}\hfill
	\gridfig{\includegraphics[width=\linewidth]{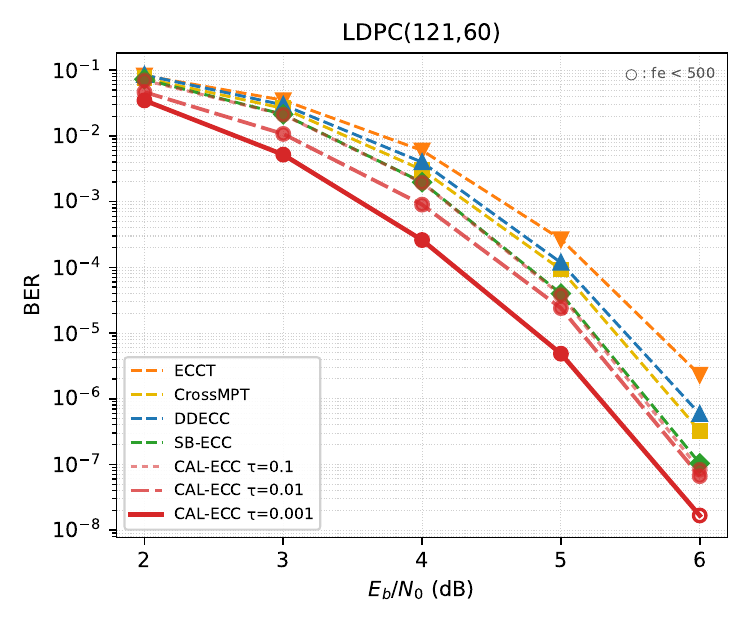}}{LDPC$(121,60)$}\hfill
	\gridfig{\includegraphics[width=\linewidth]{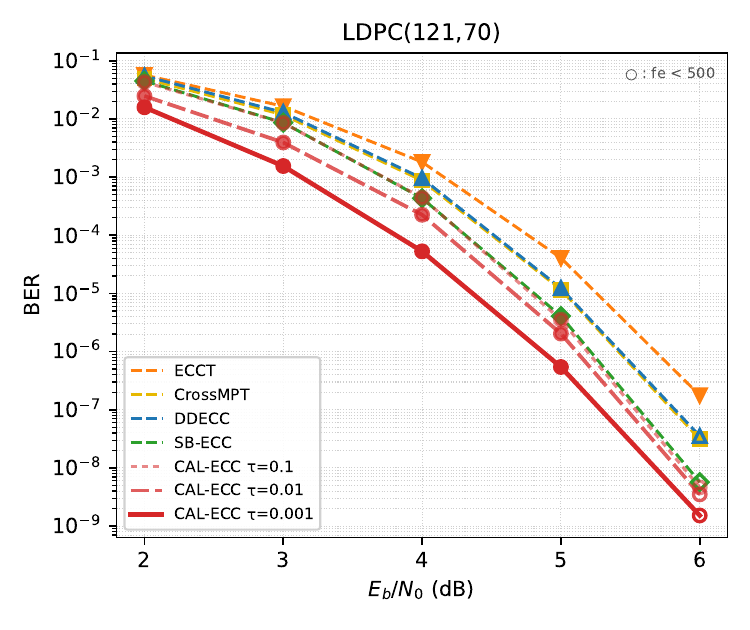}}{LDPC$(121,70)$}\\[3pt]
	\gridfig{\includegraphics[width=\linewidth]{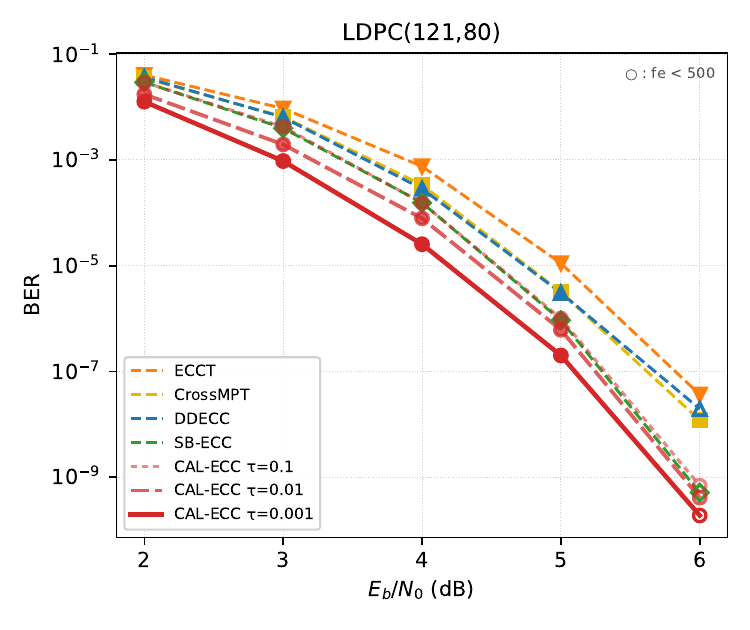}}{LDPC$(121,80)$}\hfill
	\gridfig{\includegraphics[width=\linewidth]{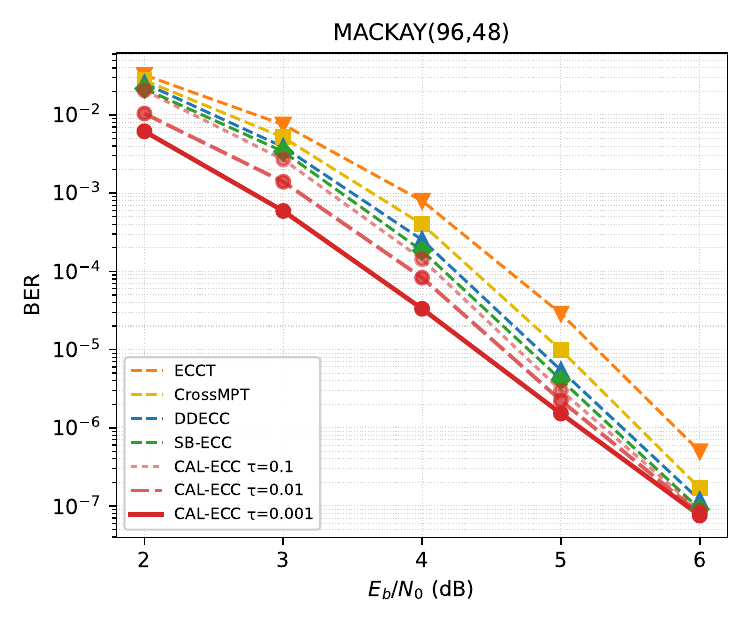}}{MacKay$(96,48)$}\hfill
	\gridfig{\includegraphics[width=\linewidth]{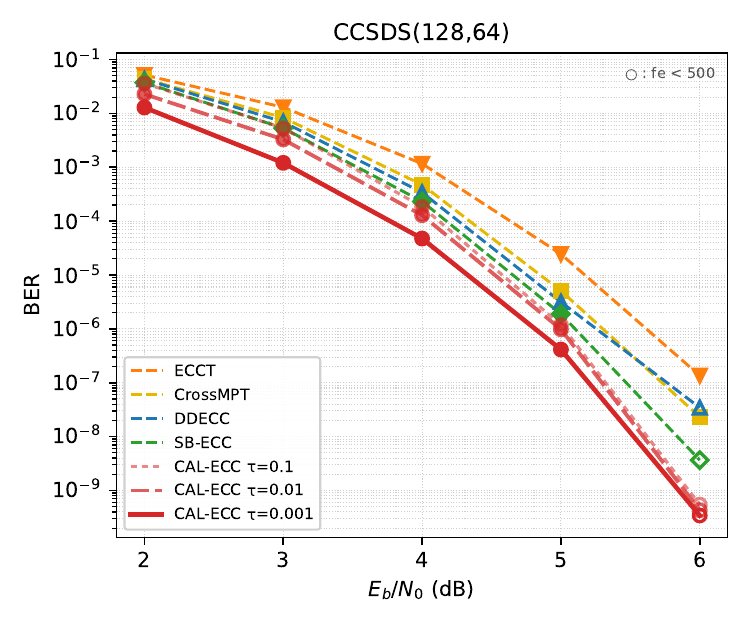}}{CCSDS$(128,64)$}
	\caption{BER on all fifteen codes. CAL-ECC curves darken as $\tau$
		decreases, and baselines are dashed. The CAL-ECC family is nested, with lower $\tau$ shifting the curve down without changing its slope, which is the signature of a decoder approaching a fixed ceiling rather than one whose
		error exponent has changed.}
	\label{fig:app-ber}
\end{figure}

\begin{figure}[p]
	\centering
	\gridfig{\includegraphics[width=\linewidth]{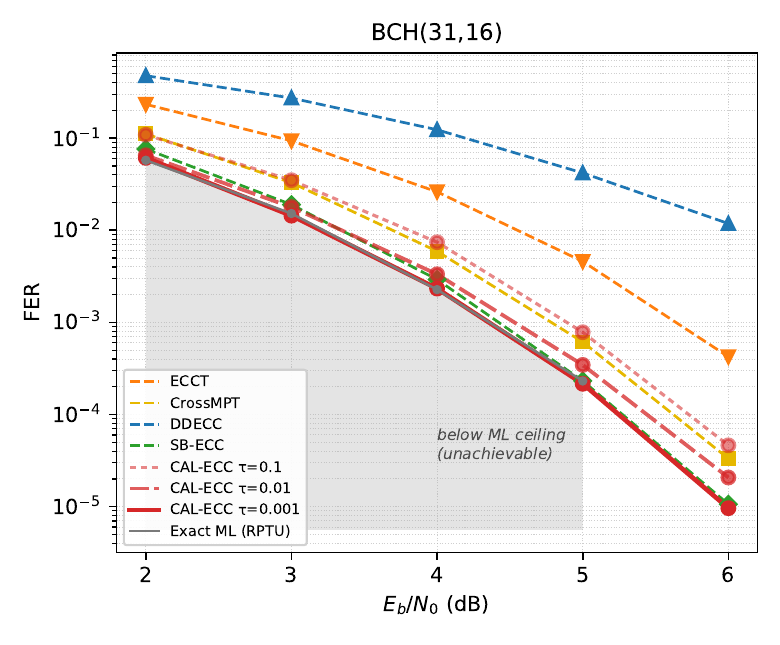}}{BCH$(31,16)^\ast$}\hfill
	\gridfig{\includegraphics[width=\linewidth]{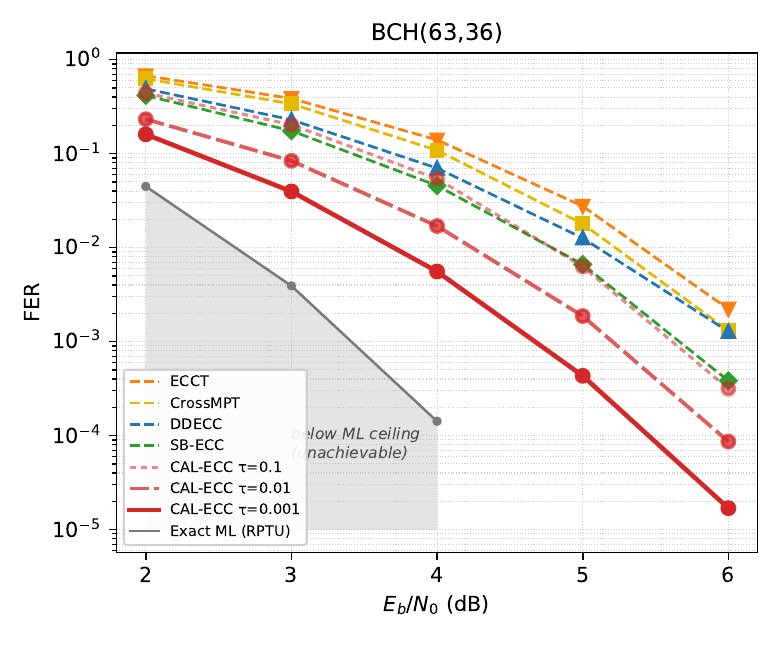}}{BCH$(63,36)^\ast$}\hfill
	\gridfig{\includegraphics[width=\linewidth]{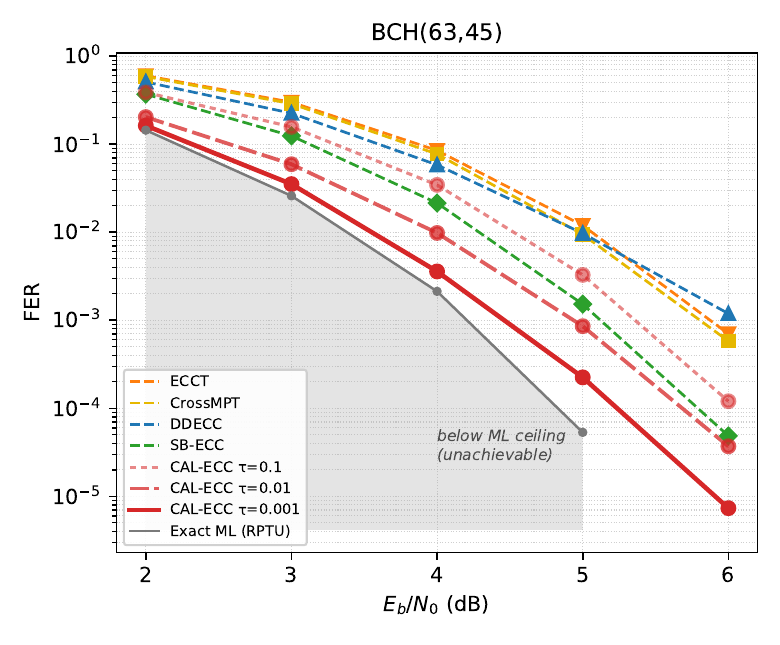}}{BCH$(63,45)^\ast$}\\[3pt]
	\gridfig{\includegraphics[width=\linewidth]{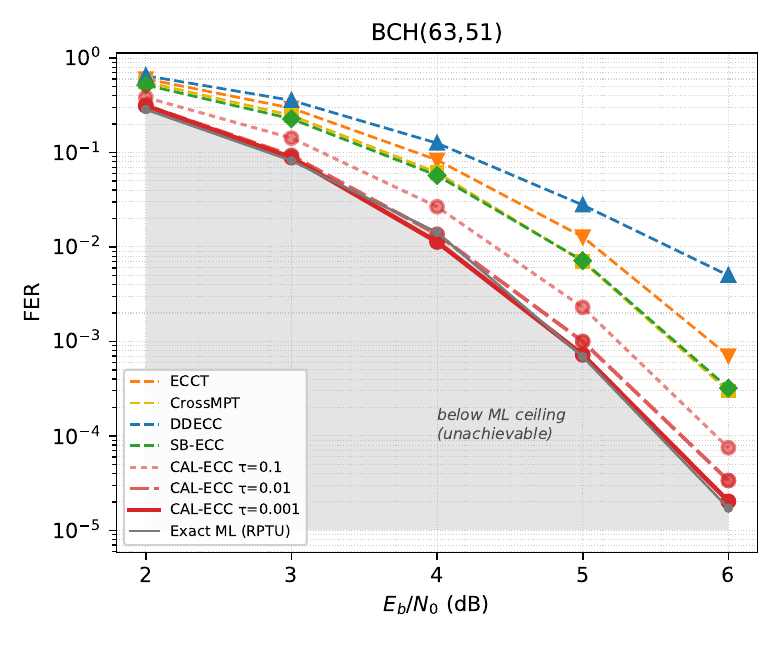}}{BCH$(63,51)^\ast$}\hfill
	\gridfig{\includegraphics[width=\linewidth]{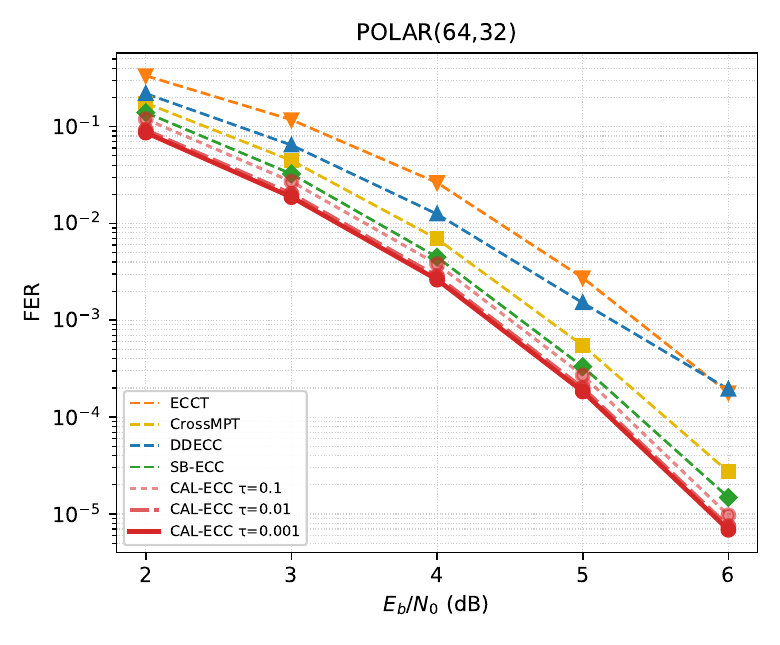}}{Polar$(64,32)$}\hfill
	\gridfig{\includegraphics[width=\linewidth]{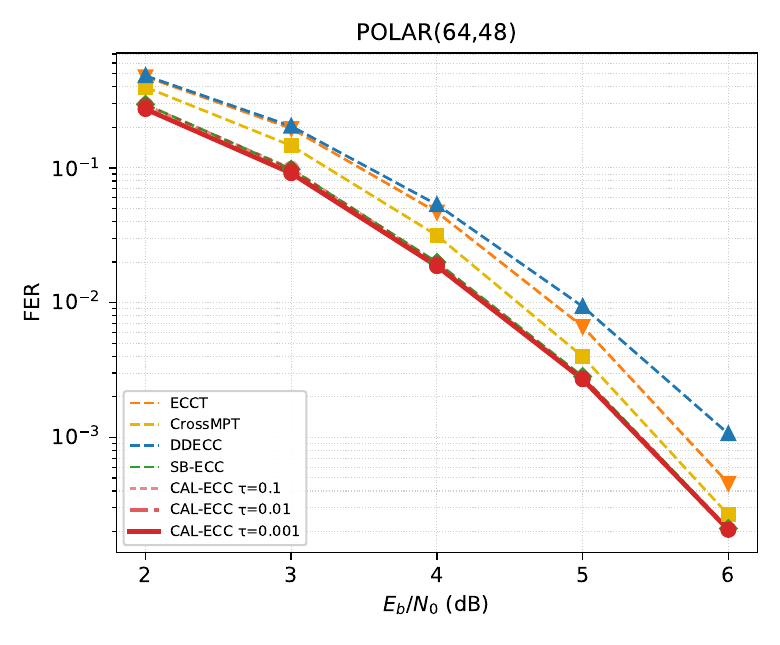}}{Polar$(64,48)$}\\[3pt]
	\gridfig{\includegraphics[width=\linewidth]{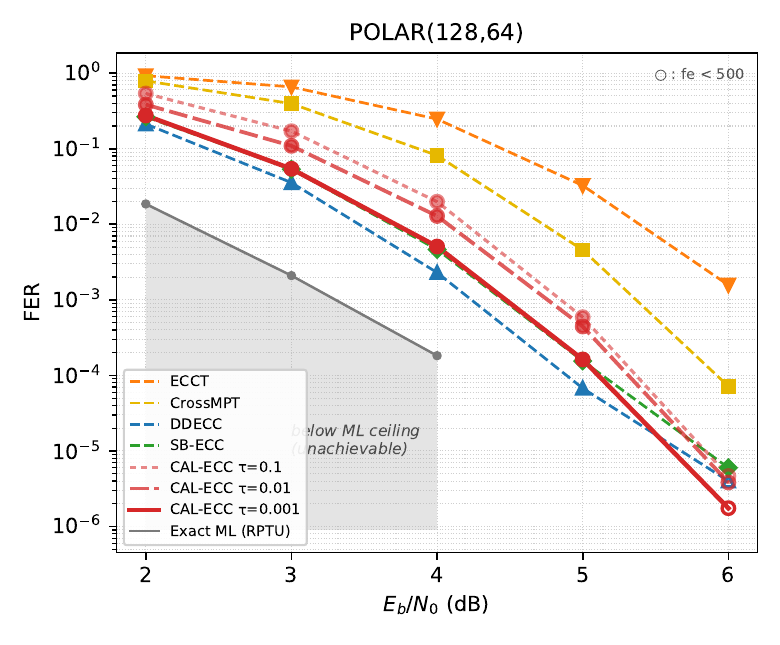}}{Polar$(128,64)^\ast$}\hfill
	\gridfig{\includegraphics[width=\linewidth]{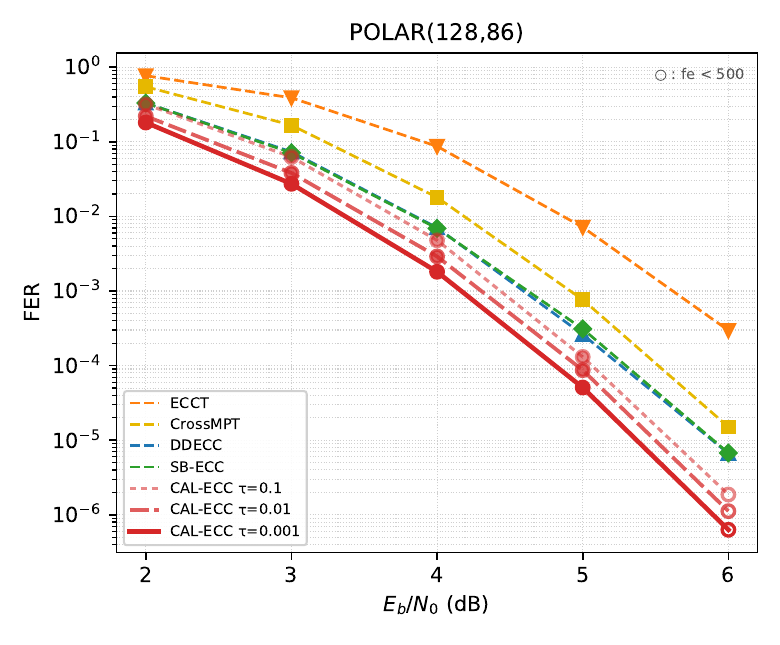}}{Polar$(128,86)$}\hfill
	\gridfig{\includegraphics[width=\linewidth]{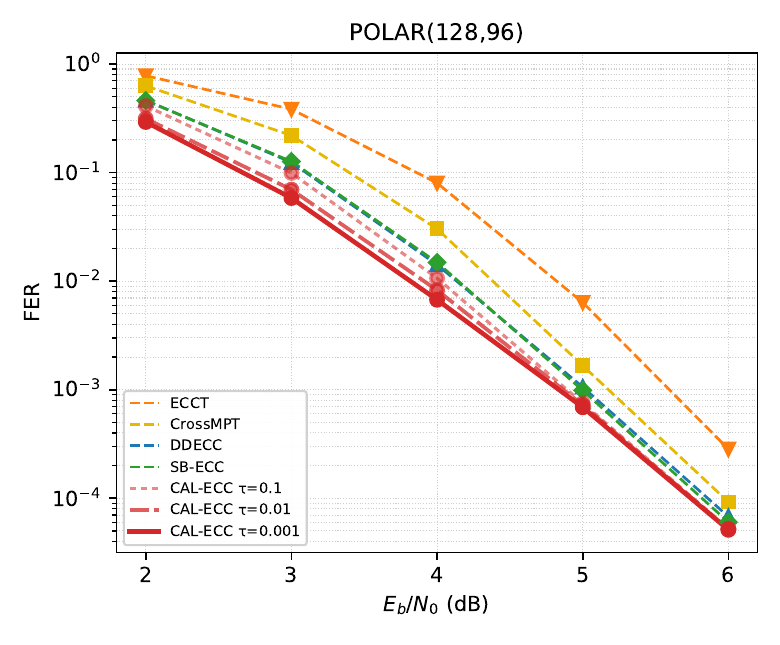}}{Polar$(128,96)$}\\[3pt]
	\gridfig{\includegraphics[width=\linewidth]{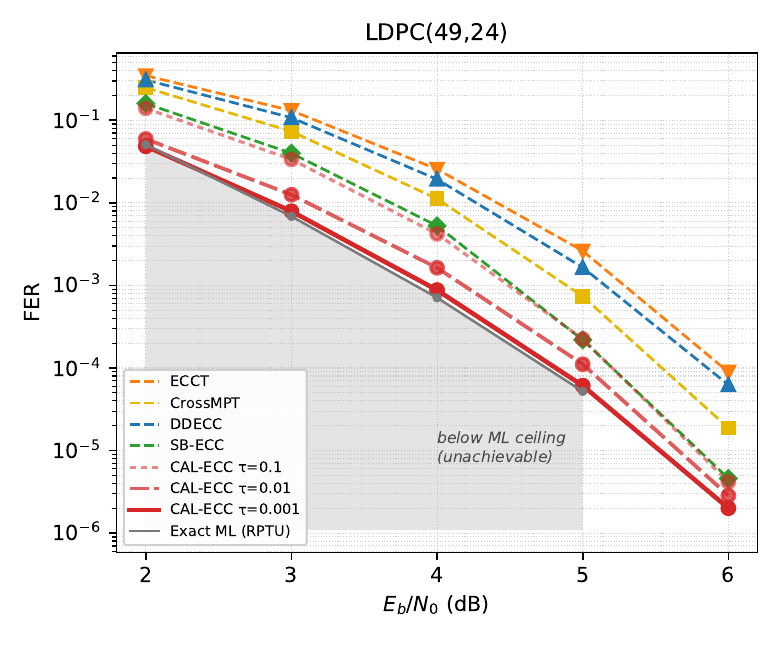}}{LDPC$(49,24)^\ast$}\hfill
	\gridfig{\includegraphics[width=\linewidth]{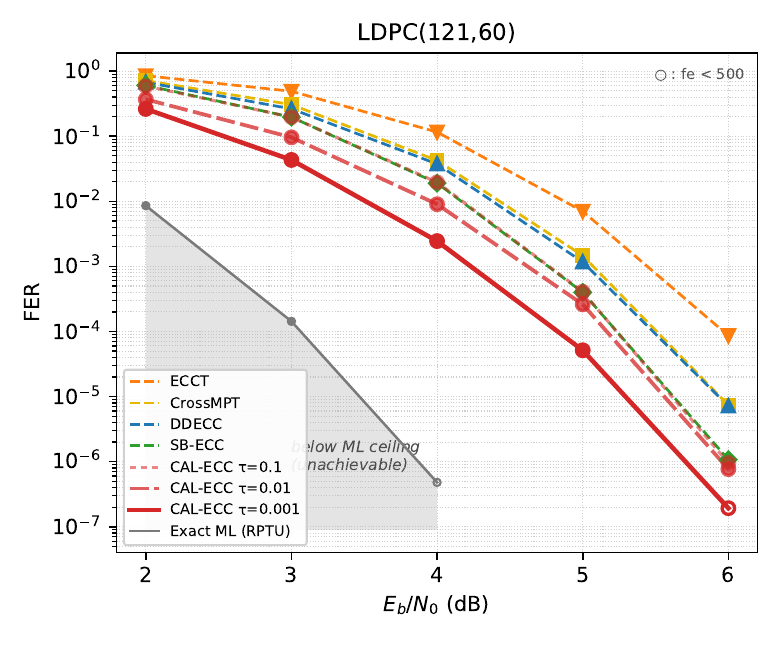}}{LDPC$(121,60)^\ast$}\hfill
	\gridfig{\includegraphics[width=\linewidth]{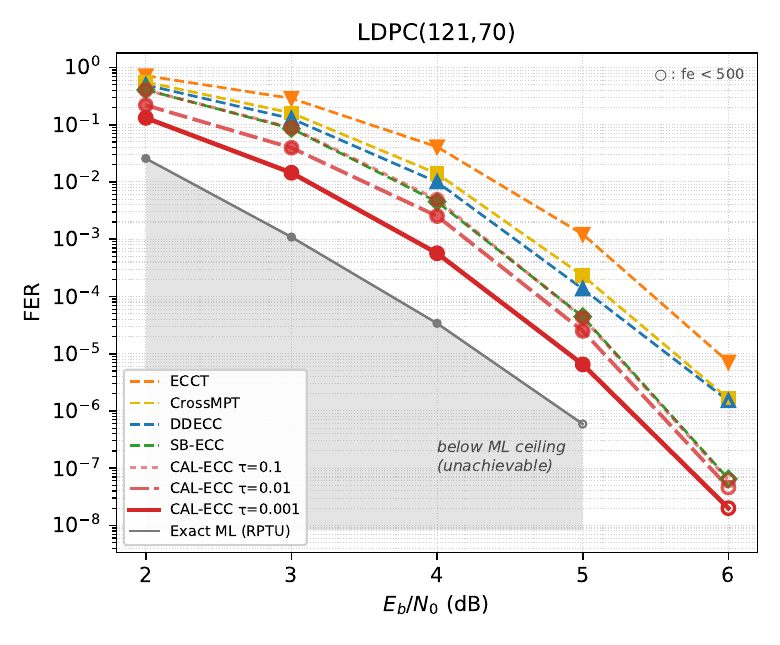}}{LDPC$(121,70)^\ast$}\\[3pt]
	\gridfig{\includegraphics[width=\linewidth]{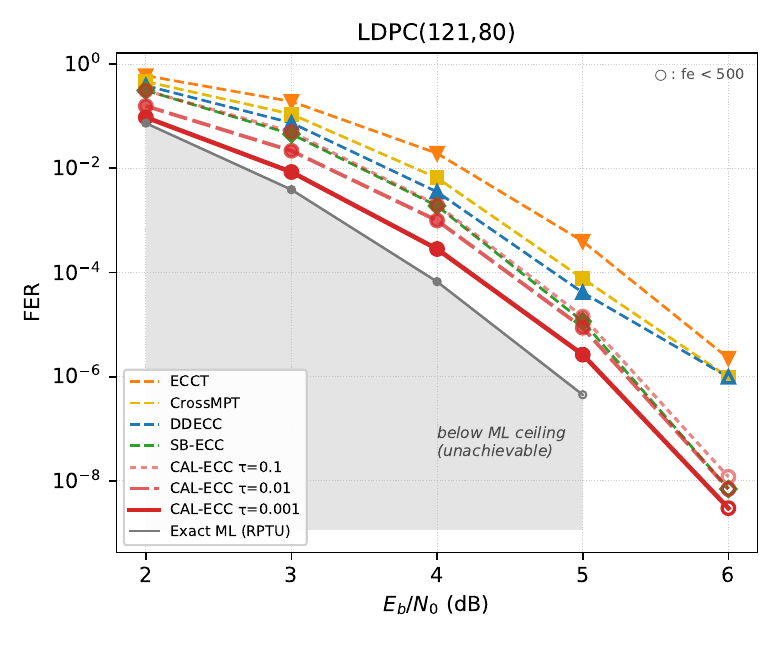}}{LDPC$(121,80)^\ast$}\hfill
	\gridfig{\includegraphics[width=\linewidth]{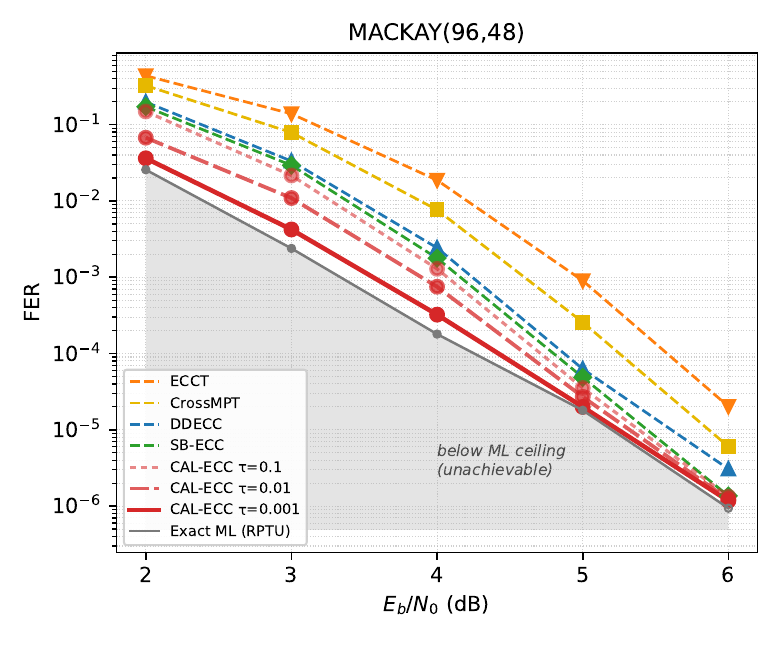}}{MacKay$(96,48)^\ast$}\hfill
	\gridfig{\includegraphics[width=\linewidth]{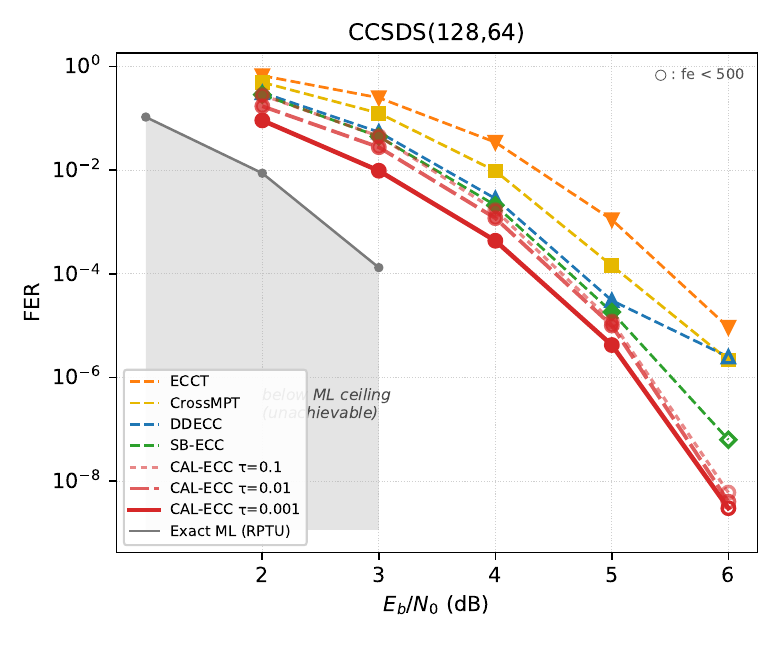}}{CCSDS$(128,64)^\ast$}
	\caption{FER on all fifteen codes, same conventions as
		Figure~\ref{fig:app-ber}. Panels marked $^\ast$ carry the exact ML reference (grey) with the unachievable region shaded. A reference is published for eleven of the fifteen codes~\citep{helmling2025}, and five of those are reproduced in
		Figure~\ref{fig:ml-ceiling}. The reference terminates at $5$ or $6$~dB
		because the integer program becomes infeasible at the frame counts needed
		below that FER.}
	\label{fig:app-fer}
\end{figure}


\end{document}

%% file: math_commands.tex
\usepackage{amsmath,amsfonts,bm}

\def\eqref#1{equation~\ref{#1}}

\def\1{\bm{1}}

\DeclareMathAlphabet{\mathsfit}{\encodingdefault}{\sfdefault}{m}{sl}
\SetMathAlphabet{\mathsfit}{bold}{\encodingdefault}{\sfdefault}{bx}{n}

